\documentclass[]{elsarticle}

\usepackage{lineno,hyperref}
\modulolinenumbers[5]
\let\linenumbers\nolinenumbers\nolinenumbers

\journal{Journal of \LaTeX\ Templates}

\usepackage{amsfonts, amsmath, amsthm, amssymb}
\usepackage{microtype}
\usepackage{multirow}
\usepackage{diagbox}
\usepackage{graphicx}
\usepackage{subfigure}
\usepackage{booktabs} 
\usepackage{algorithm}
\usepackage{wrapfig}
\usepackage{dsfont}
\usepackage{xcolor}

\newtheorem{definition}{Definition}[]
\newtheorem{assumption}{Assumption}[]
\newtheorem{theorem}{Theorem}[]
\newtheorem{prop}{Proposition}[]
\newtheorem{corollary}{Corollary}[]

\newtheorem{lemma}{Lemma}[]
\def\rz{\mathrm{z}}
\def\ry{{\mathrm{y}}}
\def\rx{\mathrm{x}}
\def\ru{\mathrm{u}}

\begin{document}

\begin{frontmatter}

\title{Entropy Estimation via Uniformization\tnoteref{mytitlenote}}

\author{Ziqiao Ao}
\ead{zxa029@bham.ac.uk}
\author{Jinglai Li}
\ead{j.li.10@bham.ac.uk}
\address{School of Mathematics, University of Birmingham, Birmingham B15 2TT, UK}




\begin{abstract}
Entropy estimation is of practical importance in information theory and statistical science. 
Many existing entropy estimators suffer from fast growing estimation bias with respect to dimensionality, 
rendering them unsuitable for high-dimensional problems. 
In this work we propose a transform-based method for high-dimensional entropy estimation, which consists of the following two main ingredients. 
First by modifying the k-NN based entropy estimator, we propose a new estimator which enjoys 
small estimation bias for samples that are close to a uniform distribution. 
Second we  design a normalizing flow based mapping that pushes samples toward a uniform distribution, and the relation between the entropy of the original samples and the transformed ones is also derived. 
As a result the entropy of a given set of samples is estimated by first transforming them toward a uniform distribution and then applying the proposed estimator to the transformed samples. { The performance of the proposed method is 
compared against several existing entropy estimators,
with both mathematical examples and real-world applications.}
\end{abstract}


\end{frontmatter}

\linenumbers

\section{Introduction}
{ Entropy, a fundamental concept  in information theory, 
has found applications in various fields such as physics,  statistics,
signal processing, and machine learning.}
For example, in the statistics and data science contexts,  various applications rely critically on the estimation of entropy, including goodness-of-fit testing \cite{vasicek1976test, goria2005new}, sensitivity analysis \cite{azzi2020sensitivity},  parameter estimation \cite{ranneby1984maximum, wolsztynski2005minimum}, and Bayesian experimental design \cite{sebastiani2000maximum,ao2020approximate}.

In this work we focus on the continuous version of entropy that takes the form, 
\begin{equation}
	H(X)=-\int \log[p_\mathrm{x}(\mathrm{x})]p_\mathrm{x}(\mathrm{x})d\mathrm{x}, \label{eq:entX}
\end{equation}
where $p_{\mathrm{x}}(\mathrm{x})$ is probability density function of a random variable $X$. Despite the rather simple definition, entropy only admits an analytical expression for a limited family of distributions and needs to be evaluated numerically in general. When the distribution of interest is analytically available, in principle its entropy can be estimated by numerical integration schemes  such as the Monte Carlo method. However, in many real-world applications, the distribution of interest is not analytically available, and one  has to {estimate the entropy} from the realizations
drawn from the target distribution, which makes it  difficult or even impossible
to directly compute the entropy via numerical integration.

Entropy estimation has attracted considerable attention from various communities in the last a few decades, 
and numerous methods have been developed to 
directly estimate entropy from realizations. 
In this work we only consider non-parametric approaches which do not assume any parametric model of the target distribution, and 
those methods can be broadly classified into two categories. The first class of {methods,
are known} as the plug-in estimators, which first estimates the underlying probability density, and then compute the integral in Eq.~\eqref{eq:entX} using numerical integration or Monte Carlo (see \cite{beirlant1997nonparametric} for a detailed description). Some examples of density estimation approaches that have been studied for plug-in methods are kernel density estimator \cite{joe1989estimation, hall1993estimation, moon2018ensemble,pichler2022differential}, histogram estimator \cite{gyorfi1987density, hall1993estimation} and field-theoretic approach \cite{chen2018density}. A major limitation of this type of methods is that they rely on an effective density estimation, which is a difficult problem in its own right, especially when the dimensionality of the problem is high. 
A different strategy is to {directly estimate} the entropy from the independent samples of the random variable. Popular methods falling in this category  include the sample-spacing \cite{miller2003new} and the k-nearest neighbors (k-NN) \cite{kozachenko1987sample,kraskov2004estimating} based estimators. The latter is particularly appealing among the existing estimation methods thanks to its theoretical and computational advantages and has been widely used in practical problems. 
Efforts have been constantly devoted to extending and improving 
the k-NN methods, and some recent variants and extensions of 
the methods are \cite{ gao2015efficient,lord2018geometric,berrett2019efficient}. 
It is also worth mentioning that there are many other types of direct entropy estimators available. For example, Ariel and Louzoun \cite{ariel2020estimating} decoupled the target entropy to a sum of the entropy of marginals, which is estimated using one-dimensional methods, and the entropy of copula, which is estimated recursively by splitting the data along statistically dependent dimensions. Kandasamy et al. \cite{kandasamy2015nonparametric} 
suggested a leave-one-out technique for the  von Mises expansion based estimator \cite{fernholz2012mises}.

It is well known that, entropy estimation becomes increasingly more difficult as the dimensionality grows, and such difficulty is mainly due to the \emph{estimation bias}, which decays very slowly with respect to sample size for high-dimensional problems. 
For example in many popular approaches including the k-NN method~\cite{kozachenko1987sample}, the estimation bias decays at the rate 
of $O(N^{-\gamma/d})$ where $N$ is the sample size, $d$ is the dimensionality, and $\gamma$ is a positive constant \cite{krishnamurthy2014nonparametric, kandasamy2015nonparametric,gao2018demystifying, sricharan2013ensemble}. 
As a result, very few, if not none, of the existing entropy estimation methods can effectively handle high-dimensional problems  without 
making strong assumptions about the smoothness of the underlying distribution~\cite{kandasamy2015nonparametric}. 
 Indeed, the well-known minimax bias results (e.g.,~\cite{han2020optimal, birge1995estimation})
indicate that without the strong smoothness assumption \cite{kandasamy2015nonparametric}, the curse of dimensionality is unavoidable. However, efforts can still be made to reduce the difference between the actual estimation bias and the theoretical bound.
%

The main goal of this work is to provide an effective entropy estimation approach which can achieve faster bias decaying rate under mild smoothness assumption,
and thus can effectively deal with high-dimensional problems.
The method presented here consists of two main ingredients. 
We propose two truncated k-NN estimators based on those by \cite{kozachenko1987sample} and \cite{kraskov2004estimating} respectively,
and also provide the bounds of the estimation bias in these estimators.  
Interestingly our theoretical results suggest that the estimators achieve \emph{zero bias} for uniform distributions, while there is no such a result for 
any existing k-NN based estimators,  according to the bias analysis that are available to date~\cite{gao2018demystifying,singh2016finite,biau2015lectures}. This property offers the possibility to significantly improve the performance of entropy estimation by mapping the data points
toward a uniform distribution,
 a procedure that we refer to as \emph{uniformization}. Therefore the second main ingredient of the method is 
 to conduct the uniformization of the data points, with the normalizing flow (NF) technique \cite{rezende2015variational, papamakarios2021normalizing}.
 Simply speaking, NF 
constructs a sequence of invertible and differentiable mappings that transform a simple base distribution such as standard Gaussian
into a more complicated distribution whose density function may not be available. 
Specifically we use the Masked Autoregressive Flow~\cite{papamakarios2017masked}, a NF algorithm originally developed for density estimation,
combined with the probability integral transform, 
to push the original data points towards the uniform distribution.
We then estimate the entropy of the resulting near-uniform data points with the proposed truncated k-NN estimators,  
and derive that of the original ones accordingly (by adding an entropic correction term due to the transformation). 
Therefore, by combining the truncated k-NN estimators and the normalizing flow model, we are able to decode a complex high-dimensional distribution represented by the realizations, and obtain an accurate estimation of its entropy. 

{The rest of the paper is organized as follows. In Section \ref{sec:knn}, we describe the traditional k-NN based methods of entropy estimation and their convergence properties. 
In Section \ref{sec:umbee}, we introduce the truncated k-NN estimators for distributions with compact support,
and then show how to combine these new estimators with the NF-based uniformization procedure
to estimate the entropy of general distributions.
Numerical examples and applications are
presented in Sections \ref{sec:examples} and Section \ref{sec:application} respectively to demonstrate the effectiveness of the proposed methods. 
Finally, in Section \ref{sec:conclusion}, we summarize our findings and 
discuss some  future research directions.
} 

\section{k-NN Based Entropy Estimation} \label{sec:knn}
We provide a brief introduction to two commonly used k-NN based entropy estimators in this section. We 
start with the original k-NN entropy estimator proposed in \cite{kozachenko1987sample}, where the $k$-th nearest neighbor is contained in the smallest possible closed ball. Next, we introduce a popular variant of the k-NN estimator proposed in \cite{kraskov2004estimating}, and this method uses the smallest possible hyper-rectangle to cover at least $k$ points. We finally discuss some theoretical analysis of estimation errors in the estimators.

\subsection{Kozachenko-Leonenko Estimator}
Recall the definition of entropy in Eq.~\eqref{eq:entX}. Given a density estimator $\widehat{p_\mathrm{x}}(\mathrm{x})$ for ${p_\mathrm{x}}(\mathrm{x})$ and a set of $N$ i.i.d. samples $S=\{\mathrm{x}^{(i)}\}_{i=1}^N$ drawn from ${p_\mathrm{x}}(\mathrm{x})$, the entropy of the random variable $X$ can be estimated as follows:
\begin{equation}\label{eq:subest}
	\widehat{H}(X)=-N^{-1}\sum_{i=1}^{N}\log \widehat{p_\mathrm{x}}(\mathrm{x}^{(i)}).
\end{equation}
The Kozachenko-Leonenko (KL) estimator depends on a local uniformity assumption to obtain the estimate $\widehat{p_\mathrm{x}}(\mathrm{x})$. 
For each $x^{(i)}$, one first identifies the $k$-nearest neighbors (in terms of the $p$-norm distance) of it,  
and defines the smallest closed ball covering all these $k$ neighbors as:
\[B(\mathrm{x}^{(i)},\epsilon_{i}/2)=\{\mathrm{x}\in \mathbb{R}^d~ \big|~ \|\mathrm{x}-\mathrm{x}^{(i)}\|_p\leq \epsilon_{i}/2\},\]
where $\epsilon_{i}$ be twice the distance between $\mathrm{x}^{(i)}$ and its $k$-th nearest neighbor among the set $S$.
We shall refer to the closed ball $B(\mathrm{x}^{(i)},\epsilon_{i}/2)$ as a \emph{cell} centered at $\mathrm{x}^{(i)}$,   and let $q_i$ be the mass of the cell $B(\mathrm{x}^{(i)},\epsilon_{i}/2)$ , i.e., 
$$q_i(\epsilon_{i})=\int_{\mathrm{x}\in B(\mathrm{x}^{(i)},\epsilon_{i}/2)}p_\mathrm{x}(\rx)d\mathrm{x}.$$ 
It can be derived that the expectation value of $\log q_i$ over $\epsilon_{i}$ is given by
\begin{equation}
	\begin{aligned}\label{eq:log_p}
		\mathbb{E}(\log q_i)=\psi(k)-\psi(N),
	\end{aligned}
\end{equation}
where $\psi(x) = \frac{\Gamma'(x)}{\Gamma(x)}$ with $\Gamma(x)$ being the Gamma function~\cite{kraskov2004estimating}.
KL estimator then assumes that the density is constant in $B(x^{(i)},\epsilon_{i})$, which gives
\begin{equation}\label{eq:localconst}
	q_i(\epsilon_{i})\approx c_d \epsilon_{i}^d p_\mathrm{x}(\mathrm{x}^{(i)}),
\end{equation}
where $d$ is the dimension of $X$ and $$c_d=\Gamma(1+\frac{1}{p})^d/\Gamma(1+\frac{d}{p}),$$ is the volume of the $d$-dimensional unit  ball
with respect to $p$-norm.
Combining \eqref{eq:log_p} and \eqref{eq:localconst} one can get an estimate of the log-density at each sample point,
\begin{equation}
	\log \widehat{p_\mathrm{x}}(\mathrm{x}^{(i)}) = \psi(k)-\psi(N)-\log c_d - d\log\epsilon_{i}.
\end{equation}
Plugging the above estimates for $i=1,...,N$ into \eqref{eq:subest} yields the KL estimator:
\begin{equation}
	\widehat{H}_\mathrm{KL}(X)= -\psi(k)+\psi(N)+\log c_d + \frac{d}{N}\sum_{i=1}^{N}\log\epsilon_{i}.
\end{equation}

\subsection{KSG Estimator}\label{sec:ksg_est}
As is mentioned earlier, the Kraskov-St{\"o}gbauer-Grassberger (KSG) estimator is an important variant of $\hat{H}_\mathrm{KL}$. 
Unlike KL estimator that is based on closed balls, KSG estimator uses hyper-rectangles to form the cells at each data point.  Namely one chooses the $\infty$-norm as the distance metric (i.e $p=\infty$), and as a result the cell $B(x^{(i)},\epsilon_{i}/2)$ becomes 
a hyper-cube with side length $\epsilon_{i}$. 
Next, we  allow the hyper-cube to become a hyper-rectangle: i.e., the cells
admit different side lengths along different dimensions.
Specifically, for $j=1,...,d$, we define
$\epsilon_{i,j}$  to be twice of the distance between $x^{(i)}$ and its $k$-th nearest neighbor along dimension $j$,
and the cell centered at $\mathrm{x}^{(i)}$ covering its $k$-nearest neighbors becomes
\begin{equation} 
	\begin{split}
	B(\mathrm{x}^{(i)},\epsilon_{i,1:d}/2) = \{ \mathrm{x}=(\mathrm{x}_1,...,\mathrm{x}_d)\,|\, |\mathrm{x}_j-\mathrm{x}_j^{(i)}|\leq \epsilon_{i,j}/2,\,\, \\ \mathrm{for}\,\,j=1,...,d\},
	\end{split}
\end{equation}
where $\epsilon_{i,1:d} = (\epsilon_{i,1},...,\epsilon_{i,d})$. 
This change leads to 
a different formula for computing
the mass of the cell $B(\mathrm{x}^{(i)},\epsilon_{i,1:d}/2)$, 
\begin{equation} \label{eq:log_q} 
	\mathbb{E}(\log q_i)\approx \psi(k)-\frac{d-1}{k}-\psi(N).
\end{equation}
It is worth noting that the equality in Eq.~\eqref{eq:log_p} is replaced by approximate equality in Eq.~\eqref{eq:log_q}, because a uniform density within the rectangle has to be assumed to obtain Eq.~\eqref{eq:log_q} (see Lemma 2 in \ref{sec:lemmas} for details). 
Using a similar local assumption as Eq.~\eqref{eq:localconst}, the KSG estimator is derived as, 
\begin{equation}
	\widehat{H}_\mathrm{KSG}(X)= -\psi(k)+\psi(N)+\frac{d-1}{k} + \frac{1}{N}\sum_{i=1}^{N}\sum_{j=1}^{d}\log\epsilon_{i,j}.
\end{equation}
We note that the KSG method was actually developed in the context of estimating mutual information~\cite{kraskov2004estimating},
and has been reported to outperform the KL estimator in a wide range of problems~\cite{gao2018demystifying}.
As has been shown above, it is straightforward to extend it to entropy estimation, and our numerical experiments also
suggest that it has competitive performance as an entropy estimator, which will be demonstrated in Section~\ref{sec:examples}. 


\subsection{Convergence Analysis}
Another important issue is to analyze the estimation errors in these entropy estimators and especially how they behave as the sample size increases.  
In most of the k-NN based estimators including the two mentioned above, the variance is generally well controlled, decaying at a rate of 
$O(N^{-1})$ with $N$ being the sample size, while the main issue lies on the estimation bias. 
In fact, the bias of estimator $\widehat{H}_\mathrm{KL}$ has been well studied, but that of $\widehat{H}_\mathrm{KSG}$ receives very little attention. Previous results related to the former are listed as follows. The original \cite{kozachenko1987sample} paper established the asymptotic unbiasedness for $k=1$ while \cite{singh2003nearest} obtained the same result for general $k$. For distributions with unbounded support, \cite{tsybakov1996root} proved that the bias bound decays at a rate of $O(\frac{1}{\sqrt{N}})$ for $d=1$. \cite{gao2018demystifying} generalized it to higher dimensions, obtaining a bias bound of $O(N^{-\frac{1}{d}})$ up to polylogarithmic factors. For distributions compactly supported, usually densities satisfying the $\beta$-H{\"o}lder condition are considered. \cite{biau2015lectures} gave a quick-and-dirty
upper bound of bias, $O(N^{-\beta})$, for a simple class of univariate densities supported on $[0,1]$ and bounded away from zero. \cite{singh2016finite}
proved the bias is around $O(N^{-\frac{\beta}{d}})$ ($\beta \in (0,2]$) for general $d$ with some additional conditions on the boundary of support. 
We reinstate that all these works obtained a variance bound of $O(N^{-1})$.

It should be noted that the bias bounds given by previous studies typically depend on some properties of target densities, such as smoothness parameter and Hessian matrix, providing insights that these estimators perform well on certain distributions that satisfy certain conditions. 
This motivates the idea that one can transform the given data points toward a desired distribution for a more accurate entropy estimation,
which is detailed in next section. 

\section{Uniformizing Mapping Based Entropy Estimation} \label{sec:umbee}
In this section, we shall present an entropy estimation approach that is based on normalizing flow. 
As is mentioned earlier, it consists of two main ingredients: a truncated version of the k-NN entropy estimators, 
and a transformation that can map data points toward a uniform distribution. 

\subsection{Truncated KL/KSG Estimators}\label{sec:truncation}

\begin{figure}
	\vspace{-0.0in}
	\begin{center}
		\includegraphics[width=0.5\textwidth]{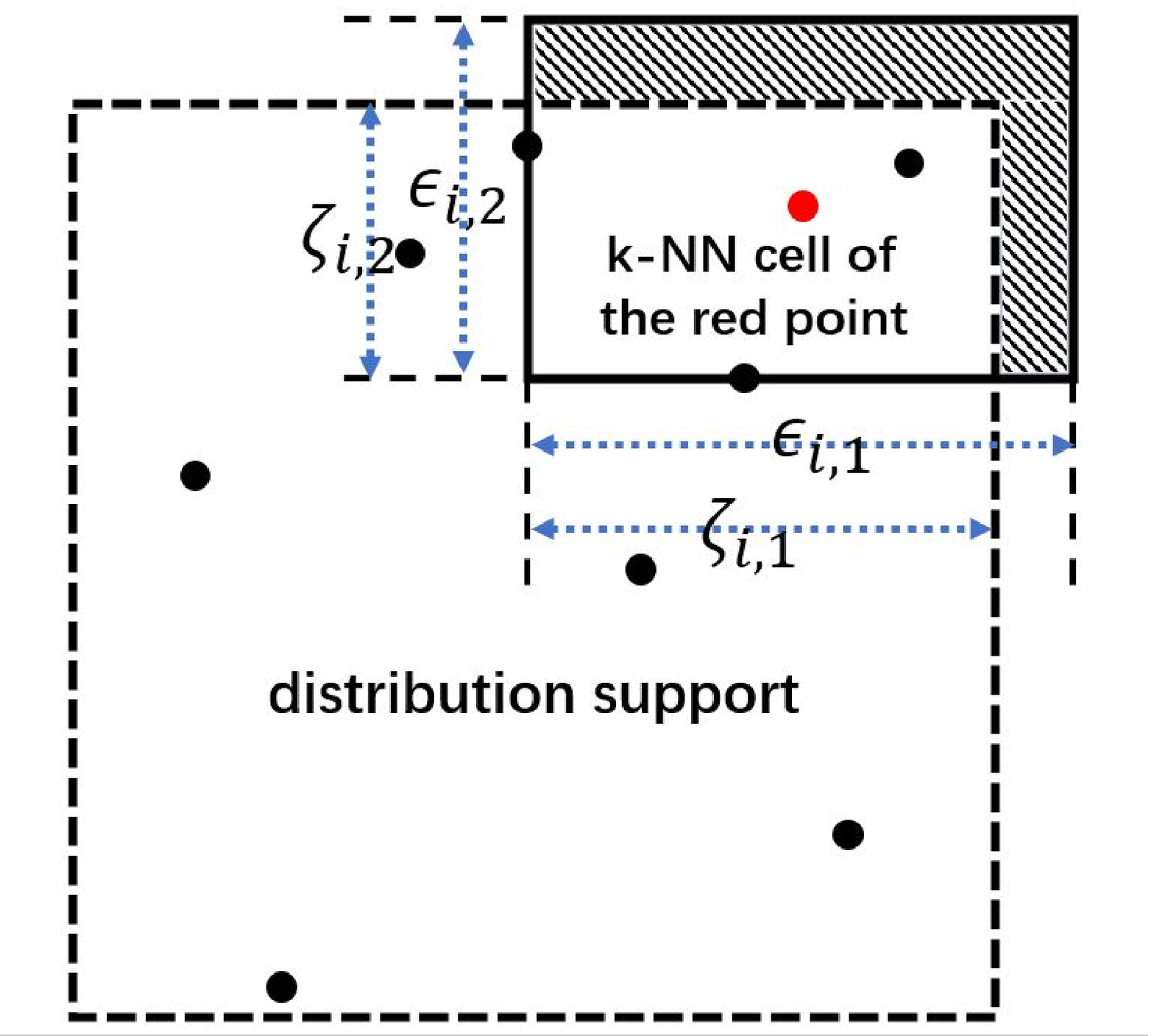}
	\end{center}
	\vspace{-0.0in}
	\caption{The schematic illustration of the truncated estimator.
		The shaded area is that removed from the k-NN cell.}
	\label{fig:tksg_ex}
	\vspace{-0.0in}
\end{figure}

For compactly supported distributions, a significant source of bias 
comes from the boundary of the support, where the $k$-NN cells are constructed including areas outside of the support of 
the distribution density~\cite{singh2016finite}. 
Intuitively speaking, incorrectly including such areas results in an underestimate of the densities, leading to bias in the estimator. 
We thus propose a method to reduce the estimation bias by excluding the areas outside of the distribution support, 
and remarkably the resulting estimator enjoy certain convergence properties which enable us to design the NF based estimation approach. 
The only additional requirement for using these estimators is that the bound of support of density should be specified.
Without loss of generality, we suppose the target density is supported on the unit cube $\mathcal{Q}:=[0,1]^d$ in $\mathbb{R}^d$. 
The procedure of our method is as follows: we first determine all the cells using either KL or KSG, 
then examine whether each k-NN cell covers area out of the distribution support, and
if so, truncate the cell at the boundary to exclude such area (see Fig.~\ref{fig:tksg_ex} for a schematic illustration).
Mathematically the truncated KL (tKL) estimator (with $\infty$-norm), is given by
\begin{equation}\label{eq:tkl}
	\widehat{H}_\mathrm{tKL}(X)= -\psi(k)+\psi(N) + \frac{1}{N}\sum_{i=1}^{N}\sum_{j=1}^{d}\log\xi_{i,j},
\end{equation}
where $$\xi_{i,j}=\min\{\mathrm{x}^{(i)}_j+\epsilon_{i}/2,1\}-\max\{\mathrm{x}^{(i)}_j-\epsilon_{i}/2,0\};$$
and the truncated KSG (tKSG) esitmator is given by
\begin{equation}\label{eq:tksg}
	\begin{aligned}
	\widehat{H}_\mathrm{tKSG}(X)= -\psi(k)+\psi(N)&+(d-1)/k + \frac{1}{N}\sum_{i=1}^{N}\sum_{j=1}^{d}\log\zeta_{i,j},
	\end{aligned}
\end{equation}
where $$\zeta_{i,j}=\min\{\mathrm{x}^{(i)}_j+\epsilon_{i,j}/2,1\}-\max\{\mathrm{x}^{(i)}_j-\epsilon_{i,j}/2,0\}.$$

Next we shall theoretically analyze the bias of the truncated estimators. 
Our analysis relies on some assumptions on the density function $p_\mathrm{x}$, which are summarized as below: 
\begin{assumption}\label{assumption11}
	The distribution $p_\rx$ satisfies: 
	\begin{enumerate}
		\item[(a)] $p_\mathrm{x}$ is continuous and supported on $\mathcal{Q}$;
		\item[(b)] $p_\mathrm{x}$ is bounded away from 0, i.e., $C_1 = \inf\limits_{\mathrm{x}\in\mathcal{Q}}p_\mathrm{x}(\mathrm{x})>0$;
		\item[(c)] The gradient of $p_\rx$ is uniformly bounded on ${\mathcal{Q}^o}$, i.e., $C_2 = \sup\limits_{\mathrm{x}\in\mathcal{Q}^o}||\triangledown p_\mathrm{x}(\mathrm{x})||_1<\infty$.
	\end{enumerate}
\end{assumption}

First we consider the bias of  estimator $\widehat{H}_\mathrm{tKL}$ and the following  theorem states that,  the bias
in  $\widehat{H}_\mathrm{tKL}$ is bounded and vanishes at the rate of $O(N^{-\frac{1}{d}})$.
\begin{theorem}\label{thm1}
	Under Assumption~\ref{assumption11} and for any finite $k$ and $d$, the bias of the truncated KL estimator is bounded by 
	$$\big|\mathbb{E}[\widehat{H}_\mathrm{tKL}(X)]-H(X)\big|\leq \frac{C_2}{C_1^{1+1/d}}\big(\frac{k}{ N}\big)^\frac{1}{d}.$$
	The variance of the truncated KL estimator is bounded by
	$$\mathrm{Var}[\widehat{H}_\mathrm{tKL}(X)]\leq C\frac{1}{N},$$
	for some $C>0$.
\end{theorem}
\begin{proof}
    We provide a skeleton proof here, where the complete proof including 
    the notations 
    is detailed in 
    \ref{biastkl} and \ref{vartkl}. 

    \textit{Proof of the bias bound for the truncated KL estimator proceeds as follows.}
    \begin{enumerate}
        \item Show that 
        \begin{equation}\label{eq:thm1eq1}
		\begin{aligned}
		\mathbb{E}[\widehat{H}_{tKL}(X)]=-\mathbb{E}\big[\log\frac{P(\overline{B}(\mathrm{x};\epsilon_k/2))}{\mu(\overline{B}(\mathrm{x};\epsilon_k/2))}\big].
		\end{aligned}
		\end{equation}
		
		\item Bound the following difference by
		\begin{equation}\label{eq:thm1eq2}
		\begin{aligned}
		\bigg|\log p(\mathrm{x})-\log\frac{P(\overline{B}(\mathrm{x};\epsilon_k/2))}{\mu(\overline{B}(\mathrm{x};\epsilon_k/2))}\bigg|
		\leq\frac{C_2}{2C_1}\epsilon_k.
		\end{aligned}
		\end{equation}
		
		\item Note that $H(X)= -\mathbb{E}(\log p(x))$,
  and using Eq.~\eqref{eq:thm1eq1}, Eq.~\eqref{eq:thm1eq2} and the upper bound of $\mathbb{E}(\epsilon_k)$ obtained from Lemma~\ref{lemma4},
  we can derive that the bias $\mathbb{E}[\widehat{H}_{tKL}(X)]$ is bounded by
		\begin{equation}
		\begin{aligned}
		\big|\mathbb{E}[\widehat{H}_{tKL}(X)]-H(X)\big|
		\leq\frac{C_2}{C_1^{1+1/d}}\big(\frac{k}{ N}\big)^\frac{1}{d}.
		\end{aligned}
		\end{equation}
    \end{enumerate}
    
    \textit{Proof of the variance bound for the truncated KL estimator proceeds as follows.}
    \begin{enumerate}
        \item Let $\alpha_i=\sum_{j=1}^{d}\log\xi_{i,j}$ and let $\alpha^*_i$ (for $i=2,...,N$) be the estimators with sample $\mathrm{x}^{(1)}$  removed. Then, by the Efron-Stein inequality \cite{efron1981jackknife},
		\begin{equation}
		\begin{aligned}
		\mathrm{Var}[\widehat{H}_{tKL}(X)]=\mathrm{Var}\bigg[\frac{1}{N}\sum_{i=1}^{N}\alpha_i\bigg]\leq
		2 N \mathbb{E}\bigg[\bigg(\frac{1}{N}\sum_{i=1}^{N}\alpha_i-\frac{1}{N}\sum_{i=2}^{N}\alpha^*_i\bigg)^2\bigg].
		\end{aligned}
		\end{equation}
		
		\item Let $\mathds{1}_{E_i}$ be the indicator function of the event $E_i=\{\epsilon_k(\mathrm{x}^{(1)})\neq\epsilon^*_k(\mathrm{x}^{(1)})\}$, where $\epsilon^*_k(\mathrm{x}^{(1)})$ is twice the $k$-NN distance of $\mathrm{x}^{(1)}$ when $\alpha^*_i$ are used. Then we show that
		\begin{equation}
		\begin{aligned}
		N^2\bigg(\frac{1}{N}\sum_{i=1}^{N}\alpha_i-\frac{1}{N}\sum_{i=2}^{N}\alpha^*_i\bigg)^2\leq (1+C_{k,d})\bigg(\alpha_1^2+2\sum_{i=2}^{N}\mathds{1}_{E_i}(\alpha_i^2+\alpha_i^{*2})\bigg),
		\end{aligned}
		\end{equation}
		where $C_{k,d}$ is a constant. 
		\item Since $\alpha_i$ and $\alpha_i^*$ are identically distributed, we  only need to 
  derive the upper bounds of the following three expectations:
			$\mathbb{E}[\alpha_1^2]$, 
			$(N-1)\mathbb{E}[\mathds{1}_{E_2}\alpha_2^2]$ 
and			$(N-1)\mathbb{E}[\mathds{1}_{E_2}\alpha_2^{*2}]$.
		
		\item Finally we obtain the bound of the variance of $\widehat{H}_{tKL}(X)$
		\begin{equation}
		\mathrm{Var}[\widehat{H}_{tKL}(X)]\leq C \frac{1}{N},
		\end{equation}
		for some $C>0$.
    \end{enumerate}
\end{proof}
Note that $C_2=0$ when $p_\mathrm{x}$ is uniform on $\mathcal{Q}$, and the following corollary follows directly: 
\begin{corollary}\label{cor1}
	Under the assumption in Theorem~\ref{thm1}, if $X$ is uniformly distributed on $\mathcal{Q}$, then the truncated KL estimator is unbiased.
\end{corollary}
This corollary is the theoretical foundation of the proposed method,
as it suggests that if one can transform the data points into a uniform distribution, the tKL method can yield an unbiased estimate. 
In reality, it is usually impossible to map the data point exactly into a uniform distribution to achieve the unbiased estimate. 
To this end, Theorem~\ref{thm1} suggests that, as long as the transformed samples are close to a uniform distribution in 
the sense that $C_2$ is small, the transformation can still significantly reduce the bias.
Since the main contribution of the mean-square estimation error comes from the bias (as the variance decays at the rate of $O(N^{-1})$), 
reducing the bias therefore leads much more accurate estimation of the entropy.

We next consider the bias of the tKSG estimator. 
The second theorem shows that the expectation of $\widehat{H}_\mathrm{tKSG}$ has the same limiting behavior up to a polylogarithmic factor in $N$. \begin{theorem}\label{thm2}
	Under Assumption~\ref{assumption11} and for any finite $k$ and $d$, the bias of the truncated KSG estimator is bounded by 
	$$\big|\mathbb{E}[\widehat{H}_\mathrm{tKSG}(X)]-H(X)\big|\leq C \frac{(\log N)^{k+2}}{C_1^{k+1} N^{\frac{1}{d}}}$$
	for some $C>0$. The variance of the truncated KSG estimator is bounded by 
	$$\mathrm{Var}[\widehat{H}_\mathrm{tKSG}(X)]\leq C’\frac{(\log N)^{k+2}}{N},$$
	for some $C’>0$.
\end{theorem}
\begin{proof}
    Again, we  only provide a skeleton proof here, with the complete details given in  \ref{biasksg} and \ref{varksg}.
    
    \textit{Proof of the bias bound for the truncated KSG estimator proceeds as follows.}
    
    \begin{enumerate}
        \item Suppose that $\widetilde{P}$, $ \widetilde{p}$, and $\widetilde{q}_{\epsilon_k^{\mathrm{x}_1},...,\epsilon_k^{\mathrm{x}_d}}(\mathrm{x})$ are defined as in Lemma~\ref{lemma2} with $l=p(\mathrm{x})^{-\frac{1}{d}}$, and by Lemma~\ref{lemma2} and the fact that $\sum_{j=1}^{d}\log\zeta_{i,j}$ are identically distributed, we have
		\begin{equation}
		\begin{aligned}
		\mathbb{E}[\widehat{H}_{tKSG}(X)]=\underset{\mathrm{x}\sim p}{\mathbb{E}}\underset{P}{\mathbb{E}}\big[\log{\zeta_k^{\mathrm{x}_1}\cdots \zeta_k^{\mathrm{x}_d}}\big]
		-
		\underset{\mathrm{x}\sim p}{\mathbb{E}}\underset{\widetilde{P}}{\mathbb{E}}\big[\log\big({p(\mathrm{x})\epsilon_k^{\mathrm{x}_1}\cdots\epsilon_k^{\mathrm{x}_d}}\big)\big].
		\end{aligned}
		\end{equation}
		
		\item We separate the $d$-dimensional unit cube $\mathcal{Q}$ into two subsets, $\mathcal{Q}=\mathcal{Q}_1+\mathcal{Q}_2$, where $\mathcal{Q}_1:= [\frac{a_N}{2}, 1-\frac{a_N}{2}]^d$, $a_N=\big(\frac{2k\log N}{C_1N}\big)^{\frac{1}{d}}$, and $\mathcal{Q}_2=\mathcal{Q}-\mathcal{Q}_1$. 
		
		\item Note that $H(X)=-\mathbb{E}(\log p(x))$, and we can then decompose the bias into three terms according to the above separation of unit cube: 
		\begin{equation}
		\begin{aligned}
		&\big|\mathbb{E}[\widehat{H}_{tKSG}(X)]-H(X)\big|\\
		=&\bigg|\underset{\mathrm{x}\sim p}{\mathbb{E}}\underset{P}{\mathbb{E}}\big[\log\big({\zeta_k^{\mathrm{x}_1}\cdots \zeta_k^{\mathrm{x}_d}}\big)\big]
		-
		\underset{\mathrm{x}\sim p}{\mathbb{E}}\underset{\widetilde{P}}{\mathbb{E}}\big[\log\big({\epsilon_k^{\mathrm{x}_1}\cdots\epsilon_k^{\mathrm{x}_d}}\big)\big]\bigg|\\
		\leq&I_1+I_2+I_3,
		\end{aligned}
		\end{equation}
		with
		\begin{equation}
		\begin{aligned}
		I_1&=\bigg|\underset{\mathrm{x}\in \mathcal{Q}_2}{\mathbb{E}}\underset{P:\epsilon_k< a_N}{\mathbb{E}}\big[\log\big({\zeta_k^{\mathrm{x}_1}\cdots \zeta_k^{\mathrm{x}_d}}\big)\big]\bigg| 
		+
		\bigg|\underset{\mathrm{x}\in \mathcal{Q}_2}{\mathbb{E}}\underset{\widetilde{P}:\epsilon_k< a_N}{\mathbb{E}}\big[\log\big({\epsilon_k^{\mathrm{x}_1}\cdots\epsilon_k^{\mathrm{x}_d}}\big)\big]\bigg|,\\
		I_2&=\bigg|\underset{\mathrm{x}\in \mathcal{Q}_1}{\mathbb{E}}\underset{P:\epsilon_k<a_N}{\mathbb{E}}\big[\log\big({\zeta_k^{\mathrm{x}_1}\cdots \zeta_k^{\mathrm{x}_d}}\big)\big]
		-
		\underset{\mathrm{x}\in \mathcal{Q}_1}{\mathbb{E}}\underset{\widetilde{P}:\epsilon_k<a_N}{\mathbb{E}}\big[\log\big({\epsilon_k^{\mathrm{x}_1}\cdots\epsilon_k^{\mathrm{x}_d}}\big)\big]\bigg|,\\
		I_3&=\bigg|\underset{\mathrm{x}\in \mathcal{Q}}{\mathbb{E}}\underset{P:\epsilon_k\geq a_N}{\mathbb{E}}\big[\log\big({\zeta_k^{\mathrm{x}_1}\cdots \zeta_k^{\mathrm{x}_d}}\big)\big]
		\bigg| 
		+
		\bigg|
		\underset{\mathrm{x}\in \mathcal{Q}}{\mathbb{E}}\underset{\widetilde{P}:\epsilon_k\geq a_N}{\mathbb{E}}\big[\log\big({\epsilon_k^{\mathrm{x}_1}\cdots\epsilon_k^{\mathrm{x}_d}}\big)\big]\bigg|,
		\end{aligned}
		\end{equation}
		where $\underset{P:\epsilon_k<a_N}{\mathbb{E}}$ means taking expectation under the probability measure $P$ over $\epsilon_k^{\mathrm{x}_j}<a_N, j=1,...,d$.
		
		\item Finally, by bounding the three terms separately, we obtain 
		\begin{equation}
		    \big|\mathbb{E}[\widehat{H}_\mathrm{tKSG}(X)]-H(X)\big|\leq C \frac{(\log N)^{k+2}}{C_1^{k+1} N^{\frac{1}{d}}},
		\end{equation}
		for some $C>0$.
    \end{enumerate}
    
    \textit{Proof of variance bound for the truncated KSG estimator proceeds as fllows.}
    \begin{enumerate}
        \item 
        Let $\beta_i=\sum_{j=1}^{d}\log\zeta_{i,j}$, and define $\beta^*_i$ (for $ i=2,...,N$) to be the estimators with sample $\mathrm{x}^{(1)}$  removed.
        Next we show that
			$(N-1)\mathbb{E}[\mathds{1}_{E_2}\beta_2^2]$
			and $(N-1)\mathbb{E}[\mathds{1}_{E_2}\beta_2^{*2}]$ are of the same order as $\mathbb{E}[\beta_1^2]$. 
   As such we only need to prove that $\mathbb{E}[\beta_1^2]=O({(\log N)^{k+2}})$, which is done in Steps 2 and 3. 
        
        \item Separate $\mathbb{E}\big[\beta_1^2\big]$ into two parts,
		\begin{equation}
		\begin{aligned}
		\mathbb{E}\big[\beta_1^2\big]=\underset{\mathrm{x}\in \mathcal{Q}}{\mathbb{E}}\underset{P:\epsilon_k< a_N}{\mathbb{E}}\big[\beta_1^2\big]
		+
		\underset{\mathrm{x}\in \mathcal{Q}}{\mathbb{E}}\underset{P:\epsilon_k\geq a_N}{\mathbb{E}}\big[\beta_1^2\big],
		\end{aligned}
		\end{equation}
		where $a_N=\big(\frac{2k\log N}{C_1N}\big)^{\frac{1}{d}}$.
		
		\item By bounding the two parts separately, we obtain 
	    the bound of the expectation of $\beta_1^2$
		\begin{equation}
		\mathbb{E}[\beta_1^2]\leq C_9 (\log N)^{k+2},
		\end{equation}
		for some $C_9>0$. 
		
		\item With the above bound, we can obtain the bound of the variance of $\widehat{H}_{tKSG}(X)$
		\begin{equation}
		\mathrm{Var}[\widehat{H}_{tKSG}(X)]\leq C'\frac{(\log N)^{k+2}}{N},
		\end{equation}
		for some $C'>0$.
		
    \end{enumerate}
    
\end{proof}
As one can see from Theorem~\ref{thm2},
while  the uniform distribution leads to zero bias for $\widehat{H}_\mathrm{tKL}$, we can not obtain the same result for $\widehat{H}_\mathrm{tKSG}$,
which means no theoretical justification for mapping the data points toward a uniform distribution for this estimator. 
That said, the tKSG estimator and Theorem~\ref{thm2} are still useful, and the reason for that is two-fold.  
First as is mentioned earlier,  no existing result on the bound of bias is available for the KSG estimator to 
the best of our knowledge,
and to this end our analysis on tKSG is the first known bias bound for this type of estimators, and
may provide useful information for understanding the convergence property of them.  
More importantly, our numerical experiments demonstrate that mapping the data points toward a uniform distribution does significantly improve 
the performance of tKSG as well. In fact, we have found that tKSG can achieve the same or slightly better results than tKL on the transformed samples in  our test cases.

\subsection{Estimating Entropy via Transformation}
As is mentioned earlier, based on the interesting convergence properties of the truncated estimators in particularly tKL, we want to 
estimate the entropy of a given set of samples by mapping them toward a uniform distribution. To implement this idea, an essential question to ask is that,
how the entropy of the transformed samples relates to that of the original ones. 
Proposition~\ref{prop1} provides an answer to this question. 
\begin{prop}[\cite{ihara1993information}]\label{prop1}
	Let $f$ be a mapping: $\mathcal{R}^d\rightarrow\mathcal{R}^d$, $X$ be random variable defined on $\mathcal{R}^d$ following distribution $p_\rx$,
	and $Z=f(X)$. 
	If $f$ is bijective and differentiable, we have
	\begin{equation}\label{eq:connect}
		H(X) = H(Z)+\int p_\mathrm{z}(\mathrm{z})\log\bigg|\det\frac{\partial f^{-1}(\mathrm{z})}{\partial \mathrm{z}}\bigg|d\mathrm{z},
	\end{equation}
	where $p_\rz(\rz)$ is the distribution of $Z$. 
\end{prop}
Therefore given a data set $S=\{\mathrm{x}^{(i)})\}_{i=1}^N$ and a mapping $Z=f(X)$, from Eq.~\eqref{eq:connect} we can construct an entropy estimator of $X$ as,
\begin{equation}\label{eq:norm1}
	\widehat{H}(X) = \widehat{H}(Z)+\frac{1}{n}\sum_{i=1}^{n}\log\bigg|\det\frac{\partial f^{-1}(\mathrm{z}^{(i)})}{\partial \mathrm{z}}\bigg|,
\end{equation}
where $\widehat{H}(Z)$ is an entropy estimator of $Z$ (either tKL or tKSG) based on the transformed samples $S_Z=\{\mathrm{z}^{(i)}=f(\mathrm{x}^{(i)})\}_{i=1}^n$. 

We refer to such a mapping $f(\cdot)$ as a uniformizing mapping (UM) and the resulting methods as UM based
entropy estimators where the main procedure is outlined in Algorithm~\ref{alg:UM}. A central question in the implementation of Algorithm~\ref{alg:UM} is obviously how to construct a UM which can push the samples toward a uniform distribution, which is discussed in next section. 

The bias of the UM based estimators rely on the property of the UM (or equivalently the NF), on which we make the following assumption:
\begin{assumption}\label{as:2}
	Let $S=\{\mathrm{x}^{(i)}\}_{i=1}^N$ be the set of i.i.d samples used to construct the UM and  $p_{\mathrm{z}}^{S}$ be the resulting density of $Z$ in Eq.~\eqref{eq:norm1}. 
	Denote $C_2^N= \sup\limits_{\mathrm{z}\in\mathcal{Q}^o}||\triangledown p_\mathrm{z}^{S}(\mathrm{z})||_1$,  
and assume that $C_2^N$ satisfies: (1) $C_2^N \underset{N \rightarrow \infty}{\stackrel{\mathbb{P}}{\longrightarrow}} 0$; (2) 
	There exist a positive integer $M$ and a positive real number $\bar{C} < 1$ such that:
	$$\forall N>M,\quad C_2^N \leq \bar{C}, \,a.s.$$
\end{assumption}

Based on Theorem~\ref{thm1} and Theorem~\ref{thm2}, we can obtain the bias bounds and the MSE bounds of the UM based estimators.

\begin{corollary}\label{cor2}
	Suppose that the density function of the original distribution is differentiable and the UM satisfies Assumption~\ref{as:2}. The bias of UM-tKL estimator is bounded by
	\begin{equation} 		    		  
		\big|\mathbb{E}[\widehat{H}_\mathrm{UM-tKL}(X)]-H(X)\big|\leq C_\mathrm{UM-tKL}^N\big(\frac{k}{ N}\big)^\frac{1}{d},
	\end{equation}
	where $\lim\limits_{N\rightarrow \infty}C_\mathrm{UM-tKL}^N=0$. The MSE of UM-tKL estimator is bounded by
	\begin{equation}
	   \mathbb{E}[(\widehat{H}_\mathrm{UM-tKL}(X)-H(X))^2]\leq C_1\frac{1}{N}+D^N_{UM-tKL}\big(\frac{k}{ N}\big)^\frac{2}{d},
	\end{equation}
	where $C_1$ is a positive constant and $\lim\limits_{N\rightarrow \infty}D^N_{UM-tKL}=0$.
\end{corollary}

\begin{proof}
    See \ref{sec:cor2}.
\end{proof}
    
\begin{corollary}\label{cor3}
    Suppose that the density function of the original distribution is differentiable and the UM satisfies Assumption~\ref{as:2}. The bias of UM-tKSG estimator is bounded by
	\begin{equation}
		\big|\mathbb{E}[\widehat{H}_\mathrm{UM-tKSG}(X)]-H(X)\big|\leq C_{UM-tKSG}\frac{(\log N)^{k+2}}{ N^{\frac{1}{d}}},
	\end{equation}
	where $C_{UM-tKSG}=C\frac{(1+\bar{C})\big((1+\bar{C})^d+1\big)}{(1-\bar{C})^{k+1}}$ and $C$ is a positive constant. The MSE of UM-tKSG estimator is bounded by
	\begin{equation}
	    \mathbb{E}[(\widehat{H}_\mathrm{UM-tKSG}(X)-H(X))^2]\leq C_2\frac{(\log N)^{k+2}}{N}+D^N_{UM-tKSG}\frac{(\log N)^{2(k+2)}}{ N^{\frac{2}{d}}},
	\end{equation}
	where $C_2$ is a positive constant and $D^N_{UM-tKSG}=\Big(C\frac{(1+\bar{C})\big((1+\bar{C})^d+1\big)}{(1-\bar{C})^{k+1}}\Big)^2$.	
\end{corollary}

\begin{proof}
	See \ref{sec:cor3}. 
\end{proof}

		\begin{algorithm}[H]
			\caption{UM based entropy estimator}
			\label{alg:UM}
			Input: a set of i.i.d samples: $S_X=\{\mathrm{x}^{(i)}\}$;\\
			Output: an entropy estimate $\widehat{H}(X)$;
			\begin{itemize}
				\item compute a uniformizing map $f(\cdot)$;
				\item let $S_Z=\{\mathrm{z}^{(i)}=f(\mathrm{x}^{(i)}),\,i=1,...,n\}$;
				\item estimate $\widehat{H}(Z)$ from $S_Z$ using Eq.~\eqref{eq:tkl} or Eq.~\eqref{eq:tksg};
				\item compute $\widehat{H}(X)$ using Eq.~\eqref{eq:norm1}.
			\end{itemize}
		\end{algorithm}

\subsection{Constructing UM  via Normalizing Flow}\label{sec:um}
We discuss in this section how to construct a UM via the NF method. First since the image of $f$ is $[0,1]^d$, we assume that 
$f$ is in the form of $f= \Phi\circ g$ where $g:\mathcal{R}^d\rightarrow\mathcal{R}^d$ {is learned} and $\Phi:\mathcal{R}^d\rightarrow[0,1]^d$ is prescribed. 
Recall that $p_\rz$ is the distribution of $Z=f(X)$ with $X$ following $p_\rx$, and we want the function $g$ by 
minimize the Kullback-Leibler divergence (KLD) between $p_\rz$ and the uniform distribution $p_\mathrm{u}$:
\begin{equation}
	\min_{g\in \Omega} D(p_\rz|p_\ru):=\int p_\rz(\rz) \log\left[\frac{p_\rz(\rz)}{p_\ru(\rz)}\right] d\rz,\label{e:minkld1}
\end{equation}
where $z=\Phi\circ g (x)$ and $\Omega$ is a suitable function space. 
Solving Eq.~\eqref{e:minkld1} directly poses some computational difficulty as the calculation involves the function $\Phi$, the choice of 
which may affect the computational efficiency. 
To simplify the computation, we recall the following proposition: 
\begin{prop}[\cite{papamakarios2021normalizing}]\label{prop3}
	Let $T:\mathcal{Y}\rightarrow \mathcal{Z}$ be a bijective and differentiable transformation,  $p_\mathrm{z}(\mathrm{z})$ be the distribution obtained by passing $p_\mathrm{y}(\mathrm{y})$ through $T$, and $\pi_\mathrm{z}(\mathrm{z})$ be the distribution obtained by passing $\pi_\mathrm{y}(\mathrm{y})$ through $T$. Then the equality
	\begin{equation}
		D(\pi_\mathrm{y}(\mathrm{y})||p_\mathrm{y}(\mathrm{y}))=D(\pi_\mathrm{z}(\mathrm{z})||p_\mathrm{z}(\mathrm{z}))
	\end{equation}
	holds.
\end{prop}
We now construct the mapping $\Phi$ with the  cumulative distribution function of the standard normal distribution, a technique known as the probability integral transform, yielding, for a given $\ry\in R^d$, 
$$ \Phi(\mathrm{y})=(\phi_1(\mathrm{y}_1),...,\phi_d(\mathrm{y}_d)),\,\, \phi_i(\mathrm{y}_i) = \frac12(1+\mathrm{erf}(\frac{\mathrm{y}}{\sqrt{2}})),$$
where $\mathrm{erf}(\cdot)$ is the error function. 
It should be clear that if $\ry$ follows a standard normal distribution, $\rz=\Phi(\ry)$ follows a uniform distribution in $[0,1]^d$, and vice versa. 
Now applying Proposition~\ref{prop3}, we can show that Eq.~\eqref{e:minkld1} is equivalent to 
\begin{equation}
	\min_{g\in \Omega} D(p_\ry(\ry)|q(\ry)),\label{e:minkld2}
\end{equation}
where $\ry=g(\rx)$ follows distribution $p_\ry(\cdot)$ and $q(\cdot)$ is the standard normal distribution. 
Now assume that $g(\cdot)$ is invertible and let its inverse be $h= g^{-1}$. We also assume that both $g$ and $h$ are differentiable. 
Applying Proposition~\ref{prop3} to Eq.~\eqref{e:minkld2} with $T=h$, we find that Eq.~\eqref{e:minkld2} is equivalent to 
\begin{equation}
	\min_{h\in \Omega^{-1}} D(p_\rx(\rx)|q_\mathrm{h}(\rx)),\label{e:minkld3}
\end{equation}
where $\Omega^{-1}=\{g^{-1}|g\in\Omega\}$ and $q_\mathrm{h}$ is the distribution obtained by passing $q$ through the mapping $h$:
\begin{equation}
	q_\mathrm{h}(\mathrm{x}) = q\left(h^{-1}(\mathbf{x})\right) \bigg| \operatorname{det}\left(\frac{\partial h^{-1}}{\partial \mathbf{x}}\right)\bigg|.
\end{equation}

Eq.~\eqref{e:minkld3} essentially says that we want to push a standard normal distribution $q$ toward a target distribution $p_\rx$,
and therefore solving Eq.~\eqref{e:minkld3} falls naturally into the framework of NF. 
Specifically, NF aims to build such a mapping $h$  by composing multiple simple mappings: $h=h_1\circ ...\circ h_K$. 
	Each $h_k$ needs to be a  diffeomorphism: namely it is invertible and both it and its inverse are differentiable,
	which ensures that their composition $h$ is also a diffeomorphism.
	Next by plugging in the data, we can rewrite Eq.~\eqref{e:minkld3} as a maximum likelihood problem:  
	\begin{equation}
	\max_{h=(h_1,...,h_K) } E_{p_\rx}[\log q_\mathrm{h}(\rx)] : \approx \frac1N \sum_{i=1}^N \log q_\mathrm{h}(x^{(i)}).
	\end{equation}
	As is mentioned earlier, the intermediate mapping $h_i$ is usually taken to be of a simple parametrized form and so that its gradient and inverse
	are easy to compute. 
	Once $h_1,...,h_K$ are computed, the function $g$ can be obtained as
	\begin{equation}
	g=	(h_1\circ\cdots\circ h_K)^{-1}=h_K^{-1}\circ\cdots\circ h_1^{-1},
	\end{equation}
	and recall that in Eq.~(13) in the main paper we also need the   det-Jacobian of mapping $g^{-1}$ (i.e., $h$), which can be calculated as, 
	\begin{equation}
	\det \frac{\partial g^{-1}(\ry)}{\partial \ry} = 
	\det \frac{\partial h_1(\mathrm{y}_1)}{\partial \mathrm{y}_1}\circ\cdots\circ\det \frac{\partial h_K(\mathrm{y}_K)}{\partial \mathrm{y}_K},
	\end{equation}
	where $\ry_K=\ry$, $\ry_0=\rx$ and $\ry_{k-1} = h_k(\ry_{k})$ for $k=1,...,K$. 
	
	The NF methods depend critically on the component layers, the choice of which has to be balanced between 
	computational efficiency and representing flexibility. 
	In this paper, we use a special version of NF, the Masked Autoregressive Flow (MAF)~\cite{papamakarios2017masked} that is originally designed for density estimation. 
	Since the purpose of  MAF 
	is to estimate the density $p_\rx$, it is  specifically designed to efficiently evaluate the inverse mappings, which is 
	thus particularly useful for our application.   
	We note, however, our method does not rely on any specific implementation of NF. 

Once the mapping $h(\cdot)$ (or equivalently $g^{-1}(\cdot)$) is obtained, it can be inserted directly into Algorithm~1 to estimate the sought entropy. 
In practice, the samples  are split into two sets, where one of them is used to construct the UM and the other is used to estimate the entropy.


\section{Numerical Experiments}\label{sec:examples}
{ Before diving into the applications, we conduct several numerical comparisons of the proposed estimators using mathematical examples. 
The code for reproducing these examples can be found in \url{https://github.com/ziq-ao/NFEE}.}
\subsection{An Illustrating Example for the Truncated Estimators}
	Here we use a toy example to demonstrate 
	the improvement of the truncated estimators over the na\"ive version. 
	Specifically, the test example is an independent multivariate Beta distributions $B(b,b)$ 
	with dimensionality $d$ and shape parameter $b$.
	In the numerical experiments, the dimensionality is varied from 1 to 40 and the 
	parameter $b$ takes three values $1$, $1.5$ and $2$. 
	In each setup, we generate 1000 samples from the distribution and use KL, KSG, tKL and tKSG to estimate the entropy. 
	All experiments are repeated 100 times and the Root-mean-square-error (RMSE) of estimates are computed. 
	In Fig.~\ref{fig:ksg_vs_tksg}, we plot the RMSE (on a logarithmic scale) against the dimensionality $d$. From this figure, we can see that the truncated methods (blue lines) significantly outperform the na{\"i}ve ones (red lines) in all cases,
	indicating that the truncation technique can improve the performance of the KL/KSG estimators for compactly supported distributions.
	\begin{figure}[ht]
		\vspace{-0.1in}
		\centering
		\includegraphics[scale=0.6]{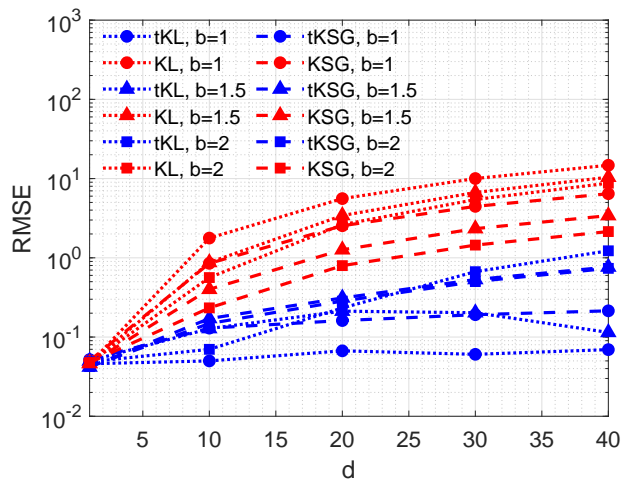}
		\caption{truncated estimators vs non-truncated estimators for multidimensional Beta distributions with various shape parameters $b$.}
		\label{fig:ksg_vs_tksg}
		\vspace{-0.1in}
	\end{figure}
	
\subsection{Multivariate Normal Distribution}
\begin{figure}
	\vspace{-0.0in}
	\begin{minipage}{0.49\linewidth}
		\centering
		\includegraphics[scale=0.30]{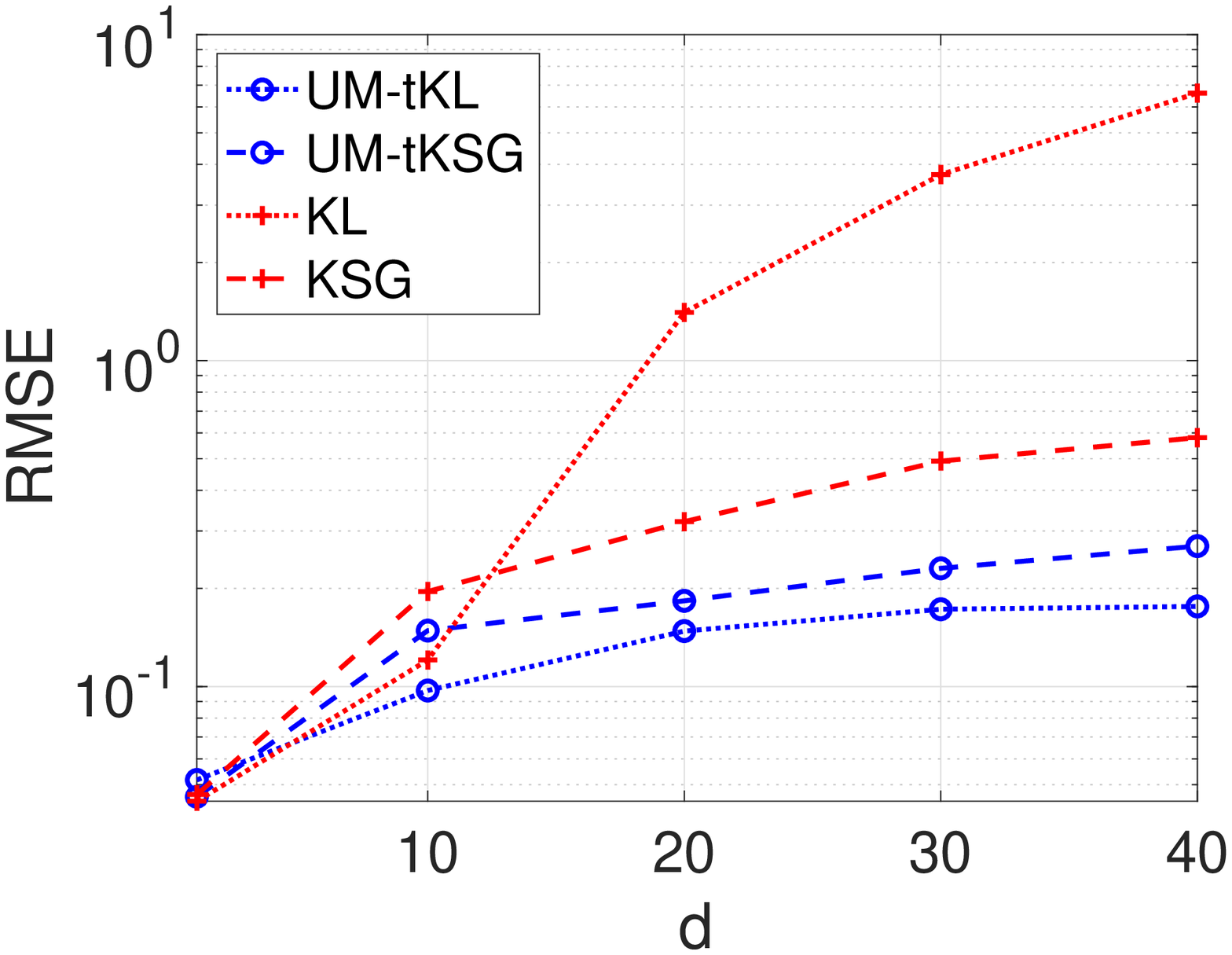}
	\end{minipage}
	\begin{minipage}{0.49\linewidth}
		\centering
		\includegraphics[scale=0.30]{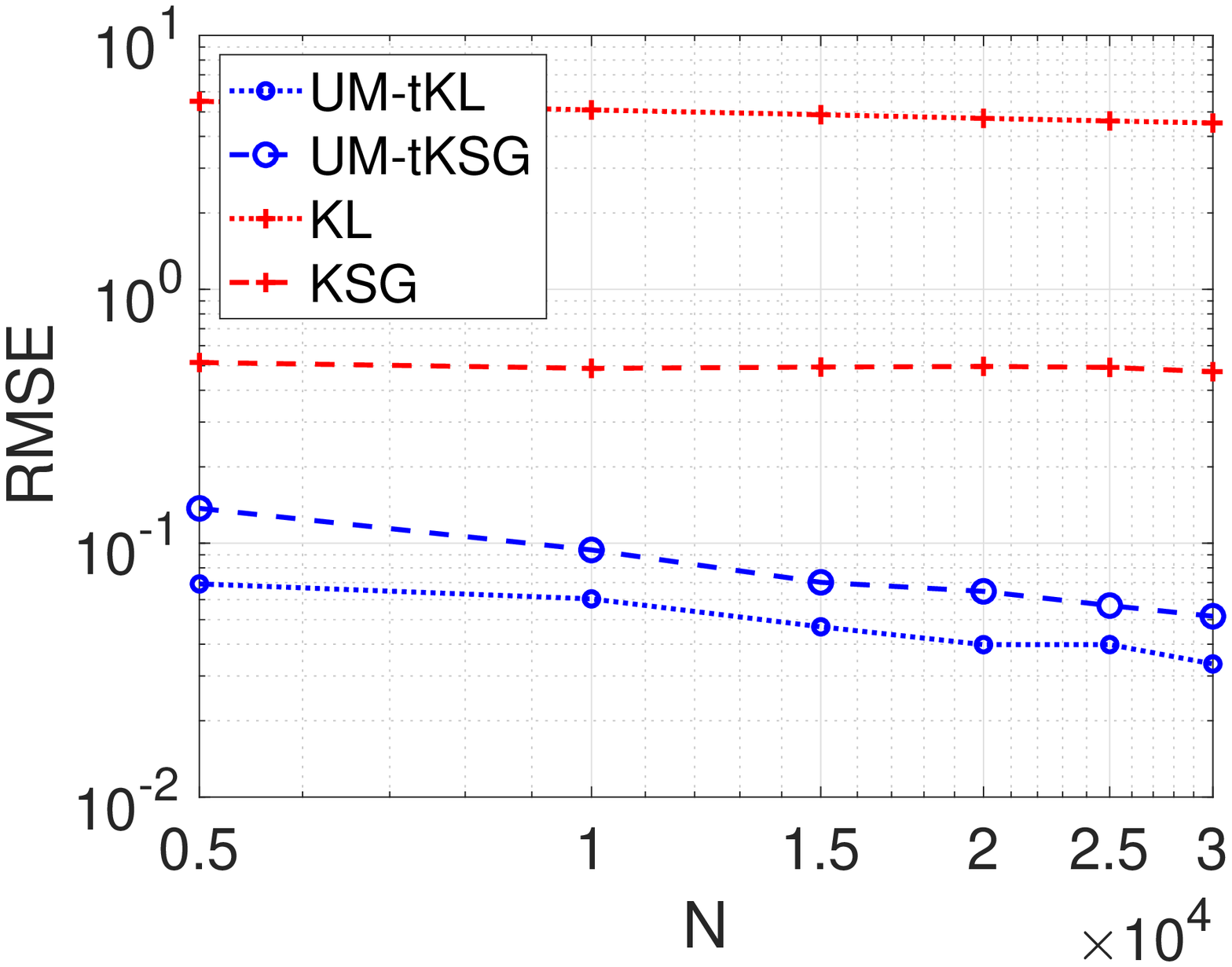}
	\end{minipage}
		\vspace{-0.0in}
	\caption{Left: RMSE plotted against the dimensionality $d$. Right: RMSE (on a logarithmic scale) plotted against the sample size $N$. }
	\label{fig:mvn_d}
	\vspace{-0.0in}
\end{figure}

To validate the idea of UM based entropy estimator, a natural question to ask is that how it works with a perfect NF transformation,
that yields exactly normally distributed samples.
To answer this question, we first conduct the numerical tests with the standard multivariate normal distribution, corresponding to the situation
that one has done a perfect NF {(in this case the function $g$ in Section~\ref{sec:um} is chosen to be identity map)}.  

Specifically we test the four methods: KL, KSG, UM-tKL and UM-tKSG,
and we conduct two sets of tests: in the first one we fix the sample size to be 1000 and vary the dimensionality, while in the second one 
we fix the dimensionality to be 40 and vary the sample size. 
All the tests are repeated 100 times and the RMSE of the estimates are calculated.
In Fig.~\ref{fig:mvn_d} (left),
we plot the RMSE  (on a logarithmic scale) as a function of the dimensionality. One can see from this figure that, as the dimensionality increases, the estimation error in KL and KSG grows significantly faster than that in the two UM based ones,
with the error in KL being particularly large. 
Next in Fig.~\ref{fig:mvn_d} (right) we plot the RMSE against the sample size $N$ (note that the plot is on a log-log scale) for $d=40$, which shows 
that for this high-dimensional case, the two UM based estimators yield much lower and faster-decaying RMSE than those two estimators on the 
original samples.  
Overall these results support the theoretical findings in Section~\ref{sec:truncation} that the estimation error can be significantly reduced by mapping the target samples toward a uniform distribution. 

\subsection{Multivariate Rosenbrock Distribution}\label{sec:mrd}
In  this example we shall see how the proposed method performs when  NF is included.
Specifically our example is the Rosenbrock type of distributions -- the standard  Rosenbrock distribution is 2-D and  widely used as a testing example for various of statistical methods. Here we consider two high-dimensional  extensions 
of the 2-D Rosenbrock ~\cite{pagani2019n}: the hybrid Rosenbrock (HR) and the even Rosenbrock (ER) distributions.  
The details of the two distributions including their density functions are provided in \ref{sec:hr}. 
The Rosenbrock distribution is strongly non-Gaussian, and that can be demonstrated by Fig.~\ref{f:rosen} (left) which shows the samples drawn from 2-D Rosenbrock.  As a comparison, Fig.~\ref{f:rosen} (right) shows the samples that have been transformed toward a uniform distribution and  used in entropy estimation. 

\begin{figure}[h]
	\vspace{-0.0in}
	\centerline{\includegraphics[width=1.0\linewidth]{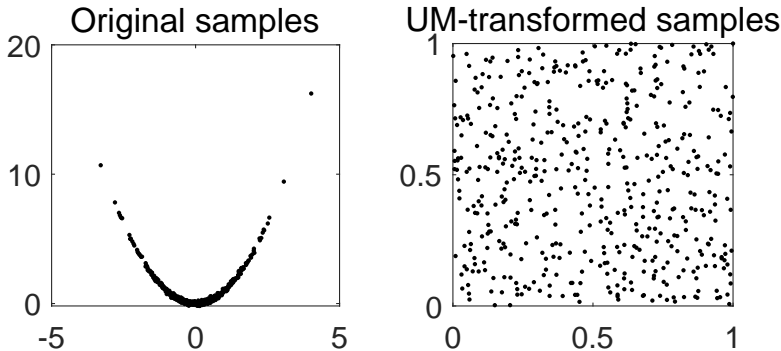}}
	\vspace{-0.0in}
	\caption{Left: the original samples drawn from a 2-D Rosenbrock distribution; Right: the UM-transformed samples 
		used in the entropy estimation.}\label{f:rosen}
	\vspace{-0.0in}
\end{figure}

\begin{figure*}[]
	\vspace{-0.0in}
	\begin{minipage}{0.49\textwidth}
		\centering
		\centerline{(a)}
		\includegraphics[scale=0.35]{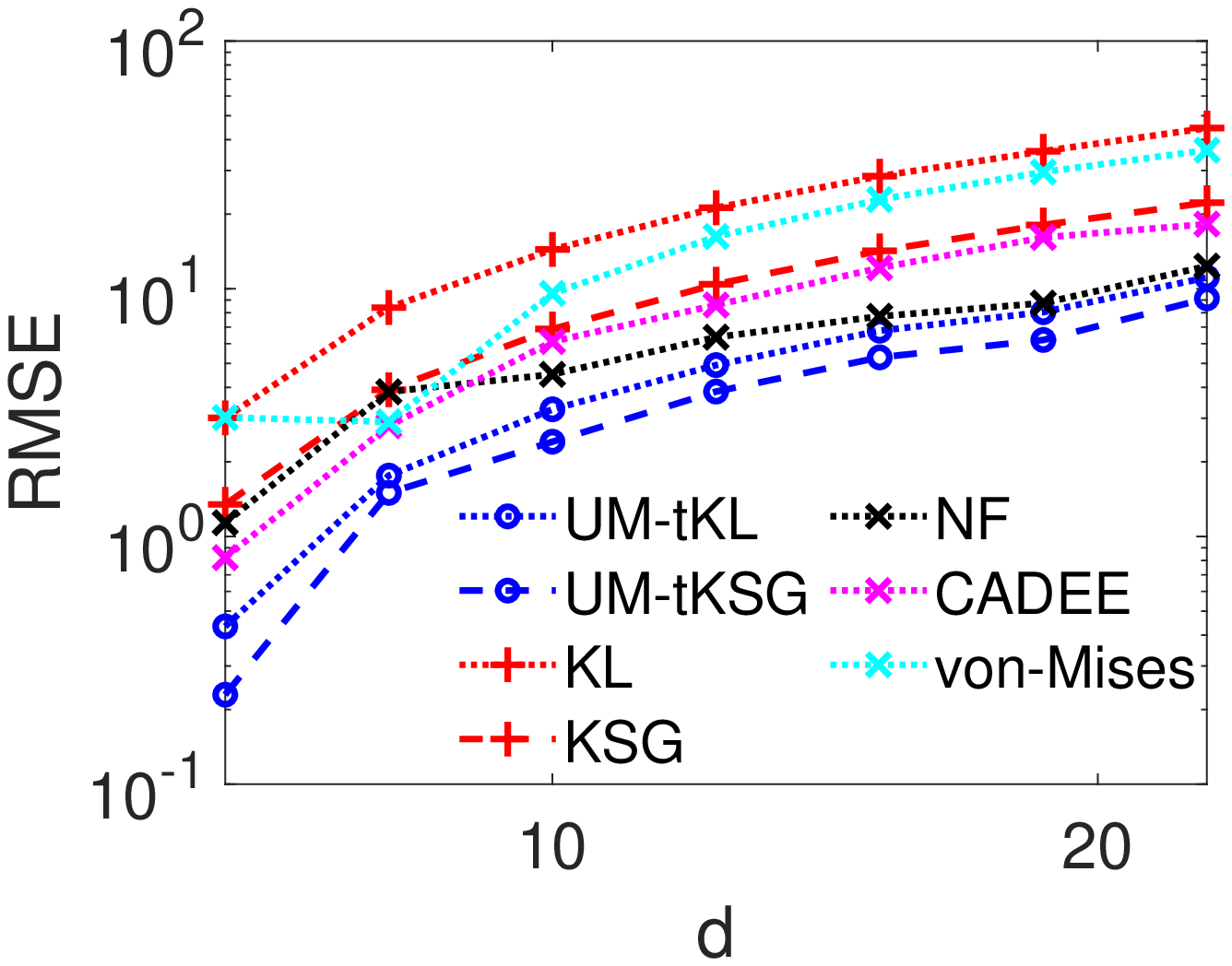}
	\end{minipage}
	\begin{minipage}{0.49\textwidth}
		\centering
		\centerline{(b)}
		\includegraphics[scale=0.35]{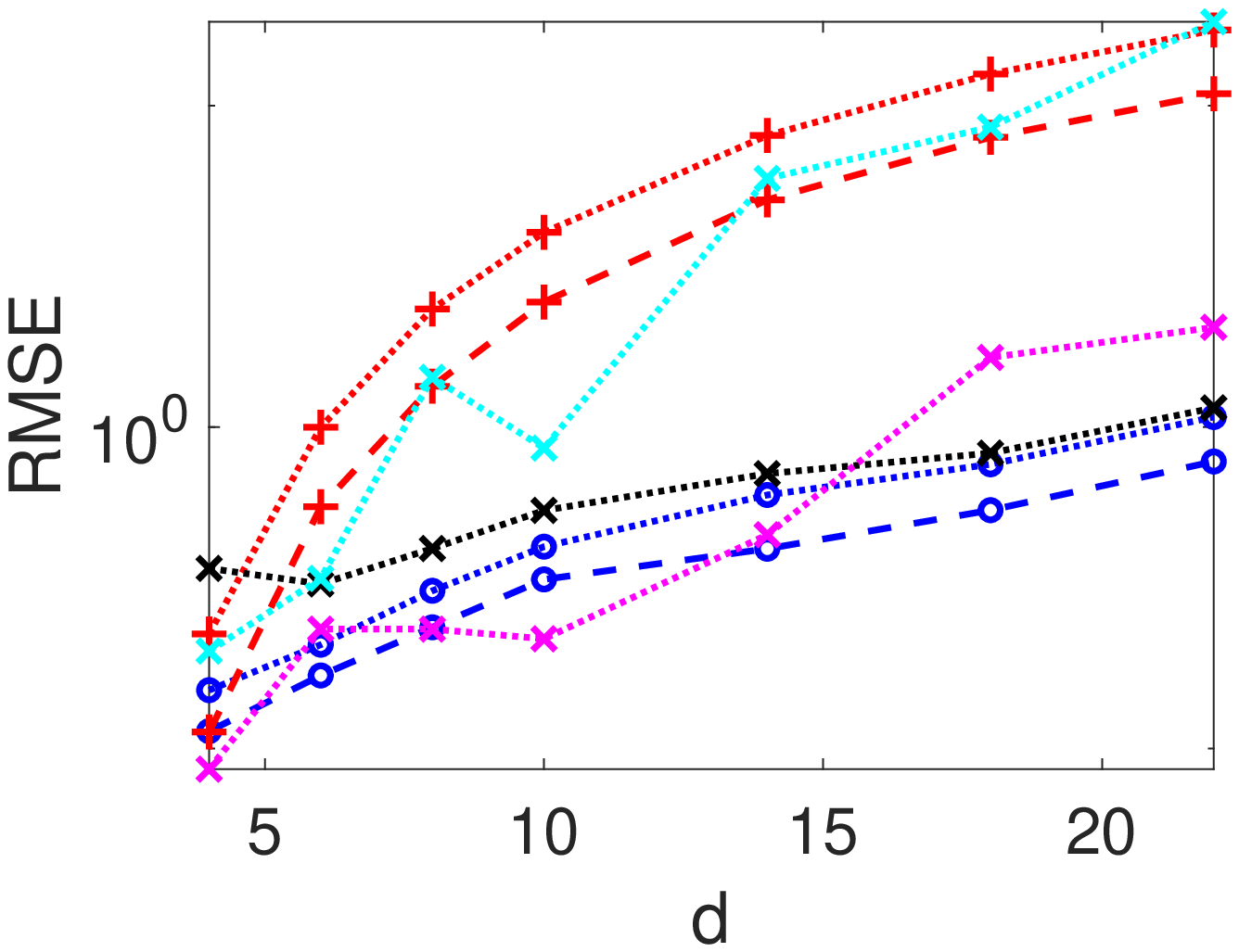}
	\end{minipage}
    \\
	\begin{minipage}{0.49\textwidth}
		\centering
		\centerline{(c)}
		\includegraphics[scale=0.35]{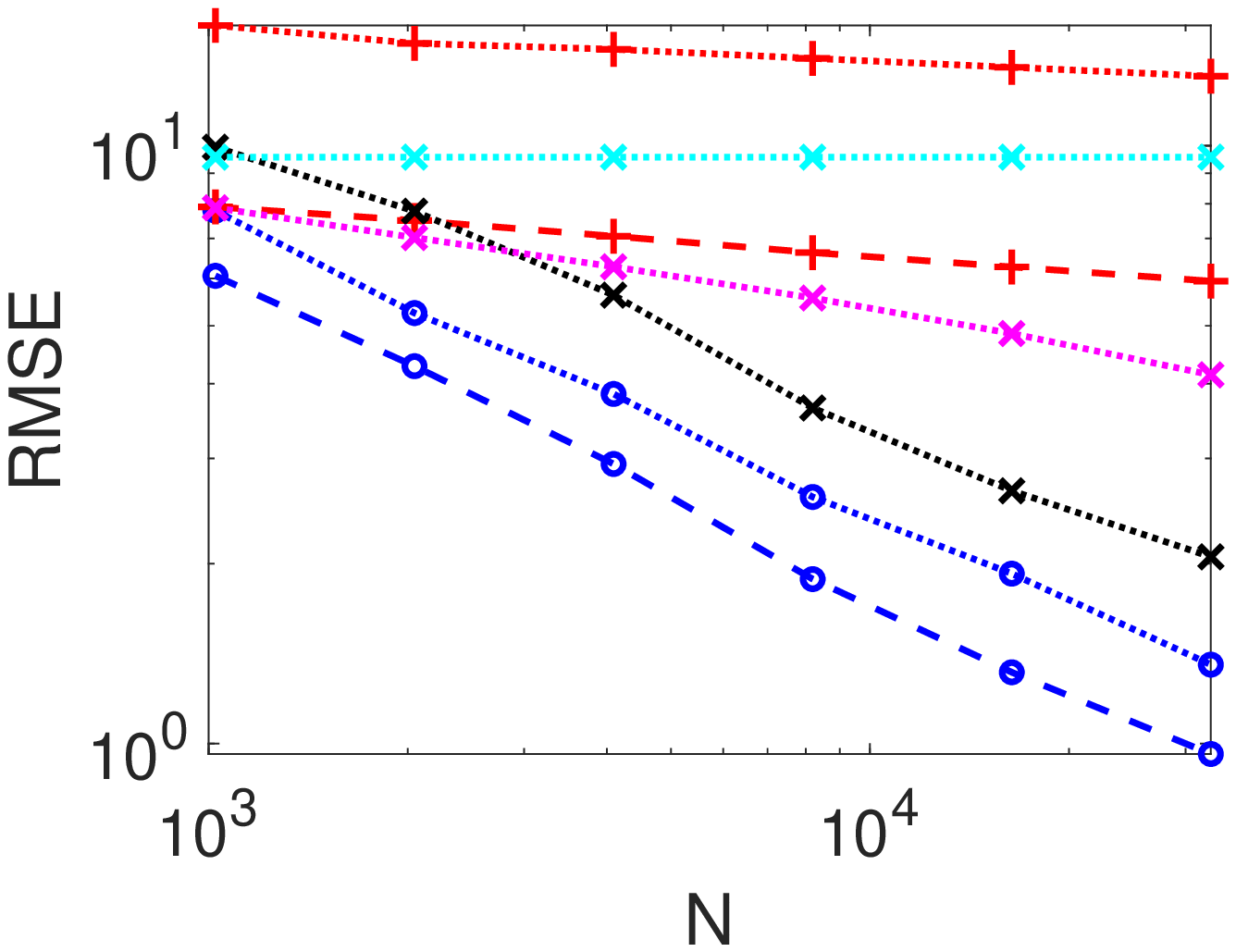}
	\end{minipage}
	\begin{minipage}{0.49\textwidth}
		\centering
		\centerline{(d)}
		\includegraphics[scale=0.35]{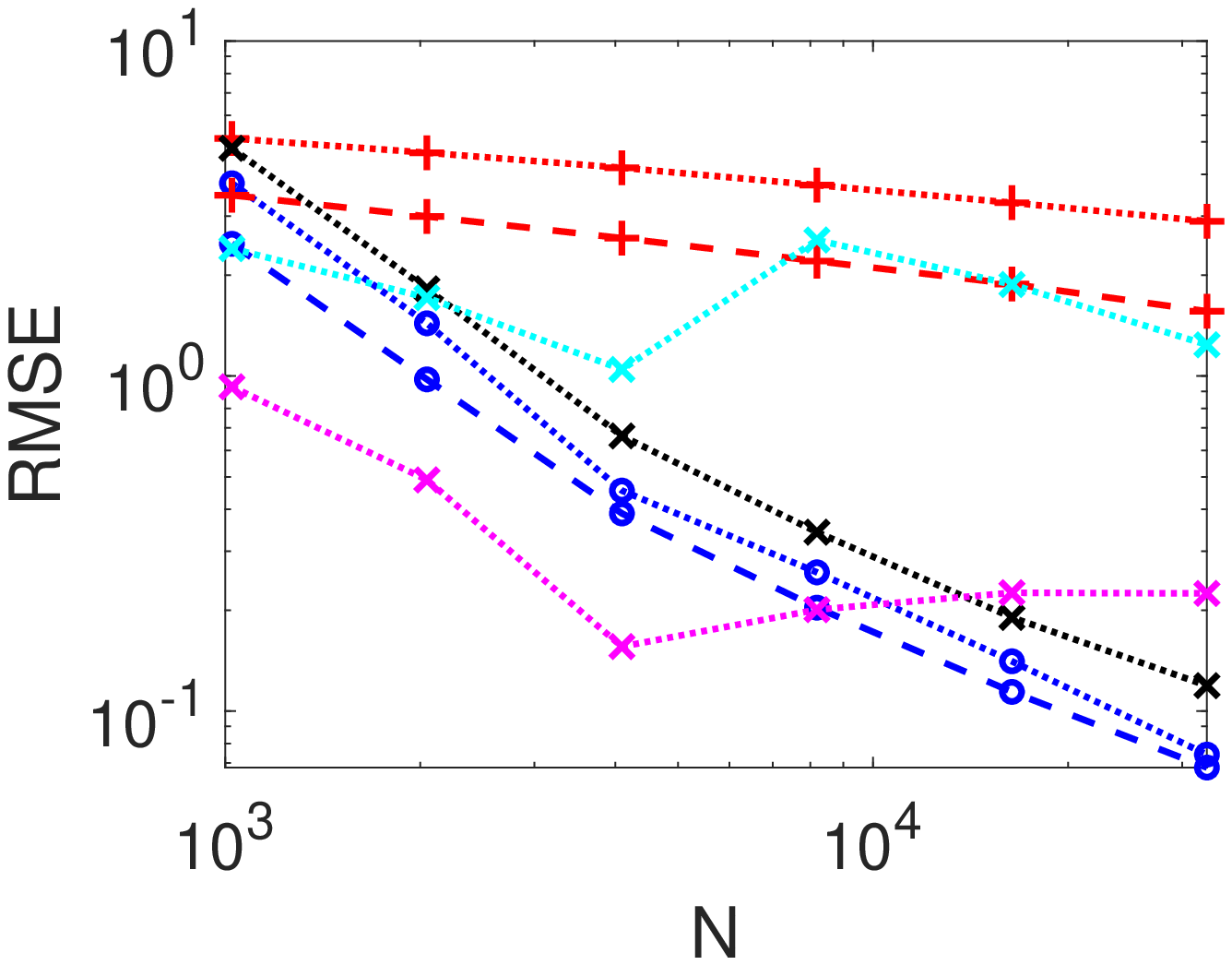}
	\end{minipage}
	\vspace{-0.0in}
	\caption{ Top: RMSE vs. dimensionality for HR (a) and  ER (b); Bottom:
		RMSE vs. sample size for HR (c) and ER (d).}
	\label{fig:mr_d}
	\vspace{-0.00in}
\end{figure*}

In this example we compare the performance of seven estimators: in addition to the four used in the previous example, 
we include an estimator only using NF (details in SI) 
as well as two state-of-the-art entropy estimators: CADEE \cite{ariel2020estimating} and the von-Mises based estimator \cite{kandasamy2015nonparametric}. 
First we test how the estimators scale with respect to dimensionality, where the sample size is taken to be $N=500d$. 
With each method, the experiment is repeated 20 times  and the RMSE is calculated.  
The RMSE against the dimensionality $d$ for both test distributions is plotted in Figs.~\ref{fig:mr_d} (a) and (b). 
One can observe here that in most cases, the UM based methods (especially UM-tKSG) offer the best performance. 
An exception is that  CADEE performs better in low dimensional cases for ER, but its RMSE grows much higher than that of the UM methods 
in the high-dimensional regime ($d>15$). 
Our second experiment is to fix the dimensionality at $d=10$ and vary the sample size, where the RMSE is plotted against the sample size 
for both HR and ER in Figs.~\ref{fig:mr_d} (c) and (d). The figures show clearly that the RMSE of the UM based estimators decays faster than other methods in both examples,  with the only exception being CADEE in the small sample ($\leq 10^{4}$) regime of ER. 
It is also worth noting that, though it is not justified theoretically, 
UM-tKSG seems to perform slightly better than UM-tKL in all the cases.

\subsection{Multivariate Rosenbrock Distribution with Discontinuous Density}

Recall that Corollaries~\ref{cor2} and \ref{cor3} assume the differentiability of the original density functions, which is often not satisfied by practice. Thus, it is also of interest to examine the performance of the proposed methods for distributions with discontinuous densities. To this end, we modify the multivariate Rosenbrock distributions studied in Section~\ref{sec:mrd}, so that their densities are discontinuous on the boundaries of their supports (see \ref{sec:hr} for the details), and repeat the comparisons conducted in Section~\ref{sec:mrd}.  The results are shown in Figs.~\ref{fig:mr_d2}. For the modified HR (in Fig.~\ref{fig:mr_d2} (a) and (c)), only the von-Mises estimator achieves a smaller RMSE than the UM based ones in the low-dimensional regime (d$\leq$10), while the UM based estimators  perform the best in the high-dimensional regime. For modified ER (in Fig.~\ref{fig:mr_d2} (b) and (d)), the UM based estimators are inferior to CADEE but outperform any other methods in most cases.

\begin{figure*}[]
	\vspace{-0.0in}
	\begin{minipage}{0.49\textwidth}
		\centering
		\centerline{(a)}
		\includegraphics[scale=0.35]{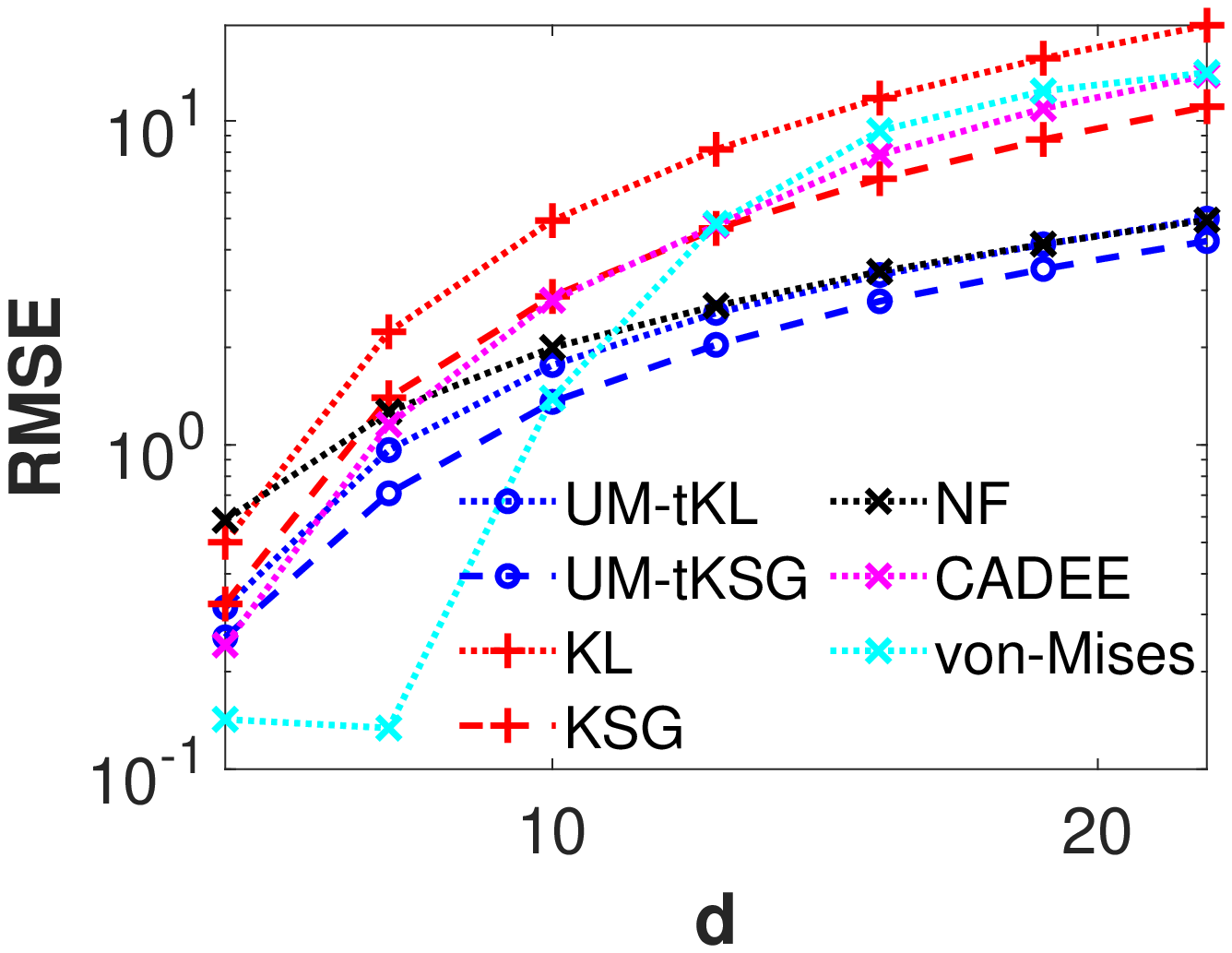}
	\end{minipage}
	\begin{minipage}{0.49\textwidth}
		\centering
		\centerline{(b)}
		\includegraphics[scale=0.35]{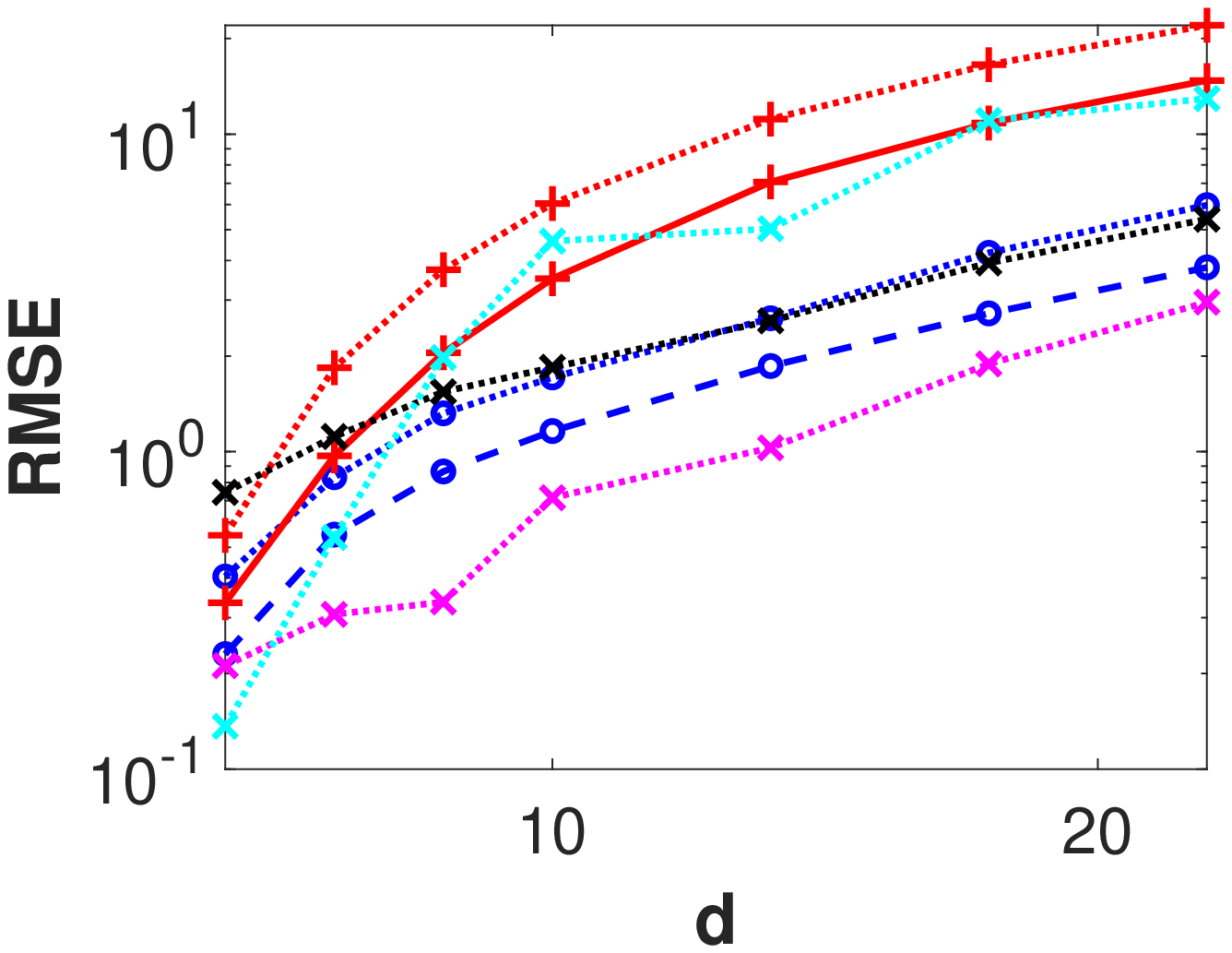}
	\end{minipage}
    \\
	\begin{minipage}{0.49\textwidth}
		\centering
		\centerline{(c)}
		\includegraphics[scale=0.35]{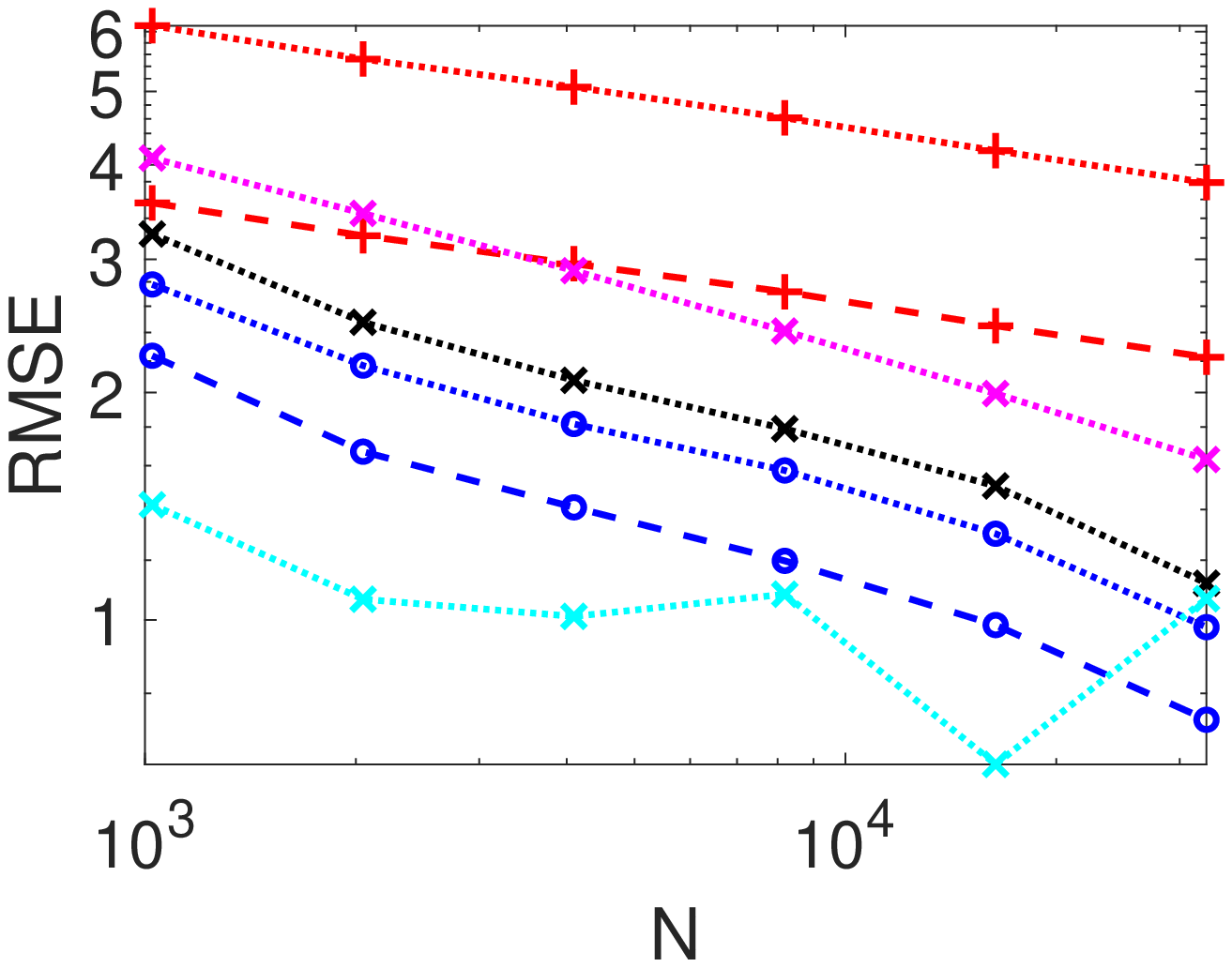}
	\end{minipage}
	\begin{minipage}{0.49\textwidth}
		\centering
		\centerline{(d)}
		\includegraphics[scale=0.35]{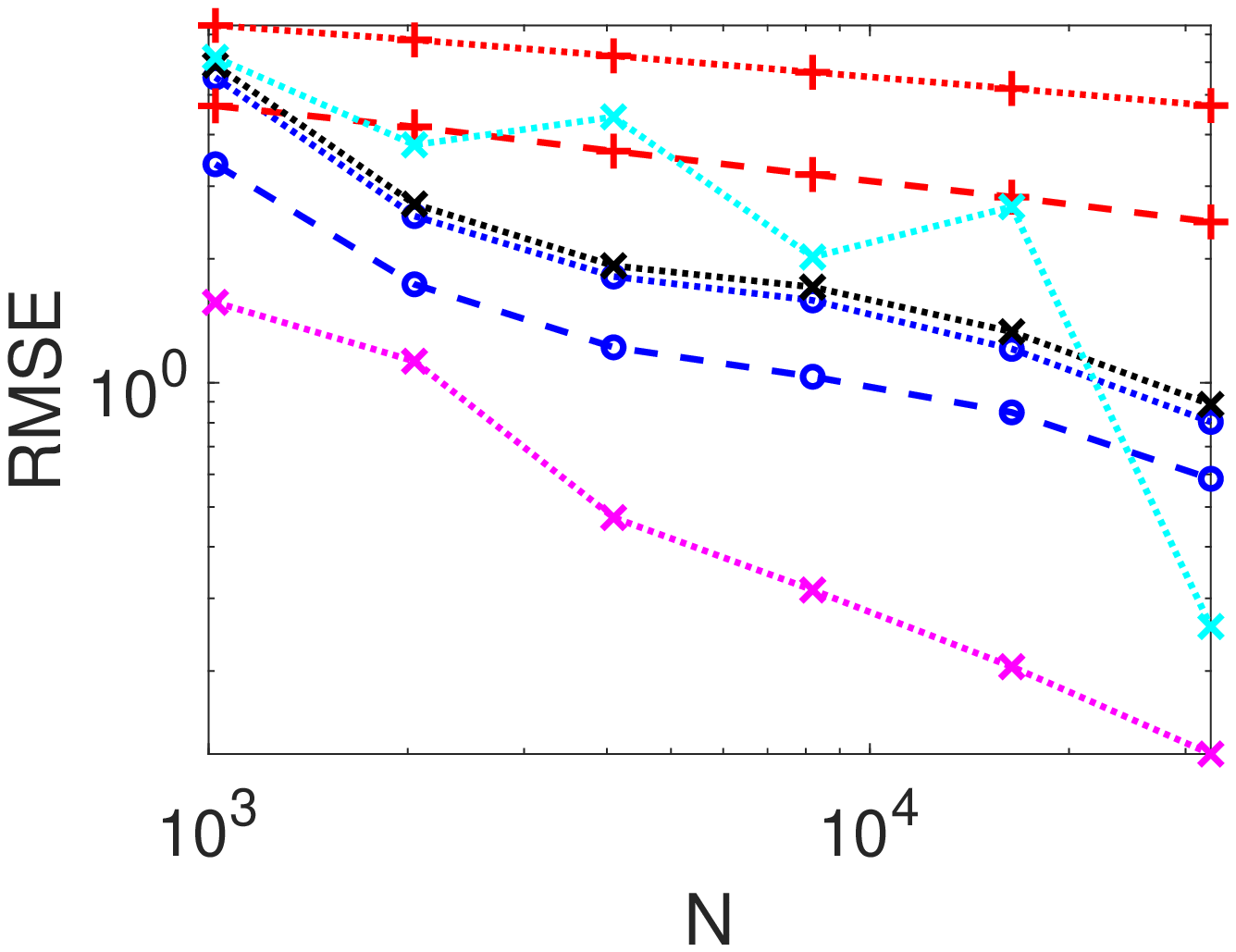}
	\end{minipage}
	\vspace{-0.0in}
	\caption{ Top: RMSE vs. dimensionality for modified HR (a) and  ER (b); Bottom:
		RMSE vs. sample size for modified HR (c) and ER (d).}
	\label{fig:mr_d2}
	\vspace{-0.00in}
\end{figure*}

\section{Application Examples}\label{sec:application}
{ In this section, we consider two applications involving entropy estimation,
in which our methods are compared with the existing ones.}
\subsection{Application to Entropy Rate Estimation}
\label{sec:ere}
    Our first application 
    example is to estimate the differential entropy rate of a continuous-valued time series. Shannon entropy rate \cite{shannon1948mathematical} measures the uncertainty of a stochastic process $\mathcal{X}=\{X_i\}_{i \in \mathbb{N}}$. For a stationary process, it is defined as, 
    \begin{equation}
        \bar{H}(\mathcal{X})=\lim\limits_{t\rightarrow\infty}H(X_t\mid X_{t-1},...,X_{1}),
    \end{equation}
    where $H(\cdot\mid\cdot)$ is the conditional entropy of two random variables. In this example, we consider the stochastic processes that satisfy the following two assumptions:
    \begin{itemize}
        \item  First $\mathcal{X}$ is a conditionally stationary process of order $p$: there exists a fixed positive integer $p$ such that, for
       any integer $t>p$,  the conditional density function of $X_t$ given $X_{t-1}=x_{t-1},..., X_{t-p}=x_{t-p}$  satisfies
    \begin{equation}
        p(X_t=x_t\mid X_{t-1}=x_{t-1},..., X_{t-p}=x_{t-p})=f(x_{t}\mid x_{t-1},..., x_{t-p}),
    \end{equation}
    where $f$ is a fixed conditional density function independent from $t$.
    \item Second $\mathcal{X}$ is  a Markov process of order $p$: there exists a positive integer $p$ such that, for any integer $t>p$, 
    \begin{multline}p(X_{t}=x_{t}\mid X_{t-1}=x_{t-1},..., X_{1}=x_{1})\\=p(X_{t}=x_{t}\mid X_{t-1}=x_{t-1},..., X_{t-p}=x_{t-p}.)\end{multline}
    \end{itemize}
    Under these assumptions, the entropy rate of $\mathcal{X}$ can be calculated as, 
    \begin{equation}
        \bar{H}=H(X_t\mid X_{(t-1):(t-p)})
        =H(X_{t:(t-p)})-H(X_{(t-1):(t-p)}), \label{eg:hbar1}
         \end{equation}
       where $X_{t:(t-p)}=(X_{t},X_{t-1},...,X_{t-p})$ and so on. 
       Note here that $t$ can be taken to be any integer $>p$, and 
       for simplicity we can take it to be $t=p+1$, and as a result Eq.~\eqref{eg:hbar1} is simplified to, 
        \[\bar{H}=H(X_t\mid X_{(t-1):(t-p)})\\
        =H(X_{(p+1):1})-H(X_{p:1}).\]
    Suppose that we have a $T$-step (with $T>p$) observation of $\mathcal{X}$: $\{x_{t}\}_{t=1}^T$, 
    and we can compute its entropy rate as follows~\cite{darmon2016specific}:
    \[\hat{H} = \widehat{H}(X_{(p+1):1})-\widehat{H}(X_{p:1}),\]
    where $\widehat{H}(X_{(p+1):1})$
    and $\widehat{H}(X_{p:1})$
    are estimated with a desired estimator from the observation $\{x_{t}\}_{t=1}^T$.

    
   
    
    In this example, we consider three autoregressive models of orders $3$,  $7$ and  $15$ respectively, which are given by
    \begin{subequations}
    \begin{align}
        AR(3): X_t = -1.35+0.5X_{t-1}+0.4X_{t-2}^2-0.3X_{t-3}+\epsilon_t,\\
        AR(7): X_t = -1.35+0.5X_{t-1}+0.3X_{t-5}^2-0.3X_{t-7}+\epsilon_t,\\
      AR(15): X_t  = -1.35+0.5X_{t-1}+0.05(X_{t-5}+X_{t-6}+X_{t-7})^2\notag\\
        -0.005(X_{t-11}+X_{t-12}+X_{t-13})^2-0.1X_{t-15}+\epsilon_t,
    \end{align}
    \end{subequations}
    where $\epsilon_t\sim\mathcal{N}(0, (0.03)^2)$ is white noise. Fig.~\ref{f:ere} shows the simulated snapshots of the three models. 
    We implemented the procedure described above to estimate the entropy rate of these 
    three models where the entropy is estimated with the seven estimators used in Section~\ref{sec:examples}. 
      On the other hand,  since the conditional density functions are analytically available in this example,  
      the entropy rate can also be directly estimated via the standard Monte Carlo integration, which will be used as the {\it ground truth}. We apply the aforementioned entropy estimators to compute the entropy rate with a simulated sequence of  $10,000$ steps. With each method, 20 repeated trials are conducted and the RMSE is calculated. The results are reported in Table~\ref{tb:ere}, from which we make the following observations. 
      The performance of the von-Mises estimator appears to be the best for the $AR(3)$ model, however, all estimators yield very small Root Mean Squared Error (RMSE) suggesting that this problem is not particularly challenging. For the $AR(7)$ model, the UM-based methods have smaller RMSE than the others, and for the $AR(15)$ model, the two UM-based methods and KSG perform better than the other three. Overall,  UM-KSG  results in the smallest RMSE for both $AR(7)$
      and $AR(15)$.
      

\begin{figure}
	\vspace{-0.0in}
	\centerline{\includegraphics[width=0.65\linewidth]{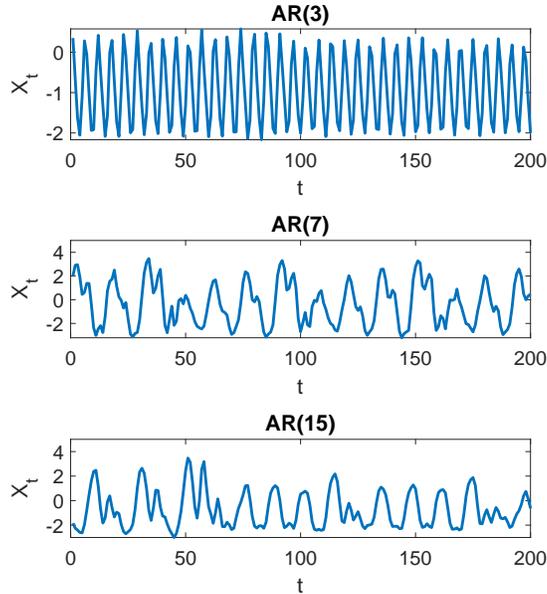}}
	\vskip -0.1in
	\caption{Snapshots of the simulated time series.}\label{f:ere}
	\vspace{-0.0in}
\end{figure}

\begin{table*}[htp]
	\footnotesize
	\resizebox{\textwidth}{8mm}{
	\centering
	\begin{tabular}{|c|c|c|c|c|c|c|c|}
		\hline
		{\bf Method}& {UM-tKL} & {UM-tKSG} & KL &KSG  & NF & CADEE & von-Mises \\ \hline
		{\bf AR(3)}& 0.029 & 0.051 & 0.027  & 0.032 & 0.12 & 0.31 & {\bf 0.016} \\ \hline
		{\bf AR(7)}& { 0.67} & {\bf 0.43} & 1.23  & 0.90 & 0.95 & 2.40 & 0.70 \\ \hline
		{\bf AR(15)}& {1.15} & {\bf 0.68} & 1.51  & 0.98 & 1.61 & 4.14 & 1.42 \\ \hline
	\end{tabular}
	}
	\vskip -0.00in
	\caption{RMSE of entropy rate estimations based on entropy estimators for the autoregressive model.
		The smallest (best) RMSE value is shown in bold. }  \label{tb:ere}
	\vskip -0.0in
\end{table*}

\subsection{Application to Optimal Experimental Design} \label{sec:lv}
In this section, we apply entropy estimation to an optimal experimental design (OED) problem. Simply put,
the goal of OED is to determine the optimal experimental conditions (e.g., locations of sensors) that maximize certain utility function 
associated with the  experiments. 
Mathematically let $\lambda\in \mathcal{D}$ be design parameters representing experimental conditions, 
$\theta$ be the parameter of interest, and $Y$ be the observed data.
An often used utility function is the entropy of the data $Y$, resulting in the so-called maximum entropy sampling method~(MES)~\cite{sebastiani2000maximum}:
\begin{equation}
	\max_{\lambda\in \mathcal{D}}U(\lambda): = H(Y|\lambda),
\end{equation}
and therefore evaluating $U(\lambda)$ becomes an entropy estimation problem.
This utility function is equivalent to the mutual entropy criterion under certain conditions \cite{shewry1987maximum}. 
This formulation is particularly useful for problems with expensive or intractable likelihoods, as the likelihoods are not needed if the utility function is computed via entropy estimation. A common application of OED  is to determine the observation times for stochastic processes so that one can accurately estimate the model parameters and here we provide such an example, arising from the field of population dynamics.

Specifically we consider the Lotka-Volterra (LV) predator-prey model \cite{lotka1925elements,volterra1927variazioni}. 
Let $x$  and $y$ be the populations of prey and  predator respectively, and the LV model is given by \\
\centerline{$\dot{x}= a x-x y,\quad \dot{y}=b x y- y$,}
where $a$ and $b$ are respectively the growth rates of the prey and the predator. 
In practice, often the parameters $a$ and $b$ are not known and need to be estimated from the population data. 
In a Bayesian framework, one can assign a prior distribution on $a$ and $b$, and infer them from measurements made on the population $(x,y)$. 
Here we assume that the prior for both $a$ and $b$ is a uniform distribution $U[0.5,4]$. 
In particular we assume that the pair $(x+\epsilon_x,y+\epsilon_y)$, where $\epsilon_x, \epsilon_y\sim\mathcal{N}(0,0.01)$ are independent observation noises, is measured at $d=5$ time points located within the interval $[0,10]$,  
and the goal is to determine the observation times for the  experiments. 
As is mentioned earlier, we shall determine the observation times using the MES method. 
Namely, the design parameter in this example is $\lambda=(t_1,...,t_d)$, the data $Y$ is the pair $(x+\epsilon_x,y+\epsilon_y)$ measured at $t_1,...,t_d$, and we want to
find $\lambda$ that maximizes the entropy $H(Y|\lambda)$. 

\begin{figure}
	\vspace{-0.0in}
	\centerline{\includegraphics[width=0.65\linewidth]{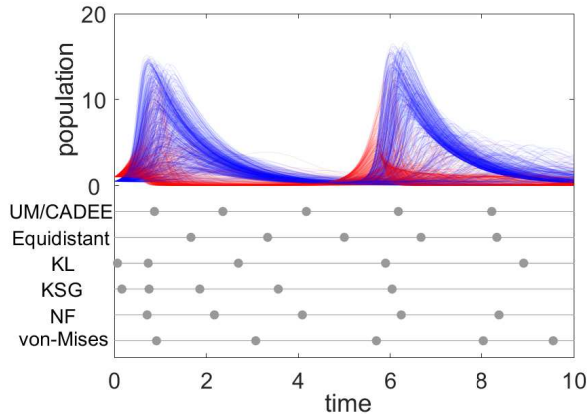}}
	\vskip -0.1in
	\caption{Top: some sample data paths of $(x,y)$; Bottom: the optimal observation times 
		obtained by the eight methods.}\label{f:oed}
	\vspace{-0.0in}
\end{figure}

\begin{table*}[htp]
	\footnotesize
	\resizebox{\textwidth}{8mm}{
	\centering
	\begin{tabular}{|c|c|c|c|c|c|c|c|c|}
		\hline
		{\bf Method}& {UM-tKL} & {UM-tKSG} & CADEE & Equidistant & KL &KSG  & NF & von-Mises \\ \hline
		{\bf NMC}& \multicolumn{3}{c|}{\bf -1.45} & -2.73 & -1.65 & -1.56 & -1.48 & -1.81    \\ 
		{\bf (SE)}&\multicolumn{3}{c|}{\bf (0.0073)} &(0.0074)&(0.0072)&(0.0076)&(0.0072)&(0.0049) \\ \hline
		{\bf RMSE}& {\bf 0.73} & {\bf 0.48} & 0.86 & --- & 3.60 & 1.05 & 0.88 & 1.31 \\ \hline
	\end{tabular}
	}
	\vskip -0.00in
	\caption{The reference entropy values of the  observation time placements obtained by using all the methods.
		The smallest (best) entropy value is shown in bold. }  \label{tb:oed}
	\vskip -0.0in
\end{table*}

A common practice in such problems is not to optimize the observation times directly and instead parametrize them using the percentiles 
of a prescribed distribution to reduce the optimization dimensionality~\cite{Ryan2014towards}. Here we use a Beta distribution, resulting in
two distribution parameters to be optimized (see \cite{Ryan2014towards} and \ref{sec:betascheme} for further details). 
We solve the resulting optimization problem with a grid search where the entropy is evaluated by the seven aforementioned estimators
each with 10,000 samples. We plot in Fig.~\ref{f:oed} the optimal observation time placements computed with the seven aforementioned estimators, as well as the equidistant placement for a comparison purpose.
Also shown in the figure are some sample paths of the population $(x,y)$ where we can see that the population samples are generally
subject to larger variations near the two ends and relative smaller ones in the middle. 
Regarding the optimization results, we see that the optimal time placements obtained by the two UM based estimators and CADEE are the same, while they are different 
from the results of other methods. 
To validate the optimization results, we compute a reference entropy value for the optimal placement obtained by each method,
using Nested Monte Carlo (NMC) (see \cite{ryan2003estimating} and \ref{sec:nestedMC} for details) with a large sample size ($10^5\times10^5$), and show the results in Table~\ref{tb:oed}. 
Note that though the NMC can produce a rather accurate entropy estimate, it is too expensive to use directly in this OED problem. 
Using the reference values as the ground truth, we can further compute the RMSE of these estimates (over 20 repetitions), which are also reported in Table~\ref{tb:oed}. 
From the table one observes that the  placement of observation times computed by the two UM methods and CADEE yields the largest entropy values,
which indicates  that  these three methods clearly outperform all the other estimators in this OED problem. 
Moreover, from the RMSE results we can see that the UM based methods (especially UM-tKSG) yield smaller RMSE than CADEE, suggesting that they
are more statistically reliable than CADEE. 

\section{Conclusion}\label{sec:conclusion}
In summary, we have presented a uniformization based entropy estimator, which we believe can be useful for a wide range of real-world applications.
While our theoretical results provide some justification for the method, further analysis is needed to establish the convergence rate and understand the estimation bias. Additionally, the method can be extended to estimate other density functionals, such as the Renyi entropy and the Kullback-Leibler divergence.
Finally  in this work  the proposed method is demonstrated only with synthetic data, and 
it is therefore sensible to further examine the method with real-world data sets. 
We will explore these research problems in future studies. 

\section{Acknowledgments}
This work was partially supported by the China Scholarship Council (CSC). The authors would also like to thank Dr. Alexander Kraskov for discussion about the KSG estimator.

\bibliography{mybibfile}

\appendix
\section{Proofs of Theorem~1 and Theorem~2}
	Here we provide proofs of Theorems 1\&2. We follow closely the framework from \cite{singh2016finite} and \cite{gao2018demystifying} of finite-sample analysis of fixed $k$ nearest neighbor
	entropy estimators. They both gave a bias bound of roughly $O(\big(\frac{1}{N}\big)^{\gamma/d})$ ($\gamma$ is some positive constant) and a variance bound of roughly $O(\frac{1}{N})$ for the entropy estimator $\widehat{H}_{KL}$, under some mild assumptions. Similarly here we prove that the proposed $\widehat{H}_{tKL}$ and $\widehat{H}_{tKSG}$ also have such bias and variance bounds. More interestingly, our analysis relates the bias bound of $\widehat{H}_{tKL}$ to the gradient of density function.
	
	\subsection{Definitions and assumptions}\label{sec:definition}
	In this section, we introduce some notations and assumptions that the proofs rely on. As is mentioned in the main paper, we only consider distributions with densities supported on the unit cube in $\mathbb{R}^d$. Let $\mathcal{Q}:=[0,1]^d$ denote the unit cube in d-dimensional Euclidean space $\mathbb{R}^d$ and $P$ denote an unknown $\mu$-absolutely continuous Borel probability measure, where $\mu$ is the Lebesgue measure. Let $p:\mathcal{Q}\rightarrow [0,\infty)$ be the density of $P$.

	\begin{definition}[Twice the k-NN distance for cubes]\label{def1}
		Suppose $\{\mathrm{x}^{(i)}\}_{i=1}^{N-1}$ is set of $N-1$ i.i.d. samples from $P$. We define twice the maximum-norm k-NN distance for cubes by $\epsilon_k(\mathrm{x})=2||\mathrm{x}-\mathrm{x}^{*}||_\infty$, where $\mathrm{x}^*$ is the k-nearest element amongst $\{\mathrm{x}^{(i)}\}_{i=1}^{N-1}$ to $\mathrm{x}$ with respect to $\infty$-norm. 
	\end{definition}
	
	\begin{definition}[Twice the k-NN distance for rectangles]\label{def2}
		Suppose $\{\mathrm{x}^{{(1')}}, ..., \mathrm{x}^{{(k')}}\}$ is set of the k nearest elements amongst $\{\mathrm{x}^{(i)}\}_{i=1}^{N-1}$ to $\mathrm{x}$ with respect to $\infty$-norm. We define twice the k-NN distance in the marginal direction $\mathrm{x}_j$ by $\epsilon_k^{\mathrm{x}_j}(\mathrm{x})=2|\mathrm{x}_j-\mathrm{x}_j^{*j}|$, where $\mathrm{x}^{*j}$ is the k-nearest element amongst $\{\mathrm{x}^{{(1')}}, ..., \mathrm{x}^{{(k')}}\}$ in the marginal direction $\mathrm{x}_j$ to $\mathrm{x}$.  It should be noted that $\epsilon_k(\mathrm{x})=\max\limits_{1\leq j\leq d}\epsilon_k^{\mathrm{x}_j}(\mathrm{x})$.
	\end{definition}
	
	\begin{definition}[Truncated twice the k-NN distance]\label{def3}
		Since we only consider densities supported on the unit cube, we define so-called truncated distance for convenience. In the cubic case, we define truncated twice the k-NN distance in the marginal direction $\mathrm{x}_j$ by $\xi_{k}^{\mathrm{x}_j}(\mathrm{x})=\min\{\mathrm{x}_j+\epsilon_{k}(\mathrm{x})/2,1\}-\max\{\mathrm{x}_j-\epsilon_{k}(\mathrm{x})/2,0\}$. In the rectangular case, such distance in the marginal direction $\mathrm{x}_j$ is defined by $\zeta_{k}^{\mathrm{x}_j}(\mathrm{x})=\min\{\mathrm{x}_j+\epsilon_{k}^{\mathrm{x}_j}(\mathrm{x})/2,1\}-\max\{\mathrm{x}_j-\epsilon_{k}^{\mathrm{x}_j}(\mathrm{x})/2,0\}$.
	\end{definition}
	
	\begin{definition}[$r$-cell] \label{def4}
		We define the $r$-cell centered at $\mathrm{x}$ by
		$B(\mathrm{x};r) = \{\mathrm{x}'\in \mathbb{R}^d: ||\mathrm{x}'-\mathrm{x}||_\infty<r\}$ in the cubic case, and by $B(\mathrm{x};r_{1:d}) = \bigcap\limits_{j=1}^d\{\mathrm{x}'\in \mathbb{R}^d: |\mathrm{x}'_j-\mathrm{x}_j|<r_j\}$ in the rectangular case.
	\end{definition}
	
	\begin{definition}[Truncated $r$-cell] \label{def5}
		We define the truncated $r$-ball centered at $\mathrm{x}$ by
		$\overline{B}(\mathrm{x};r) = \mathcal{Q}\cap B(\mathrm{x};r)$ in the cubic case, and by $\overline{B}(\mathrm{x};r_{1:d}) = \mathcal{Q}\cap B(\mathrm{x};r_{1:d})$ in the rectangular case.
	\end{definition}
	
	\begin{definition}[Mass function] We define the mass of the cell $B(\mathrm{x};r/2)$ as a function with respect to $r$, which is given by $p_{r}(\mathrm{x})=P(B(\mathrm{x};r/2))$, and define the mass of the cell $B(\mathrm{x};r_{1:d}/2)$ as a function with respect to $r_1,...,r_d$, which is given by $q_{r_1,...,r_d}(\mathrm{x})=P(B(\mathrm{x};r_{1:d}/2))$.
	\end{definition}
	
	\begin{assumption}\label{assumption1}
		We make the following assumptions:
		\begin{enumerate}
			\item[(a)] $p$ is continuous and supported on $\mathcal{Q}$;
			\item[(b)] $p$ is bounded away from 0, i.e., $C_1 = \inf\limits_{\mathrm{x}\in\mathcal{Q}}p(\mathrm{x})>0$;
			\item[(c)] The gradient of $p$ is uniformly bounded on ${\mathcal{Q}^o}$, i.e., $C_2 = \sup\limits_{\mathrm{x}\in\mathcal{Q}^o}||\triangledown p(\mathrm{x})||_1<\infty$.
		\end{enumerate}
	\end{assumption}
	
	\subsection{Preliminary lemmas}\label{sec:lemmas}

	Here, we present some lemmas that support the proofs of the main results. 
	
	\begin{lemma}[\cite{kraskov2004estimating}]\label{lemma1}
		The expectation of $\log p_{\epsilon_k}(\mathrm{x})$ satisfies
		$$\mathbb{E}[\log p_{\epsilon_k}(\mathrm{x})]=\psi(k)-\psi(N).$$
	\end{lemma}
	
	\begin{lemma}\label{lemma2}
		Let $\widetilde{P}$ be the probability measure of a uniform distribution supported on a $d$-dimensional (hyper-)cubic area $S:=B(\mathrm{x};l/2)$, and $\widetilde{p}(\mathrm{x})=\frac{1}{l^d}, \mathrm{x}\in S$  be the density function. Define $\widetilde{q}_{r_1,...,r_d}(\mathrm{x})=\widetilde{P}(B(\mathrm{x};r_1/2,...,r_d/2))$ and $\widetilde{p}_{r}(\mathrm{x})=\widetilde{P}(B(\mathrm{x};r/2))$. Then, we have 
		$$\mathbb{E}[\log \widetilde{q}_{\epsilon_k^{\mathrm{x}_1},...,\epsilon_k^{\mathrm{x}_d}}(\mathrm{x})]=\psi(k)-\frac{d-1}{k}-\psi(N),$$
		where $\epsilon_k^{\mathrm{x}_j}, j=1,...,d$ are defined as Definition~\ref{def2} after replacing $P$ by $\widetilde{P}$.
	\end{lemma}
	\begin{proof}
		The probability density function  for $(\epsilon_k^{\mathrm{x}_1},...,\epsilon_k^{\mathrm{x}_d})$ is given by, 
		\begin{equation}
		\begin{aligned}
		f_{N,k}(r_1,...,r_d) &= \frac{(N-1)!}{k!(N-k-1)!}\times \frac{\partial^d(\widetilde{q}_{r_1,...,r_d}^k)}{\partial r_1\cdots\partial r_d} \times (1-\widetilde{p}_{r_{\mathrm{m}}})^{N-k-1},
		\end{aligned}
		\end{equation}
		where $\widetilde{p}_{r}=\widetilde{P}(B(\mathrm{x};r/2))$, and $r_{\mathrm{m}}=\max\limits_{1\leq j \leq d}r_j$ \cite{kraskov2004estimating}. Then we have
		\begin{equation}
		\begin{aligned}
		&\mathbb{E}[\log \widetilde{q}_{\epsilon_k^{\mathrm{x}_1},...,\epsilon_k^{\mathrm{x}_d}}(\mathrm{x})]=\int_{0}^{l}\cdots\int_{0}^{l}
		\left(\begin{matrix}
		N-1\\k
		\end{matrix}\right)
		\cdot \frac{\partial^d(\widetilde{q}_{r_1,...,r_d}^k)}{\partial r_1\cdots\partial r_d} \cdot (1-\widetilde{p}_{r_{\mathrm{m}}})^{N-k-1}\log \widetilde{q}_{r_1,...,r_d} dr_1\cdots dr_d\\
		&=\int_{0}^{l}\cdots\int_{0}^{l}
		\left(\begin{matrix}
		N-1\\k
		\end{matrix}\right)
		\cdot \frac{\partial^d\big((\frac{1}{l^d}r_1\cdots r_d)^k\big)}{\partial r_1\cdots\partial r_d} \cdot (1-\frac{1}{l^d}{r_{\mathrm{m}}^d})^{N-k-1}\log (\frac{1}{l^d}r_1\cdots r_d) dr_1\cdots dr_d\\
		&=\left(\begin{matrix}
		N-1\\k
		\end{matrix}\right)
		k^d\frac{1}{l^d}\int_{0}^{l}\cdots\int_{0}^{l}(\frac{1}{l^d}r_1\cdots r_d)^{k-1}(1-\frac{1}{l^d}{r_{\mathrm{m}}^d})^{N-k-1}\log (\frac{1}{l^d}r_1\cdots r_d) dr_1\cdots dr_d\\
		&=\left(\begin{matrix}
		N-1\\k
		\end{matrix}\right)
		k^d\int_{0}^{1}\cdots\int_{0}^{1}(u_1\cdots u_d)^{k-1}(1-{u_{\mathrm{m}}^d})^{N-k-1}\log(u_1\cdots u_d)du_1\cdots du_d,
		\end{aligned}
		\end{equation}
		where the last equality comes from the change of variables $u_i=\frac{1}{l}r_i, i=1,...,d$. Note that the integrand is symmetric under a permutation of the labels $1,...,d$, and so we have
		\begin{equation}
		\begin{aligned}\label{eq1}
		&\mathbb{E}[\log \widetilde{q}_{\epsilon_k^{\mathrm{x}_1},...,\epsilon_k^{\mathrm{x}_d}}(\mathrm{x})]\\
		=&dk^d
		\left(\begin{matrix}
		N-1\\k
		\end{matrix}\right)
		\int_{0}^{1}du_d \bigg(u_d^{k-1}(1-u_{d}^d)^{N-k-1}\int_{0}^{u_d}\cdots\int_{0}^{u_d}(u_1\cdots u_{d-1})^{k-1}\log(u_1\cdots u_d)du_1\cdots du_{d-1}\bigg)
		\end{aligned}
		\end{equation}
		Computing the integral over $u_1,...,u_{d-1}$ using the symmetry again, we obtain
		\begin{equation}\label{eq2}
		\begin{aligned}
		&\int_{0}^{u_d}\cdots\int_{0}^{u_d}(u_1\cdots u_{d-1})^{k-1}\log(u_1\cdots u_d)du_1\cdots du_{d-1}\\
		=&(d-1)\int_{0}^{u_d}\cdots\int_{0}^{u_d}(u_1\cdots u_{d-1})^{k-1}\log u_1 du_1\cdots du_{d-1}+\log u_m\int_{0}^{u_d}\cdots\int_{0}^{u_d}(u_1\cdots u_{d-1})^{k-1} du_1\cdots du_{d-1}\\
		=& I_1+I_2,
		\end{aligned}
		\end{equation}
		where $I_1$ is the first term and $I_2$ is the second term. By basic calculus, we have 
		\begin{equation}
		\begin{aligned}
		I_1 &= (d-1)\int_{0}^{u_d}u_1^{k-1}\log u_1 du_1 \bigg(\int_{0}^{u_d}(u_2)^{k-1}du_2\bigg)^{d-2}\\
		&=(d-1)\big(\frac{1}{k}u_d^k\big)^{d-1}(\log u_d-\frac{1}{k}),
		\end{aligned}
		\end{equation}
		and
		\begin{equation}
		\begin{aligned}
		I_2 & =\log u_d (\frac{1}{k}u_d^k\big)^{d-1},
		\end{aligned}
		\end{equation}
		which yield $I_1+I_2=(\frac{1}{k}u_d^k\big)^{d-1}\big(d\log u_d-\frac{d-1}{k}\big)$. Plug this into Eq~\eqref{eq1} and change the variables by $t=u_d^d$, and we finally have
		\begin{equation}
		\begin{aligned}
		&\mathbb{E}[\log \widetilde{q}_{\epsilon_k^{\mathrm{x}_1},...,\epsilon_k^{\mathrm{x}_d}}(\mathrm{x})]\\
		=&dk\left(\begin{matrix}
		N-1\\k
		\end{matrix}\right)
		\int_{0}^{1}u_d^{kd-1}(1-u_{d}^d)^{N-k-1}\big(d\log u_d-\frac{d-1}{k}\big)du_d\\
		=&k\left(\begin{matrix}
		N-1\\k
		\end{matrix}\right)
		\int_{0}^{1}t^{k-1}(1-t)^{N-k-1}\big(\log t -\frac{d-1}{k}\big)dt\\
		=& \psi(k)-\frac{d-1}{k}-\psi(N).
		\end{aligned}
		\end{equation}
	\end{proof}

	\begin{lemma}[Lemma 3 in \cite{singh2016finite}]\label{lemma3}
		Suppose $p$ satisfies Assumption (a) and (b). Then, for any $\mathrm{x}\in \mathcal{Q}$ and $r>\big(\frac{k}{C_1 N}\big)^{1/d}$, we have
		$$\mathbb{P}(\epsilon_k(\mathrm{x})>r)\leq e^{-C_1 r^d N}\big(\frac{e C_1 r^d N}{k}\big)^k.$$
	\end{lemma}
	
	\begin{lemma}[Lemma 4 in \cite{singh2016finite}]\label{lemma4}
		Suppose $p$ satisfies Assumption (a) and (b). Then, for any $\mathrm{x}\in \mathcal{Q}$ and $\alpha>0$, we have
		$$\mathbb{E}[\epsilon_k^\alpha(\mathrm{x})]\leq (1+\frac{\alpha}{d})\big(\frac{k}{C_1 N}\big)^\frac{\alpha}{d}.$$
	\end{lemma}
	
	\begin{lemma}\label{lemma5}
		Suppose $p$ satisfies Assumption~\ref{assumption1}, then, for any $\mathrm{x}\in \mathcal{Q}$ and array $(r_1,...,r_d)$ that satisfy 
		$$\bigg\{
		\begin{matrix}
		\mathrm{x}_j+\frac{r_j}{2}\leq 1, if~ \mathrm{x}_j\leq\frac{1}{2}\\
		\mathrm{x}_j-\frac{r_j}{2}\geq 0, if~ \mathrm{x}_j>\frac{1}{2}
		\end{matrix}
		\bigg.
		$$
		for $j=1,...,d$, we have
		\begin{equation*}
		\bigg|\frac{\partial^d{q}_{r_1,...,r_d}(\mathrm{x})}{\partial r_{1}\cdots\partial r_{d}}-
		\frac{1}{2^{\sum_{j=1}^{d}\mathds{1}_j}}p(\mathrm{x})\bigg|\leq 
		\frac{1}{2^{\sum_{j=1}^{d}\mathds{1}_j+1}}C_2r_{\mathrm{m}},
		\end{equation*}
		and 
		\begin{equation*}
		\bigg|\frac{\partial^u{q}_{r_1,...,r_d}(\mathrm{x})}{\partial r_{1}\cdots\partial r_{u}}-\frac{1}{2^{\sum_{j=1}^{u}\mathds{1}_j}}p(\mathrm{x})\mu\big(\overline{B}(\mathrm{x}_{u+1:d};\frac{r_{u+1}}{2},...,\frac{r_d}{2})\big)\bigg|\leq \frac{1}{2^{\sum_{j=1}^{u}\mathds{1}_j+1}}C_2r_{\mathrm{m}}\mu\big(\overline{B}(\mathrm{x}_{u+1:d};\frac{r_{u+1}}{2},...,\frac{r_d}{2})\big),
		\end{equation*}
		where $u < d$, $r_{\mathrm{m}}=\max\limits_{1\leq j\leq d}r_j$ and $\mathds{1}_j$ is the indicator function admitting the value 1 if $[\mathrm{x}_j-\frac{r_j}{2},\mathrm{x}_j+\frac{r_j}{2}]$ intersects $[0,1]$ and  0 otherwisely.
	\end{lemma}
	\begin{proof}
		For the sake of convenience, we only discuss the case when $\mathrm{x}\in [0,\frac{1}{2}]^d$ and $\mathds{1}_j=1$ for $j=1,...,n \leq u$. The proof for other cases can be obtained by permuting the labels $1,...,d$. By the definition of $	{q}_{r_1,...,r_d}(\mathrm{x})$, we have
		\begin{equation}
		\begin{aligned}
		{q}_{r_1,...,r_d}(\mathrm{x})&=\int_{\mathrm{x}_1-r_1/2}^{\mathrm{x}_1+r_1/2}\cdots \int_{\mathrm{x}_d-r_d/2}^{\mathrm{x}_d+r_d/2}p(\mathrm{x}'_1,...,\mathrm{x}'_d)d\mathrm{x}'_d\cdots d\mathrm{x}'_1\\
		&=\int_{0}^{\mathrm{x}_1+r_1/2}\cdots\int_{0}^{\mathrm{x}_n+r_n/2}\int_{\mathrm{x}_{n+1}-\frac{r_{n+1}}{2}}^{\mathrm{x}_{n+1}+\frac{r_{n+1}}{2}}\cdots \int_{\mathrm{x}_d-r_d/2}^{\mathrm{x}_d+r_d/2}p(\mathrm{x}'_1,...,\mathrm{x}'_d)d\mathrm{x}'_d\cdots d\mathrm{x}'_1,
		\end{aligned}
		\end{equation}
		and the partial derivative of it with respect to the first $n$ variables is given by
		\begin{equation}
		\begin{aligned}
		&\frac{\partial^n{q}_{r_1,...,r_d}(\mathrm{x})}{\partial r_{1}\cdots\partial r_{n}}\\
		=&\frac{1}{2^n}\int_{\mathrm{x}_{n+1}-\frac{r_{n+1}}{2}}^{\mathrm{x}_{n+1}+\frac{r_{n+1}}{2}}\cdots \int_{\mathrm{x}_d-r_d/2}^{\mathrm{x}_d+r_d/2} p(\mathrm{x}_1+\frac{r_1}{2},...,\mathrm{x}_n+\frac{r_n}{2},\mathrm{x}'_{n+1},...,\mathrm{x}'_d)d\mathrm{x}'_d\cdots d\mathrm{x}'_{n+1}.
		\end{aligned}
		\end{equation}
		Next we obtain the partial derivative of ${q}_{r_1,...,r_d}(\mathrm{x})$ with respect to the first $u$ variables
		\begin{equation}
		\begin{aligned}
		&\frac{\partial^u{q}_{r_1,...,r_d}(\mathrm{x})}{\partial r_{1}\cdots\partial r_{u}}\\
		=&\frac{1}{2^{u}}\int_{\mathrm{x}_{u+1}-r_{u+1}/2}^{\mathrm{x}_{u+1}+r_{u+1}/2}\cdots \int_{\mathrm{x}_d-r_d/2}^{\mathrm{x}_d+r_d/2}
		p(\mathrm{x}_1+\frac{r_1}{2},...,\mathrm{x}_n+\frac{r_n}{2},\mathrm{x}_{n+1}\pm\frac{r_{n+1}}{2},...,\mathrm{x}_u\pm\frac{r_u}{2},\mathrm{x}'_{u+1},..., \mathrm{x}'_d)d\mathrm{x}'_{u+1}\cdots d\mathrm{x}'_d\\
		=&\frac{1}{2^{u}}\int_{\overline{B}(\mathrm{x}_{u+1:d};\frac{r_{u+1}}{2},...,\frac{r_d}{2})}p(\mathrm{x}_1+\frac{r_1}{2},...,\mathrm{x}_n+\frac{r_n}{2},\mathrm{x}_{n+1}\pm\frac{r_{n+1}}{2},...,\mathrm{x}_u\pm\frac{r_u}{2},\mathrm{x}'_{u+1},..., \mathrm{x}'_d)d\mathrm{x}'_{u+1}\cdots d\mathrm{x}'_d,
		\end{aligned}
		\end{equation}
		where the notation $p(...,x\pm \frac{r}{2},...)=p(...,x+ \frac{r}{2},...)+p(...,x- \frac{r}{2},...)$.
		
		Finally, we have
		\begin{equation}
		\begin{aligned}
		&\bigg|\frac{\partial^u{q}_{r_1,...,r_d}(\mathrm{x})}{\partial r_{1}\cdots\partial r_{u}}-\frac{1}{2^{\sum_{j=1}^{u}\mathds{1}_j}}p(\mathrm{x})\mu\big(\overline{B}(\mathrm{x}_{u+1:d};\frac{r_{u+1}}{2},...,\frac{r_d}{2})\big)\bigg|\\
		\leq&\frac{1}{2^{u}}\int_{\overline{B}(\mathrm{x}_{u+1:d};\frac{r_{u+1}}{2},...,\frac{r_d}{2})}\bigg|p(\mathrm{x}_1+\frac{r_1}{2},...,\mathrm{x}_n+\frac{r_n}{2},\mathrm{x}_{n+1}\pm\frac{r_{n+1}}{2},...,\mathrm{x}_u\pm\frac{r_u}{2},\mathrm{x}'_{u+1},..., \mathrm{x}'_d)\\
		&~~~~~~~~~~~~~~~~~~~~~~~~~~~~~~~~~~~~~~~~~~~~~~~~~~~~~~~~~~~~~~~~~~~~~~~~~~~~~~~~~~~~~~~~~~~~~~~~~~~~~-2^{u-n}p(\mathrm{x})\bigg|d\mathrm{x}'_{u+1}\cdots d\mathrm{x}'_d\\
		\leq & \frac{2^{u-n}}{2^u}\int_{\overline{B}(\mathrm{x}_{u+1:d};\frac{r_{u+1}}{2},...,\frac{r_d}{2})}C_2 \frac{r_{\mathrm{m}}}{2}d\mathrm{x}'_{u+1}\cdots d\mathrm{x}'_d\\
		= & \frac{1}{2^{n+1}}C_2 r_{\mathrm{m}}\mu\big(\overline{B}(\mathrm{x}_{u+1:d};\frac{r_{u+1}}{2},...,\frac{r_d}{2})\big),
		\end{aligned}
		\end{equation}
		which completes the proof for $u<d$.
		
		Particularly, we have
		\begin{equation}
		\begin{aligned}
		\bigg|\frac{\partial^d{q}_{r_1,...,r_d}(\mathrm{x})}{\partial r_{1}\cdots\partial r_{d}}-\frac{1}{2^{\sum_{j=1}^{d}\mathds{1}_j}}p(\mathrm{x})\bigg|\leq \frac{1}{2^{\sum_{j=1}^{d}\mathds{1}_j+1}}C_2r_{\mathrm{m}}.
		\end{aligned}
		\end{equation}
		
	\end{proof}
	
	\begin{lemma}\label{lemma6}
		Suppose $p$ satisfies Assumption~\ref{assumption1}, then, for any $\mathrm{x}\in \mathcal{Q}$ and $r$ that satisfy 
		$$\bigg\{
		\begin{matrix}
		\mathrm{x}_j+\frac{r}{2}\leq 1, if~ \mathrm{x}\leq\frac{1}{2}\\
		\mathrm{x}_j-\frac{r}{2}\geq 0, if~ \mathrm{x}>\frac{1}{2}
		\end{matrix}
		\bigg.
		$$
		for $j=1,...,d$, we have
		\begin{equation*}
		\bigg|p_r(\mathrm{x})-p(\mathrm{x})\mu\big(\overline{B}(\mathrm{x};\frac{r}{2})\big)\bigg|\leq C_2\frac{r}{2}\overline{B}(\mathrm{x};\frac{r}{2}),
		\end{equation*}
		and 
		\begin{equation*}
		\bigg|\frac{d p_r(\mathrm{x})}{d r}-\sum_{j=1}^{d}\frac{1}{2^{\mathds{1}_j}}p(\mathrm{x})\mu\big(\overline{B}(\mathrm{x}_{\hat{j}};\frac{r}{2})\big)\bigg|\leq \sum_{j=1}^{d}\frac{1}{2^{\mathds{1}_j+1}}C_2r\mu\big(\overline{B}(\mathrm{x}_{\hat{j}};\frac{r}{2})\big),
		\end{equation*}
		where $m< d$ and $\mathds{1}_j$ is the indicator function admitting the value 1 if $[\mathrm{x}_j-\frac{r}{2},\mathrm{x}_j+\frac{r}{2}]$ intersects $[0,1]$ and  0 otherwiesly.
	\end{lemma}
	
	\begin{proof}
		By the definition of ${p}_{r}(\mathrm{x})$, we have
		\begin{equation}
		\begin{aligned}
		{p}_{r}(\mathrm{x})&=\int_{\overline{B}(\mathrm{x};\frac{r}{2})}p(\mathrm{x}'_1,...,\mathrm{x}'_d)d\mathrm{x}'_d\cdots d\mathrm{x}'_1.
		\end{aligned}
		\end{equation}
		It then follows that, 
		\begin{equation}
		\begin{aligned}
		&\bigg|p_r(\mathrm{x})-p(\mathrm{x})\mu\big(\overline{B}(\mathrm{x};\frac{r}{2})\big)\bigg|\\
		\leq&\int_{\overline{B}(\mathrm{x};\frac{r}{2})}\big|p(\mathrm{x}'_1,...,\mathrm{x}'_d)-p(\mathrm{x})\big|d\mathrm{x}'_d\cdots d\mathrm{x}'_1\\
		\leq&\int_{\overline{B}(\mathrm{x};\frac{r}{2})}
		C_2\frac{r}{2}d\mathrm{x}'_d\cdots d\mathrm{x}'_1\\
		=&C_2\frac{r}{2}\overline{B}(\mathrm{x};\frac{r}{2}),
		\end{aligned}
		\end{equation}
		which completes proof of the first inequality. For the second inequality, one can easily see that 
		\begin{equation}
		p_r(\mathrm{x})=q_{r,...,r}(\mathrm{x}).
		\end{equation}
		Now using Lemma~\ref{lemma5}, we obtain
		\begin{equation}
		\begin{aligned}
		&\bigg|\frac{d p_r(\mathrm{x})}{d r}-\sum_{j=1}^{d}\frac{1}{2^{\mathds{1}_j}}p(\mathrm{x})\mu\big(\overline{B}(\mathrm{x}_{\hat{j}};\frac{r}{2})\big)\bigg|\\
		\leq&
		\sum_{j=1}^{d}\bigg|\frac{\partial q_{r_1,...,r_d}(\mathrm{x})}{\partial r_j}\Big|_{r_{1:d}=r}\Big.-\frac{1}{2^{\mathds{1}_j}}p(\mathrm{x})\mu\big(\overline{B}(\mathrm{x}_{\hat{j}};\frac{r}{2})\big)\bigg|\\
		\leq&\sum_{j=1}^{d}\frac{1}{2^{\mathds{1}_j+1}}C_2r\mu\big(\overline{B}(\mathrm{x}_{\hat{j}};\frac{r}{2})\big).
		\end{aligned}
		\end{equation}
	\end{proof}

	\subsection{Proof of bias bound for the truncated KL estimator}\label{biastkl}
	
	\begin{proof}
		Note that $\sum_{j=1}^{d}\log\xi_{i,j}$ are identically distributed, then we have
		\begin{equation}
		\begin{aligned}
		\mathbb{E}[\widehat{H}_{tKL}(X)]&=-\psi(k)+\psi(N) + \frac{1}{N}\sum_{i=1}^{N}\mathbb{E}\big[\sum_{j=1}^{d}\log\xi_{i,j}\big]\\
		&=-\psi(k)+\psi(N) + \mathbb{E}\big[\sum_{j=1}^{d}\log\xi_{k}^{\mathrm{x}_j}(\mathrm{x})\big]\\
		&=-\mathbb{E}[\log p_{\epsilon_k}(\mathrm{x})]+\mathbb{{E}}[\log\mu({B}(\mathrm{x};\xi_{k}^{\mathrm{x}_1}/2,...,\xi_{k}^{\mathrm{x}_d}/2))]\\
		&=-\mathbb{E}\big[\log\frac{P(B(\mathrm{x};\epsilon_k/2))}{\mu({B}(\mathrm{x};\xi_{k}^{\mathrm{x}_1}/2,...,\xi_{k}^{\mathrm{x}_d}/2))}\big]\\
		&=-\mathbb{E}\big[\log\frac{P(\overline{B}(\mathrm{x};\epsilon_k/2))}{\mu(\overline{B}(\mathrm{x};\epsilon_k/2))}\big],
		\end{aligned}
		\end{equation}
		where the third equality is from Lemma~\ref{lemma1} and the fifth equality is due to the fact that $p$ is supported on $\mathcal{Q}$. Note that 
		\begin{equation}
		C_1\leq \frac{P(\overline{B}(\mathrm{x};\epsilon_k/2))}{\mu(\overline{B}(\mathrm{x};\epsilon_k/2))} \leq \sup\limits_{\mathrm{x}\in \mathcal{Q}}p(\mathrm{x})<\infty,
		\end{equation}
		and we have
		\begin{equation}
		\begin{aligned}
		&\bigg|\log p(\mathrm{x})-\log\frac{P(\overline{B}(\mathrm{x};\epsilon_k/2))}{\mu(\overline{B}(\mathrm{x};\epsilon_k/2))}\bigg|\\
		\leq&\frac{1}{C_1}\bigg|p(\mathrm{x})-\frac{P(\overline{B}(\mathrm{x};\epsilon_k/2))}{\mu(\overline{B}(\mathrm{x};\epsilon_k/2))}\bigg|\\
		\leq&\frac{1}{C_1\mu(\overline{B}(\mathrm{x};\epsilon_k/2))}\int_{\overline{B}(\mathrm{x};\epsilon_k/2)}|p(\mathrm{x})-p(\mathrm{x'})| d\mathrm{x'}\\
		\leq&\frac{1}{C_1\mu(\overline{B}(\mathrm{x};\epsilon_k/2))}\int_{\overline{B}(\mathrm{x};\epsilon_k/2)}C_2||\mathrm{x}-\mathrm{x'}||_\infty d\mathrm{x'}\\
		\leq&\frac{C_2}{2C_1}\epsilon_k.
		\end{aligned}
		\end{equation}
		Finally, using Lemma~\ref{lemma4}, the bias bound of $\mathbb{E}[\widehat{H}_{tKL}(X)]$ can be obtained by
		\begin{equation}
		\begin{aligned}
		&\big|\mathbb{E}[\widehat{H}_{tKL}(X)]-H(X)\big|\\
		\leq&\underset{\mathrm{x} \sim p}{\mathbb{E}}{\mathbb{E}}\big[\bigg|\log p(\mathrm{x})-\log\frac{P(\overline{B}(\mathrm{x};\epsilon_k/2))}{\mu(\overline{B}(\mathrm{x};\epsilon_k/2))}\bigg|\big]\\
		\leq&\frac{C_2}{2C_1}\underset{\mathrm{x} \sim p}{\mathbb{E}}\mathbb{E}[\epsilon_k]\\
		\leq&\frac{C_2}{C_1^{1+1/d}}\big(\frac{k}{ N}\big)^\frac{1}{d},
		\end{aligned}
		\end{equation}
		which completes the proof.
	\end{proof}
	
	\subsection{Proof of variance bound for the truncated KL estimator}\label{vartkl}
	
	\begin{proof}
		For the sake of convenience, we define  $\alpha_i=\sum_{j=1}^{d}\log\xi_{i,j}$. We 
		then define $\alpha'_i, i=1,...,N$ as the estimators after $\mathrm{x}^{(1)}$ is resampled and $\alpha^*_i, i=2,...,N$ as the estimators after $\mathrm{x}^{(1)}$ is removed. Then, by the Efron-Stein inequality \cite{efron1981jackknife},
		\begin{equation}
		\begin{aligned}
		\mathrm{Var}[\widehat{H}_{tKL}(X)]&=\mathrm{Var}\bigg[\frac{1}{N}\sum_{i=1}^{N}\alpha_i\bigg]\\
		&\leq \frac{N}{2}\mathbb{E}\bigg[\bigg(\frac{1}{N}\sum_{i=1}^{N}\alpha_i-\frac{1}{N}\sum_{i=1}^{N}\alpha'_i\bigg)^2\bigg]\\
		&\leq N \mathbb{E}\bigg[\bigg(\frac{1}{N}\sum_{i=1}^{N}\alpha_i-\frac{1}{N}\sum_{i=2}^{N}\alpha^*_i\bigg)^2+\bigg(\frac{1}{N}\sum_{i=1}^{N}\alpha'_i-\frac{1}{N}\sum_{i=2}^{N}\alpha^*_i\bigg)^2\bigg]\\
		& = 2 N \mathbb{E}\bigg[\bigg(\frac{1}{N}\sum_{i=1}^{N}\alpha_i-\frac{1}{N}\sum_{i=2}^{N}\alpha^*_i\bigg)^2\bigg].
		\end{aligned}
		\end{equation}
		Let $\mathds{1}_{E_i}$ be the indicator function of the event $E_i=\{\epsilon_k(\mathrm{x}^{(1)})\neq\epsilon^*_k(\mathrm{x}^{(1)})\}$, where $\epsilon^*_k(\mathrm{x}^{(1)})$ is twice the $k$-NN distance of $\mathrm{x}^{(1)}$ when $\alpha^*_i$ are used. Then,
		\begin{equation}
		N\bigg(\frac{1}{N}\sum_{i=1}^{N}\alpha_i-\frac{1}{N}\sum_{i=2}^{N}\alpha^*_i\bigg) = \alpha_1+\sum_{i=2}^{N}\mathds{1}_{E_i}(\alpha_i-\alpha^*_i).
		\end{equation}
		By Cauchy-Schwarz inequality, we have
		\begin{equation}
		\begin{aligned}
		N^2\bigg(\frac{1}{N}\sum_{i=1}^{N}\alpha_i-\frac{1}{N}\sum_{i=2}^{N}\alpha^*_i\bigg)^2 &\leq \bigg(1+\sum_{i=2}^{N}\mathds{1}_{E_i}\bigg)\bigg(\alpha_1^2+\sum_{i=2}^{N}\mathds{1}_{E_i}(\alpha_i-\alpha^*_i)^2\bigg)\\
		&\leq (1+C_{k,d})\bigg(\alpha_1^2+\sum_{i=2}^{N}\mathds{1}_{E_i}(\alpha_i-\alpha^*_i)^2\bigg)\\
		&\leq (1+C_{k,d})\bigg(\alpha_1^2+2\sum_{i=2}^{N}\mathds{1}_{E_i}(\alpha_i^2+\alpha_i^{*2})\bigg),
		\end{aligned}
		\end{equation}
		where $C_{k,d}$ is the constant such that $\mathrm{x}^1$ is amongst the $k$-nearest neighbors of at most $C_{k,d}$ other samples. Note that $\alpha_i$ and $\alpha_i^*$ are identically distributed, we only need to bound
		\begin{subequations}
			\begin{equation}\label{eq:b1}
			\mathbb{E}[\alpha_1^2],
			\end{equation}
			\begin{equation}\label{eq:b2}
			(N-1)\mathbb{E}[\mathds{1}_{E_2}\alpha_2^2],
			\end{equation}
			\begin{equation}\label{eq:b3}
			(N-1)\mathbb{E}[\mathds{1}_{E_2}\alpha_2^{*2}].
			\end{equation}
		\end{subequations}
		
		\textit{Bound of \eqref{eq:b1}}:
		
		We separate \eqref{eq:b1} into two parts,
		\begin{equation}\label{eq15}
		\begin{aligned}
		\mathbb{E}\big[\alpha_1^2\big]=\underset{\mathrm{x}\in \mathcal{Q}}{\mathbb{E}}\underset{P:\epsilon_k< a_N}{\mathbb{E}}\big[\alpha_1^2\big]
		+
		\underset{\mathrm{x}\in \mathcal{Q}}{\mathbb{E}}\underset{P:\epsilon_k\geq a_N}{\mathbb{E}}\big[\alpha_1^2\big],
		\end{aligned}
		\end{equation}
		where $a_N=\big(\frac{2k\log N}{C_1N}\big)^{\frac{1}{d}}$.
		
		First, we consider the bound of the first term in Eq~\eqref{eq15}.
		For any $\mathrm{x}\in \mathcal{Q}$,
		\begin{equation}\label{eq16}
		\begin{aligned}
		&\underset{P:\epsilon_k< a_N}{\mathbb{E}}\big[\alpha_1^2\big]\\
		=&\int_{0}^{a_N}f_{N,k}(r)
		\big[\log \big(\xi_k^{\mathrm{x}_1}\cdots\xi_k^{\mathrm{x}_d}\big)\big]^2 dr.\\
		\end{aligned}
		\end{equation}
		where $f_{N,k}(r)=k \left(\begin{matrix}
		N-1\\k
		\end{matrix}\right)
		\cdot \frac{d{p}_{r}}{dr}\cdot p_r^{k-1} \cdot (1-{p}_{r})^{N-k-1}$ \cite{kraskov2004estimating}.
		Note that for sufficiently large $N$,
		\begin{equation}\label{eq17}
		\begin{aligned}
		&\int_{0}^{a_N}
		[\log \big(\xi_k^{\mathrm{x}_1}\cdots\xi_k^{\mathrm{x}_d}\big)\big]^2
		dr\\
		\leq&
		\int_{0}^{a_N}
		\big[	\log \big(\frac{r}{2}\cdots\frac{r}{2}\big)\big]^2
		dr\\
		\leq&C_3 \frac{(\log N)^3}{N^{1/d}},
		\end{aligned}
		\end{equation}
		for some $C_3>0$,
		we now focus on bounding $f_{N,k}(r)$. By basic calculus, we can see that 
		\begin{equation}
		k \left(\begin{matrix}
		N-1\\k
		\end{matrix}\right)
		\cdot p_r^{k-1} \cdot (1-{p}_{r})^{N-k-1}\leq C_4 N,
		\end{equation}
		for some $C_4>0$ and $p_r\in (0,1)$. Also, by Lemma~\ref{lemma6}, we have $\frac{d{p}_{r}}{dr}\leq C_5\frac{\log N}{N}$ for some $C_5>0$ and $r<a_N$. Therefore, the pdf term can be bounded by 
		\begin{equation}\label{eq18}
		f_{N,k}(r)\leq C_4 C_5\log N.
		\end{equation}
		Combining Eq~\eqref{eq17} and Eq~\eqref{eq18}, we can bound Eq~\eqref{eq16} by:
		\begin{equation}
		\underset{P:\epsilon_k< a_N}{\mathbb{E}}\big[\alpha_1^2\big]\leq C_3C_4C_5\frac{(\log N)^4}{N^{1/d}}\leq C_6,
		\end{equation}
		for some $C_6>0$.
		Thus, the first term in Eq~\eqref{eq15} is bounded by 
		\begin{equation}\label{eq21}
		\underset{\mathrm{x}\in \mathcal{Q}}{\mathbb{E}}\underset{P:\epsilon_k< a_N}{\mathbb{E}}\big[\alpha_1^2\big]
		\leq  C_6.
		\end{equation}
		
		Now we consider the second term in Eq~\eqref{eq15}. For $\epsilon_k\geq a_N$ and sufficiently large $N$, we have
		\begin{equation}\label{eq19}
		\begin{aligned}
		\big[\log\big({\xi_k^{\mathrm{x}_1}\cdots \xi_k^{\mathrm{x}_d}}\big)\big]^2
		&\leq \big[\log\big({\epsilon_k/2\cdots \epsilon_k}/2\big)\big]^2\\
		&\leq d^2\big[\log \big(\frac{a_N}{2}\big)\big]^2\\
		&\leq C_7 (\log N)^2,
		\end{aligned}
		\end{equation}
		for some $C_7>0$. Using Lemma~\ref{lemma3} and Eq~\eqref{eq19}, the second term in Eq~\eqref{eq15} can be bounded by
		\begin{equation}\label{eq35}
		\begin{aligned}
		\underset{\mathrm{x}\in \mathcal{Q}}{\mathbb{E}}\underset{P:\epsilon_k\geq a_N}{\mathbb{E}}\big[\alpha_1^2\big]
		&= \underset{\mathrm{x}\in \mathcal{Q}}{\mathbb{E}}\underset{P:\epsilon_k\geq a_N}{\mathbb{E}}\bigg[\big[\log\big({\xi_k^{\mathrm{x}_1}\cdots \xi_k^{\mathrm{x}_d}}\big)\big]^2\bigg]\\
		&\leq C_7(\log N)^2\cdot P(\epsilon_k\geq a_N)\\
		&\leq C_8 \frac{(\log N)^{k+2}}{N^{2k}},
		\end{aligned}
		\end{equation}
		for some $C_8>0$.
		
		Combining Eq~\eqref{eq21} and Eq~\eqref{eq35}, the expectation of $\alpha_1^2$ is bounded by 
		\begin{equation}\label{eq36}
		\mathbb{E}[\alpha_1^2]\leq C_9,
		\end{equation}
		for some $C_9>0$.
		
		\medskip

		\textit{Bound of \eqref{eq:b2}}:
		
		Since the event $E_2$ is equivalent to the event that $\mathrm{x}^{(1)}$ is amongst the $k$-NN of $\mathrm{x}^{(2)}$, $\mathbb{E}[\mathds{1}_{E_2}]=\mathbb{P}\{\mathrm{x}^{(1)} \in B(\mathrm{x}^{(2)}; \epsilon_k(\mathrm{x}^{(2)})\}=\frac{k}{N-1}$. Additionally, since $E_2$ is independent of $\epsilon_k(\mathrm{x}^{(2)})$, \eqref{eq:b2} is therefore bounded as
		\begin{equation}
		(N-1)\mathbb{E}[\mathds{1}_{E_2}\alpha_2^2]\leq (N-1)\mathbb{E}[\mathds{1}_{E_2}]\mathbb{E}[\alpha_2^2]\leq kC_9 ,
		\end{equation}
		where the second inequality is from Eq~\eqref{eq36}.
		
		\textit{Bound of \eqref{eq:b3}}:
		
		Using the independence between $E_2$ and $\epsilon^*_k(\mathrm{x}^{(2)})$ (twice the $k$-NN distance of $\mathrm{x}^{(2)}$ after $\mathrm{x}^{(1)}$ is removed), we can bound \eqref{eq:b3} as
		\begin{equation}
		(N-1)\mathbb{E}[\mathds{1}_{E_2}\alpha_2^{*2}]\leq (N-1)\mathbb{E}[\mathds{1}_{E_2}]\mathbb{E}[\alpha_2^{*2}]\leq kC_{10},
		\end{equation}
		for some $C_{10}>0$,
		where the second inequality is obtained from Eq~\eqref{eq36} when the sample size is reduced to $N-1$.
		
		Finally we obtain the bound of the variance of $\widehat{H}_{tKL}(X)$
		\begin{equation}
		\mathrm{Var}[\widehat{H}_{tKL}(X)]\leq C_{11} \frac{1}{N},
		\end{equation}
		for some $C_{11}>0$.
	\end{proof}
	
	\subsection{Proof of bias bound for the truncated KSG estimator}\label{biasksg}
	
	\begin{proof}
		We separate the $d$-dimensional unit cube $\mathcal{Q}$ into two subsets, $\mathcal{Q}=\mathcal{Q}_1+\mathcal{Q}_2$, where $\mathcal{Q}_1:= [\frac{a_N}{2}, 1-\frac{a_N}{2}]^d$, $a_N=\big(\frac{2k\log N}{C_1N}\big)^{\frac{1}{d}}$, and $\mathcal{Q}_2=\mathcal{Q}-\mathcal{Q}_1$. 
		Suppose that $\widetilde{P}$, $ \widetilde{p}$, and $\widetilde{q}_{\epsilon_k^{\mathrm{x}_1},...,\epsilon_k^{\mathrm{x}_d}}(\mathrm{x})$ are defined as in Lemma~\ref{lemma2} with $l=p(\mathrm{x})^{-\frac{1}{d}}$, and by Lemma~\ref{lemma2} and the fact that $\sum_{j=1}^{d}\log\zeta_{i,j}$ are identically distributed, we have
		\begin{equation}
		\begin{aligned}
		\mathbb{E}[\widehat{H}_{tKSG}(X)]&= -\psi(k)+\psi(N)+(d-1)/k + \frac{1}{N}\sum_{i=1}^{N}\mathbb{E}\big[\sum_{j=1}^{d}\log\zeta_{i,j}\big]\\
		&=\underset{\mathrm{x}\sim p}{\mathbb{E}}\underset{P}{\mathbb{E}}\big[\log{\zeta_k^{\mathrm{x}_1}\cdots \zeta_k^{\mathrm{x}_d}}\big]
		-
		\underset{\mathrm{x}\sim p}{\mathbb{E}}\underset{\widetilde{P}}{\mathbb{E}}\big[\log{\widetilde{q}_{\epsilon_k^{\mathrm{x}_1},...,\epsilon_k^{\mathrm{x}_d}}}\big]
		\\
		&=\underset{\mathrm{x}\sim p}{\mathbb{E}}\underset{P}{\mathbb{E}}\big[\log{\zeta_k^{\mathrm{x}_1}\cdots \zeta_k^{\mathrm{x}_d}}\big]
		-
		\underset{\mathrm{x}\sim p}{\mathbb{E}}\underset{\widetilde{P}}{\mathbb{E}}\big[\log\big({p(\mathrm{x})\epsilon_k^{\mathrm{x}_1}\cdots\epsilon_k^{\mathrm{x}_d}}\big)\big].
		\end{aligned}
		\end{equation}
		We decompose the bias into three terms and bound them separately: 
		\begin{equation}
		\begin{aligned}
		&\big|\mathbb{E}[\widehat{H}_{tKSG}(X)]-H(X)\big|\\
		=&\bigg|\underset{\mathrm{x}\sim p}{\mathbb{E}}\underset{P}{\mathbb{E}}\big[\log\big({\zeta_k^{\mathrm{x}_1}\cdots \zeta_k^{\mathrm{x}_d}}\big)\big]
		-
		\underset{\mathrm{x}\sim p}{\mathbb{E}}\underset{\widetilde{P}}{\mathbb{E}}\big[\log\big({\epsilon_k^{\mathrm{x}_1}\cdots\epsilon_k^{\mathrm{x}_d}}\big)\big]\bigg|\\
		\leq&I_1+I_2+I_3,
		\end{aligned}
		\end{equation}
		with
		\begin{equation}
		\begin{aligned}
		I_1&=\bigg|\underset{\mathrm{x}\in \mathcal{Q}_2}{\mathbb{E}}\underset{P:\epsilon_k< a_N}{\mathbb{E}}\big[\log\big({\zeta_k^{\mathrm{x}_1}\cdots \zeta_k^{\mathrm{x}_d}}\big)\big]\bigg| 
		+
		\bigg|\underset{\mathrm{x}\in \mathcal{Q}_2}{\mathbb{E}}\underset{\widetilde{P}:\epsilon_k< a_N}{\mathbb{E}}\big[\log\big({\epsilon_k^{\mathrm{x}_1}\cdots\epsilon_k^{\mathrm{x}_d}}\big)\big]\bigg|,\\
		I_2&=\bigg|\underset{\mathrm{x}\in \mathcal{Q}_1}{\mathbb{E}}\underset{P:\epsilon_k<a_N}{\mathbb{E}}\big[\log\big({\zeta_k^{\mathrm{x}_1}\cdots \zeta_k^{\mathrm{x}_d}}\big)\big]
		-
		\underset{\mathrm{x}\in \mathcal{Q}_1}{\mathbb{E}}\underset{\widetilde{P}:\epsilon_k<a_N}{\mathbb{E}}\big[\log\big({\epsilon_k^{\mathrm{x}_1}\cdots\epsilon_k^{\mathrm{x}_d}}\big)\big]\bigg|,\\
		I_3&=\bigg|\underset{\mathrm{x}\in \mathcal{Q}}{\mathbb{E}}\underset{P:\epsilon_k\geq a_N}{\mathbb{E}}\big[\log\big({\zeta_k^{\mathrm{x}_1}\cdots \zeta_k^{\mathrm{x}_d}}\big)\big]
		\bigg| 
		+
		\bigg|
		\underset{\mathrm{x}\in \mathcal{Q}}{\mathbb{E}}\underset{\widetilde{P}:\epsilon_k\geq a_N}{\mathbb{E}}\big[\log\big({\epsilon_k^{\mathrm{x}_1}\cdots\epsilon_k^{\mathrm{x}_d}}\big)\big]\bigg|,
		\end{aligned}
		\end{equation}
		where $\underset{P:\epsilon_k<a_N}{\mathbb{E}}$ means taking expectation under the probability measure $P$ over $\epsilon_k^{\mathrm{x}_j}<a_N, j=1,...,d$.

		\textit{Bound of $I_1$}:
		
		For any $\mathrm{x}\in \mathcal{Q}_2$,
		\begin{equation}\label{eq7}
		\begin{aligned}
		&\underset{P:\epsilon_k< a_N}{\mathbb{E}}\big[\log\big({\zeta_k^{\mathrm{x}_1}\cdots \zeta_k^{\mathrm{x}_d}}\big)\big]\\
		=&\int_{0}^{a_N}\cdots\int_{0}^{a_N}
		f_{N,k}(r_1,...,r_d)\log \big(\zeta_k^{\mathrm{x}_1}\cdots\zeta_k^{\mathrm{x}_d}\big) dr_1\cdots dr_d.\\
		\end{aligned}
		\end{equation}
		where $f_{N,k}(r_1,...,r_d)=\left(\begin{matrix}
		N-1\\k
		\end{matrix}\right)
		\cdot \frac{\partial^d({q}_{r_1,...,r_d}^k)}{\partial r_1\cdots\partial r_d} \cdot (1-{p}_{r_{\mathrm{m}}})^{N-k-1}$, and $r_m = \max\limits_{1\leq j\leq d}r_j$ \cite{kraskov2004estimating}. Note that for sufficiently large $N$, we have, 
		\begin{equation}\label{eq6}
		\begin{aligned}
		&\int_{0}^{a_N}\cdots\int_{0}^{a_N}
		\big|	\log \big(\zeta_k^{\mathrm{x}_1}\cdots\zeta_k^{\mathrm{x}_d}\big)\big|
		dr_1\cdots dr_d\\
		\leq&
		\int_{0}^{a_N}\cdots\int_{0}^{a_N}
		\big|	\log \big(\frac{r_1}{2}\cdots\frac{r_d}{2}\big)\big|
		dr_1\cdots dr_d\\
		\leq&
		\int_{0}^{a_N}\cdots\int_{0}^{a_N}
		\big|\log \big({r_1}\cdots{r_d}\big)\big|dr_1\cdots dr_d+\int_{0}^{a_N}\cdots\int_{0}^{a_N}
		d\log 2 dr_1\cdots dr_d		
		\\
		=&-d(a_N)^{d-1}\int_{0}^{a_N}\log r dr+d\log 2\bigg(\int_{0}^{a_N} dr\bigg)^d\\
		\leq&C_3\frac{\big(\log N\big)^{2}}{C_1 N},
		\end{aligned}
		\end{equation}
		for some $C_3>0$.
		We now focus on bounding $f_{N,k}(r_1,...,r_d)$. We omit the subscripts of ${q}_{r_1,...,r_d}$ for simplicity from now. By the multivariate version of Faà di Bruno's formula \cite{hardy2006combinatorics}, one obtains
		\begin{equation}\label{eq4}
		\begin{aligned}
		\frac{\partial^d({q}^k)}{\partial r_1\cdots\partial r_d}=\sum_{\pi\in \Pi}\frac{d^{|\pi|}q^k}{(d q)^{|\pi|}}\cdot \prod_{B\in \pi}\frac{\partial^{|B|}q}{\prod_{j\in B}\partial r_j},
		\end{aligned}
		\end{equation}
		where $\pi$ runs through the set $\Pi$ of all partitions of the set ${1,...,d}$. By Lemma~\ref{lemma5}, we have
		\begin{equation}
		\begin{aligned}
		\frac{\partial^{|B|}q}{\prod_{j\in B}\partial r_j}\leq p(\mathrm{x})r_{\mathrm{m}}^{d-|B|}+C_2 r_{\mathrm{m}}^{d-|B|+1},
		\end{aligned}
		\end{equation}
		which implies that 
		\begin{equation}\label{eq30}
		\begin{aligned}
		\prod_{B\in \pi}\frac{\partial^{|B|}q}{\prod_{j\in B}\partial r_j}\leq Mr_{\mathrm{m}}^{(|\pi|-1)d},
		\end{aligned}
		\end{equation}
		where $M=p^{*d}+1$ and $p^*=\sup\limits_{\mathrm{x}\in \mathcal{Q}}p(\mathrm{x})$. Therefore, for $|\pi|\leq k$ and $r_{\mathrm{m}}\leq a_N$ we can bound $f_{N,k}(r_1,...,r_d)$ as
		\begin{equation}\label{eq11}
		\begin{aligned}
		f_{N,k}(r_1,...,r_d)=
		&\sum_{\pi\in \Pi} 
		\left(\begin{matrix}
		N-1\\k
		\end{matrix}\right)
		\cdot \frac{d^{|\pi|}q^k}{(d q)^{|\pi|}}\cdot \prod_{B\in \pi}\frac{\partial^{|B|}q}{\prod_{j\in B}\partial r_j} \cdot (1-{p}_{r_{\mathrm{m}}})^{N-k-1}\\
		\leq&\sum_{\pi\in \Pi}\frac{(N-1)!}{(k-|\pi|)!(N-k-1)!}q^{k-|\pi|}(1-{p}_{r_{\mathrm{m}}})^{N-k-1}Mr_{\mathrm{m}}^{(|\pi|-1)d}\\
		\leq&\sum_{\pi\in \Pi} M\cdot N^k p_{r_{\mathrm{m}}}^{k-|\pi|}(1-{p}_{r_{\mathrm{m}}})^{N-k-1}r_{\mathrm{m}}^{(|\pi|-1)d}\\
		\leq&\sum_{\pi\in \Pi} CM\cdot N^{|\pi|}r_{\mathrm{m}}^{(|\pi|-1)d}\\
		\leq&\sum_{\pi\in \Pi} CM\bigg(\frac{2k\log N}{C_1}\bigg)^{|\pi|-1}N\\
		\leq& |\Pi|CM\bigg(\frac{2k\log N}{C_1}\bigg)^{k-1}N,
		\end{aligned}
		\end{equation}
		where the third inequality is due to the fact that $p^{k-|\pi|}(1-p)^{N-k-1}\leq CN^{-k+|\pi|}$ for $p\in [0,1]$.
		Combining Eq~\eqref{eq11} and Eq~\eqref{eq6}, we can bound the expectation in Eq~\eqref{eq7}  by 
		\begin{equation}
		\begin{aligned}
		\bigg|\underset{P:\epsilon_k< a_N}{\mathbb{E}}\big[\log\big({\zeta_k^{\mathrm{x}_1}\cdots \zeta_k^{\mathrm{x}_d}}\big)\big]\bigg|\leq C_4  \frac{\big(\log N\big)^{k+1}}{C_1^k}
		\end{aligned}
		\end{equation}
		for some $C_4>0$.
		It follows that the first term of $I_1$ is bounded by 
		\begin{equation}
		\begin{aligned}
		\bigg|\underset{\mathrm{x}\in \mathcal{Q}_2}{\mathbb{E}}\underset{P:\epsilon_k< a_N}{\mathbb{E}}\big[\log\big({\zeta_k^{\mathrm{x}_1}\cdots \zeta_k^{\mathrm{x}_d}}\big)\big]\bigg| 
		&\leq C_4  \frac{\big(\log N\big)^{k+1}}{C_1^k}\underset{\mathrm{x}\in \mathcal{Q}_2}{\mathbb{E}}[1]\\
		&\leq  C_4  \frac{\big(\log N\big)^{k+1}}{C_1^k}p^*\mu(x\in \mathcal{Q}_2)\\
		&\leq p^*C_4  \frac{\big(\log N\big)^{k+1}}{C_1^k}(d+1)a_N\\
		&= (d+1)p^*C_4 \frac{\big(\log N\big)^{k+1}}{C_1^k}\big(\frac{2k\log N}{C_1N}\big)^{\frac{1}{d}}.
		\end{aligned}
		\end{equation}
		Since $\widetilde{P}$ is a sepcial case of $P$, the second term of $I_1$ can also be bounded by the same order. Thus, $I_1$ is bounded by 
		\begin{equation}
		|I_1|\leq C_5\frac{\big(\log N\big)^{k+2}}{C_1^{k+1}N^{\frac{1}{d}}},
		\end{equation}
		for some $C_5>0$.
		
		\textit{Bound of $I_2$}:
		
		For any $\mathrm{x} \in \mathcal{Q}_1$ and $\epsilon_k^{\mathrm{x}_j}<a_N, j=1,...,d$, it is easy to see that $\zeta_k^{\mathrm{x}_j}=\epsilon_k^{\mathrm{x}_j}$. Thus, $I_2$ can be bounded and rewritten as 
		\begin{equation}
		\begin{aligned}
		I_2&\leq\underset{\mathrm{x}\in \mathcal{Q}_1}{\mathbb{E}}\bigg|\underset{P:\epsilon_k<a_N}{\mathbb{E}}\big[\log\big({\zeta_k^{\mathrm{x}_1}\cdots \zeta_k^{\mathrm{x}_d}}\big)\big]
		-
		\underset{\widetilde{P}:\epsilon_k<a_N}{\mathbb{E}}\big[\log\big({\epsilon_k^{\mathrm{x}_1}\cdots\epsilon_k^{\mathrm{x}_d}}\big)\big]\bigg|\\
		&=
		\underset{\mathrm{x}\in \mathcal{Q}_1}{\mathbb{E}}\bigg| \int_{0}^{a_N}\cdots\int_{0}^{a_N}
		\big({f}_{N,k}(r_1,...,r_d)-\widetilde{f}_{N,k}(r_1,...,r_d)\big)\log \big(r_1\cdots r_d\big) dr_1\cdots dr_d \bigg|,
		\end{aligned}
		\end{equation}
		where $\widetilde{f}_{N,k}(r_1,...,r_d)=\left(\begin{matrix}
		N-1\\k
		\end{matrix}\right)\frac{\partial^d(\widetilde{q}^k_{r1,...,r_d})}{\partial r_1\cdots\partial r_d} \cdot (1-\widetilde{p}_{r_{\mathrm{m}}})^{N-k-1}$. Again, we omit the subscripts of $\widetilde{q}_{r_1,...,r_d}$  in the following analysis. Since we have
		\begin{equation}\label{eq13}
		\begin{aligned}
		&\int_{0}^{a_N}\cdots\int_{0}^{a_N}
		\big|\log \big({r_1}\cdots{r_d}\big)\big|dr_1\cdots dr_d	
		\\
		\leq&C_3\frac{\big(\log N\big)^{2}}{C_1 N},
		\end{aligned}
		\end{equation}
		from \eqref{eq6}, we now focus on bounding 
		${f}_{N,k}(r_1,...,r_d)-\widetilde{f}_{N,k}(r_1,...,r_d)$.
		Recall the Faà di Bruno's formula in Eq~\eqref{eq4}, and we have
		\begin{equation}\label{eq10}
		\hspace{-0cm}
		\begin{aligned}
		&{f}_{N,k}(r_1,...,r_d)\\
		=&\sum_{\pi\in \Pi}
		\left(\begin{matrix}
		N-1\\k
		\end{matrix}\right)
		\frac{\partial^{|\pi|}q^k}{(\partial q)^{|\pi|}} \prod_{B\in \pi}\frac{\partial^{|B|}q}{\prod_{j\in B}\partial r_j}  (1-{p}_{r_{\mathrm{m}}})^{N-k-1}\\
		=&\sum_{\pi\in \Pi}
		\left(\begin{matrix}
		N-1\\k
		\end{matrix}\right)\frac{k!}{(k-|\pi|)!}\big(p(\mathrm{x})r_1\cdots r_d+O(r_1\cdots r_dr_{\mathrm{m}})\big)^{k-|\pi|}
		\\&~~~~~~~~~~~~~~\times
		\prod_{B\in \pi}\big(p(\mathrm{x})\prod_{j\in \widehat{B}}r_j+O(r_{\mathrm{m}}\prod_{j\in \widehat{B}}r_j)\big) \big(1-p(\mathrm{x})r_{\mathrm{m}}^d-O(r_{\mathrm{m}}^{d+1})\big)^{N-k-1}\\
		=&\sum_{\pi\in \Pi}
		\left(\begin{matrix}
		N-1\\k
		\end{matrix}\right)\frac{k!}{(k-|\pi|)!}\big(p(\mathrm{x})r_1\cdots r_d\big)^{k-|\pi|}\big(1+O(r_{\mathrm{m}})\big)^{k-|\pi|} \prod_{B\in \pi}\big(p(\mathrm{x})\prod_{j\in \widehat{B}}r_j\big)\\
		&~~~~~~~~~~~~~~~~~~~~~~~~~~\times\big(1+O(r_{\mathrm{m}})\big) \big(1-p(\mathrm{x})r_{\mathrm{m}}^d\big)^{N-k-1}\big(1-O(r_{\mathrm{m}}^{d+1})\big)^{N-k-1}\\
		=&\sum_{\pi\in \Pi}
		\left(\begin{matrix}
		N-1\\k
		\end{matrix}\right)\frac{k!}{(k-|\pi|)!}\big(p(\mathrm{x})r_1\cdots r_d\big)^{k-|\pi|}\cdot \prod_{B\in \pi}\big(p(\mathrm{x})\prod_{j\in \widehat{B}}r_j\big)\\
		&~~~~~~~~~~~~~~~~~~~~~~~\times \big(1-p(\mathrm{x})r_{\mathrm{m}}^d\big)^{N-k-1}\cdot\big(1+O(r_{\mathrm{m}})\big)^{k}\big(1-O(r_{\mathrm{m}}^{d+1})\big)^{N-k-1}\\
		=&\sum_{\pi\in \Pi}
		\left(\begin{matrix}
		N-1\\k
		\end{matrix}\right)\frac{\partial^{|\pi|}\widetilde{q}^k}{(\partial \widetilde{q})^{|\pi|}}\cdot \prod_{B\in \pi}\frac{\partial^{|B|}\widetilde{q}}{\prod_{j\in B}\partial r_j} \cdot (1-\widetilde{p}_{r_{\mathrm{m}}})^{N-k-1}\cdot\big(1+O(r_{\mathrm{m}})\big)^{k}\big(1-O(r_{\mathrm{m}}^{d+1})\big)^{N-k-1}\\
		=&\widetilde{f}_{N,k}(r_1,...,r_d)\cdot\big(1+O(r_{\mathrm{m}})\big)^{k}\big(1-O(r_{\mathrm{m}}^{d+1})\big)^{N-k-1}
		\end{aligned}
		\end{equation}
		where the second equality is from Lemma~\ref{lemma5} and Lemma~\ref{lemma6} and the fifth equality is from the fact that $\widetilde{q}=p(\mathrm{x})r_1\cdots r_d$ and $\widetilde{p}_{r_{\mathrm{m}}}=p(\mathrm{x})r_{\mathrm{m}}^d$ for $\mathrm{x}\in \mathcal{Q}_1$ and $r_{\mathrm{m}}\leq a_N$.
		
		By Eq~\eqref{eq10}, we obtain the bound of the difference ${f}_{N,k}(r_1,...,r_d)-\widetilde{f}_{N,k}(r_1,...,r_d)$
		\begin{equation}\label{eq12}
		\begin{aligned}
		&|{f}_{N,k}(r_1,...,r_d)-\widetilde{f}_{N,k}(r_1,...,r_d)|\\
		=&\bigg|\big(1+O(r_{\mathrm{m}})\big)^{k}\big(1-O(r_{\mathrm{m}}^{d+1})\big)^{N-k-1}-1\bigg|\widetilde{f}_{N,k}(r_1,...,r_d)\\
		\leq & C_6 r_{\mathrm{m}}\widetilde{f}_{N,k}(r_1,...,r_d)\\
		\leq & C_6\bigg(\frac{2k\log N}{C_1N}\bigg)^{\frac{1}{d}}|\Pi|CM\bigg(\frac{2k\log N}{C_1}\bigg)^{k-1}N,
		\end{aligned}
		\end{equation}
		for some $C_6>0$, where the last inequality is from Eq~\eqref{eq11} and the fact that $\widetilde{P}$ is a special case of $P$. Combining Eq~\eqref{eq12} and Eq~\eqref{eq13}, we obtain the bound of $I_2$
		\begin{equation}
		\begin{aligned}
		I_2&\leq C_3C_6\bigg(\frac{2k\log N}{C_1N}\bigg)^{\frac{1}{d}}|\Pi|CM\bigg(\frac{2k\log N}{C_1}\bigg)^{k-1} \frac{\big(\log N\big)^{2}}{C_1} \underset{\mathrm{x}\in \mathcal{Q}_1}{\mathbb{E}}[1]\\
		&\leq C_7 \frac{(\log N)^{k+2}}{C_1^{k+1} N^{\frac{1}{d}}},
		\end{aligned}
		\end{equation}
		for some $C_7>0$, as $\underset{\mathrm{x}\in \mathcal{Q}_1}{\mathbb{E}}[1]\leq 1$.
		
		\textit{Bound of $I_3$}:
		
		To bound the first term of $I_3$, we need to bound $\underset{P:\epsilon_k\geq a_N}{\mathbb{E}}\big[\big|\log\big({\zeta_k^{\mathrm{x}_1}\cdots \zeta_k^{\mathrm{x}_d}}\big)\big|\big]$ first.
		Note that the event $\{\epsilon_k\geq a_N\}$ is equivalent to that there is at least one $j\in\{1,...,d\}$ such that $\epsilon_k^{\mathrm{x}_j}\geq a_N$, and by the symmetry of the equation, the expectation over this set can be rewritten as
		\begin{equation}\label{eq28}
		\begin{aligned}
		\underset{P:\epsilon_k\geq a_N}{\mathbb{E}}\big[\big|\log\big({\zeta_k^{\mathrm{x}_1}\cdots \zeta_k^{\mathrm{x}_d}}\big)\big|\big]=\sum_{i=1}^{d}C_d^i \underset{P:\bigg\{\begin{matrix}
			\epsilon_{k,1:i}\geq a_N\\\epsilon_{k,i:d}< a_N
			\end{matrix}}{\mathbb{E}}\big[\big|\log\big({\zeta_k^{\mathrm{x}_1}\cdots \zeta_k^{\mathrm{x}_d}}\big)\big|\big].
		\end{aligned}
		\end{equation}
		Consider each term in Eq~\eqref{eq28}
		\begin{equation}\label{eq29}
		\begin{aligned}
		&\underset{P:\bigg\{\begin{matrix}
			\epsilon_{k,1:i}\geq a_N\\\epsilon_{k,i:d}< a_N
			\end{matrix}}{\mathbb{E}}\big[\big|\log\big({\zeta_k^{\mathrm{x}_1}\cdots \zeta_k^{\mathrm{x}_d}}\big)\big|\big]\\
		\leq&\underset{P:\bigg\{\begin{matrix}
			\epsilon_{k,1:i}\geq a_N\\\epsilon_{k,i:d}< a_N
			\end{matrix}}{\mathbb{E}}\big[\big|\log\big({\zeta_k^{\mathrm{x}_1}\cdots \zeta_k^{\mathrm{x}_i}}\big)\big|\big]+\underset{P:\bigg\{\begin{matrix}
			\epsilon_{k,1:i}\geq a_N\\\epsilon_{k,i:d}< a_N
			\end{matrix}}{\mathbb{E}}\big[\big|\log\big({\zeta_k^{\mathrm{x}_{i+1}}\cdots \zeta_k^{\mathrm{x}_d}}\big)\big|\big].
		\end{aligned}
		\end{equation}
		For $\epsilon_k^{\mathrm{x}_j}\geq a_N , j=1,...,i$ and sufficiently large $N$, we have
		\begin{equation}\label{eq3}
		\begin{aligned}
		\big|\log\big({\zeta_k^{\mathrm{x}_1}\cdots \zeta_k^{\mathrm{x}_i}}\big)\big|
		&\leq \big|\log\big({\epsilon_k^{\mathrm{x}_1}/2\cdots \epsilon_k^{\mathrm{x}_i}}/2\big)\big|\\
		&\leq \big|\log \big(\frac{a_N}{2}\big)^i\big|\\
		&\leq C_8 \log N,
		\end{aligned}
		\end{equation}
		for some $C_8>0$. Using Lemma~\ref{lemma3} and Eq~\eqref{eq3}, the first term of Eq~\eqref{eq29} can be bounded by
		\begin{equation}\label{eq33}
		\begin{aligned}
		\underset{P:\bigg\{\begin{matrix}
			\epsilon_{k,1:i}\geq a_N\\\epsilon_{k,i:d}< a_N
			\end{matrix}}{\mathbb{E}}\big[\big|\log\big({\zeta_k^{\mathrm{x}_1}\cdots \zeta_k^{\mathrm{x}_i}}\big)\big|\big]
		&\leq C_8\log N\cdot \mathbb{P}\{\epsilon_{k,1:i}\geq a_N,\epsilon_{k,i:d}< a_N\}\\
		&\leq C_8\log N\cdot P\{\epsilon_k\geq a_N\}\\
		&\leq C_9\frac{(\log N)^{k+1}}{N^{2k}},
		\end{aligned}
		\end{equation}
		For some $C_9>0$.
		
		Now consider the second term of Eq~\eqref{eq29}. 
		Like Eq~\eqref{eq6}, the integration with respect to Lebesgue measure can be bounded as 
		\begin{equation}\label{eq32}
		\begin{aligned}
		&\int_{a_N}^{1}\cdots\int_{a_N}^{1}\bigg(\int_{0}^{a_N}\cdots\int_{0}^{a_N}\big|\log\big({\zeta_k^{\mathrm{x}_{i+1}}\cdots \zeta_k^{\mathrm{x}_d}}\big)\big|dr_{i+1}\cdots dr_d\bigg)dr_d\cdots dr_{i}\\
		\leq& -(d-i)(a_N)^{d-i-1}\int_{0}^{a_N}\log r dr +(d-i)\log 2(\int_{0}^{a_N}dr)^{d-i}\\
		\leq&C_{10}\log N,
		\end{aligned}
		\end{equation}
		for some $C_{10}>0$.
		Again using the multivariate version of Faà di Bruno's formula, we can bound $f_{N,k}(r_1,...,r_d)$ for $|\pi|\leq k$ and $r_{\mathrm{m}}\geq a_N$  as
		\begin{equation}\label{eq31}
		\begin{aligned}
		f_{N,k}(r_1,...,r_d)=
		&\sum_{\pi\in \Pi} 
		\left(\begin{matrix}
		N-1\\k
		\end{matrix}\right)
		\cdot \frac{d^{|\pi|}q^k}{(d q)^{|\pi|}}\cdot \prod_{B\in \pi}\frac{\partial^{|B|}q}{\prod_{j\in B}\partial r_j} \cdot (1-{p}_{r_{\mathrm{m}}})^{N-k-1}\\
		\leq&\sum_{\pi\in \Pi}\frac{(N-1)!}{(k-|\pi|)!(N-k-1)!}q^{k-|\pi|}(1-{p}_{r_{\mathrm{m}}})^{N-k-1}Mr_{\mathrm{m}}^{(|\pi|-1)d}\\
		\leq&\sum_{\pi\in \Pi} \frac{(N-1)!}{(k-|\pi|)!(N-k-1)!}(1-C_1 a_N^d)^{N-k-1}M\\
		\leq& C_{11}\frac{1}{N^k},
		\end{aligned}
		\end{equation}
		for some $C_{11}>0$. Therefore, combining Eq~\eqref{eq32} and Eq~\eqref{eq31} leads to the bound of the second term of Eq~\eqref{eq29}
		\begin{equation}
		\underset{P:\bigg\{\begin{matrix}
			\epsilon_{k,1:i}\geq a_N\\\epsilon_{k,i:d}< a_N
			\end{matrix}}{\mathbb{E}}\big[\big|\log\big({\zeta_k^{\mathrm{x}_{i+1}}\cdots \zeta_k^{\mathrm{x}_d}}\big)\big|\big]
		\leq C_{10}C_{11}\frac{\log N}{N^k},
		\end{equation}
		which is a larger bound then Eq~\eqref{eq33}.
		As a result we can bound Eq~\eqref{eq29} by
		\begin{equation}\label{eq34}
		\underset{P:\bigg\{\begin{matrix}
			\epsilon_{k,1:i}\geq a_N\\\epsilon_{k,i:d}< a_N
			\end{matrix}}{\mathbb{E}}\big[\big|\log\big({\zeta_k^{\mathrm{x}_1}\cdots \zeta_k^{\mathrm{x}_d}}\big)\big|\big]\leq C_{10} C_{11}\frac{\log N}{N^k}.
		\end{equation}
		Given Eq~\eqref{eq34}, we are now able to estimate Eq~\eqref{eq28} and then the first term of $I_3$ by the same bound up to a constant. Similarly, we can also bound the second term of $I_3$ by $O\big(\frac{\log N}{N^k}\big)$. Thus, $I_3$ can be bounded by
		\begin{equation}
		I_3\leq C_{12}\frac{\log N}{N^k},
		\end{equation}
		for some $C_{12}>0$.
		
		Finally, combining the upper bounds of $I_1$, $I_2$ and $I_3$, we obtain that the bias is bounded by
		\begin{equation}
		\big|\mathbb{E}[\widehat{H}_{tKSG}(X)]-H(X)\big|\leq C_{13} \frac{(\log N)^{k+2}}{C_1^{k+1} N^{\frac{1}{d}}},
		\end{equation}
		for some $C_{13}>0$.
	\end{proof}

	\subsection{Proof of variance bound for the truncated KSG estimator}\label{varksg}
	
	\begin{proof}
		We let $\beta_i=\sum_{j=1}^{d}\log\zeta_{i,j}$, and define $\beta'_i, i=1,...,N$ as the estimators after $\mathrm{x}^{(1)}$ is resampled and $\beta^*_i, i=2,...,N$ as the estimators after $\mathrm{x}^{(1)}$ is removed. It should be noted that this proof can be completed by following the roadmap in \ref{vartkl}, and the only issue that needs to be validated here is that $\mathbb{E}[\beta_1^2]=O({(\log N)^{k+2}})$.
		
		Again, we separate $\mathbb{E}\big[\beta_1^2\big]$ into two parts,
		\begin{equation}\label{eq23}
		\begin{aligned}
		\mathbb{E}\big[\beta_1^2\big]=\underset{\mathrm{x}\in \mathcal{Q}}{\mathbb{E}}\underset{P:\epsilon_k< a_N}{\mathbb{E}}\big[\beta_1^2\big]
		+
		\underset{\mathrm{x}\in \mathcal{Q}}{\mathbb{E}}\underset{P:\epsilon_k\geq a_N}{\mathbb{E}}\big[\beta_1^2\big],
		\end{aligned}
		\end{equation}
		where $a_N$ is defined as in \ref{biasksg}.

		First, we consider the bound of the first term in Eq~\eqref{eq23}.
		For any $\mathrm{x}\in \mathcal{Q}$,
		\begin{equation}\label{eq24}
		\begin{aligned}
		&\underset{P:\epsilon_k< a_N}{\mathbb{E}}\big[\beta_1^2\big]\\
		=&\int_{0}^{a_N}\cdots\int_{0}^{a_N}
		f_{N,k}(r_1,...,r_d)\big[\log \big(\zeta_k^{\mathrm{x}_1}\cdots\zeta_k^{\mathrm{x}_d}\big)\big]^2 dr_1\cdots dr_d,\\
		\end{aligned}
		\end{equation}
		where $f_{N,k}(r_1,...,r_d)=\left(\begin{matrix}
		N-1\\k
		\end{matrix}\right)
		\cdot \frac{\partial^d({q}_{r_1,...,r_d}^k)}{\partial r_1\cdots\partial r_d} \cdot (1-{p}_{r_{\mathrm{m}}})^{N-k-1}$, and $r_{\mathrm{m}}=\max\limits_{1\leq j\leq d}r_j$ \cite{kraskov2004estimating}.
		
		Note that for sufficiently large $N$, we have,
		\begin{equation}\label{eq25}
		\begin{aligned}
		&\int_{0}^{a_N}\cdots\int_{0}^{a_N}
		\big[\log \big(\zeta_k^{\mathrm{x}_1}\cdots\zeta_k^{\mathrm{x}_d}\big)\big]^2 dr_1\cdots dr_d\\
		\leq&\int_{0}^{a_N}\cdots
		\int_{0}^{a_N}
		\big[	\log \big(\frac{r_1}{2}\cdots\frac{r_d}{2}\big)\big]^2
		dr_1\cdots dr_d\\
		=& d\int_{0}^{a_N}\cdots
		\int_{0}^{a_N}
		\big[	\log \big(\frac{r_1}{2}\big)\big]^2
		dr_1\cdots dr_d+d(d-1)\int_{0}^{a_N}\cdots
		\int_{0}^{a_N}
		\log \big(\frac{r_1}{2}\big)\log \big(\frac{r_2}{2}\big)
		dr_1\cdots dr_d\\
		\leq & C_3\frac{(\log N)^3}{N},
		\end{aligned}
		\end{equation}
		for some $C_3>0$.
		Recall Eq~\eqref{eq11}, and we can bound Eq~\eqref{eq24} as:
		\begin{equation}
		\underset{P:\epsilon_k< a_N}{\mathbb{E}}\big[\beta_1^2\big]\leq C_4{(\log N)^{k+2}},
		\end{equation}
		for some $C_4>0$.
		Thus, the first term in Eq~\eqref{eq23} is bounded by 
		\begin{equation}\label{eq26}
		\underset{\mathrm{x}\in \mathcal{Q}}{\mathbb{E}}\underset{P:\epsilon_k< a_N}{\mathbb{E}}\big[\beta_1^2\big]
		\leq C_4{(\log N)^{k+2}}.
		\end{equation}

		Now we consider the second term in Eq~\eqref{eq23}. 
		
		Like the bound analysis of $I_3$ in \ref{biasksg}, we can rewrite $\underset{P:\epsilon_k\geq a_N}{\mathbb{E}}\big[\beta_1^2\big]$ as
		\begin{equation}\label{eq37}
		\begin{aligned}
		\underset{P:\epsilon_k\geq a_N}{\mathbb{E}}\big[\beta_1^2\big]=\sum_{i=1}^{d}C_d^i \underset{P:\bigg\{\begin{matrix}
			\epsilon_{k,1:i}\geq a_N\\\epsilon_{k,i:d}< a_N
			\end{matrix}}{\mathbb{E}}\big[\beta_1^2\big].
		\end{aligned}
		\end{equation}
		Consider each term of Eq~\eqref{eq28}
		\begin{equation}\label{eq38}
		\begin{aligned}
		&\underset{P:\bigg\{\begin{matrix}
			\epsilon_{k,1:i}\geq a_N\\\epsilon_{k,i:d}< a_N
			\end{matrix}}{\mathbb{E}}\big[\beta_1^2\big]\\
		\leq&2\bigg(\underset{P:\bigg\{\begin{matrix}
			\epsilon_{k,1:i}\geq a_N\\\epsilon_{k,i:d}< a_N
			\end{matrix}}{\mathbb{E}}\big[\big|\log\big({\zeta_k^{\mathrm{x}_1}\cdots \zeta_k^{\mathrm{x}_i}}\big)\big|^2\big]+\underset{P:\bigg\{\begin{matrix}
			\epsilon_{k,1:i}\geq a_N\\\epsilon_{k,i:d}< a_N
			\end{matrix}}{\mathbb{E}}\big[\big|\log\big({\zeta_k^{\mathrm{x}_{i+1}}\cdots \zeta_k^{\mathrm{x}_d}}\big)\big|^2\big]\bigg)\\
		\end{aligned}
		\end{equation}
		For $\epsilon_k^{\mathrm{x}_j}\geq a_N , j=1,...,i$ and sufficiently large $N$, we have
		\begin{equation}\label{eq39}
		\begin{aligned}
		\big|\log\big({\zeta_k^{\mathrm{x}_1}\cdots \zeta_k^{\mathrm{x}_i}}\big)\big|^2
		&\leq \big|\log\big({\epsilon_k^{\mathrm{x}_1}/2\cdots \epsilon_k^{\mathrm{x}_i}}/2\big)\big|^2\\
		&\leq \big|\log \big(\frac{a_N}{2}\big)^i\big|^2\\
		&\leq C_5 (\log N)^2,
		\end{aligned}
		\end{equation}
		for some $C_5>0$. Using Lemma~\ref{lemma3} and Eq~\eqref{eq39}, the first term of Eq~\eqref{eq38} can be bounded by
		\begin{equation}\label{eq40}
		\begin{aligned}
		\underset{P:\bigg\{\begin{matrix}
			\epsilon_{k,1:i}\geq a_N\\\epsilon_{k,i:d}< a_N
			\end{matrix}}{\mathbb{E}}\big[\big|\log\big({\zeta_k^{\mathrm{x}_1}\cdots \zeta_k^{\mathrm{x}_i}}\big)\big|^2\big]
		&\leq C_5(\log N)^2\cdot \mathbb{P}\{\epsilon_{k,1:i}\geq a_N,\epsilon_{k,i:d}< a_N\}\\
		&\leq C_5(\log N)^2\cdot P\{\epsilon_k\geq a_N\}\\
		&\leq C_6,
		\end{aligned}
		\end{equation}
		for some $C_6>0$.
		
		Now consider the second term of Eq~\eqref{eq38}. 
		Like Eq~\eqref{eq25}, the integration with respect to Lebesgue measure is bounded as 
		\begin{equation}\label{eq41}
		\begin{aligned}
		&\int_{a_N}^{1}\cdots\int_{a_N}^{1}\bigg(\int_{0}^{a_N}\cdots\int_{0}^{a_N}\big|\log\big({\zeta_k^{\mathrm{x}_{i+1}}\cdots \zeta_k^{\mathrm{x}_d}}\big)\big|^2dr_{i+1}\cdots dr_d\bigg)dr_d\cdots dr_{i}\\
		\leq&C_{7},
		\end{aligned}
		\end{equation}
		for some $C_{7}>0$. Therefore, combining Eq~\eqref{eq41} and the PDF bound in Eq~\eqref{eq31} leads to the bound of the second term of Eq~\eqref{eq38}
		\begin{equation}
		\underset{P:\bigg\{\begin{matrix}
			\epsilon_{k,1:i}\geq a_N\\\epsilon_{k,i:d}< a_N
			\end{matrix}}{\mathbb{E}}\big[\big|\log\big({\zeta_k^{\mathrm{x}_{i+1}}\cdots \zeta_k^{\mathrm{x}_d}}\big)\big|^2\big]
		\leq C_8,
		\end{equation}
		for some $C_8>0$.
		As a result we can bound Eq~\eqref{eq38} by
		\begin{equation}\label{eq42}
		\underset{P:\bigg\{\begin{matrix}
			\epsilon_{k,1:i}\geq a_N\\\epsilon_{k,i:d}< a_N
			\end{matrix}}{\mathbb{E}}\big[\big|\log\big({\zeta_k^{\mathrm{x}_1}\cdots \zeta_k^{\mathrm{x}_d}}\big)\big|\big]\leq C_6+C_8.
		\end{equation}
		Given Eq~\eqref{eq42}, we are now able to estimate Eq~\eqref{eq37} and then the second term of Eq~\eqref{eq23} by the same bound up to a constant.

		Finally, the expectation of $\beta_1^2$ is bounded as
		\begin{equation}
		\mathbb{E}[\beta_1^2]\leq C_9 (\log N)^{k+2},
		\end{equation}
		for some $C_9>0$. Following the same procedure in \ref{vartkl}, we can obtain the bound of the variance of $\widehat{H}_{tKSG}(X)$
		\begin{equation}
		\mathrm{Var}[\widehat{H}_{tKSG}(X)]\leq C_{10}\frac{(\log N)^{k+2}}{N},
		\end{equation}
		for some $C_{10}>0$.
	\end{proof}

	\section{Proof of Corollary~2}\label{sec:cor2}
	\begin{proof}
		Given a UM $f$, the density of the original distribution satisfies the change of variable formula,
		\begin{equation}
			p_\mathrm{x}(\mathrm{x})=p_\mathrm{z}(f(\mathrm{x}))g(\mathrm{x}),
		\end{equation}
		where $g(\mathrm{x})=\left|\operatorname{det} \frac{\partial f(\mathrm{x})}{\partial \mathrm{x}}\right|$  is differentiable and positive for any $\mathrm{x}\in \mathbb{R}^d$ (\cite{papamakarios2017masked, Dinh2017density}). 
		Recall that $p_\mathrm{x}$ is differentiable, and it follows that, 
		\begin{equation}
			p_\mathrm{z}(\mathrm{z})=\frac{p_\mathrm{x}(f^{-1}(\mathrm{z}))}{g(f^{-1}(\mathrm{z})))},
		\end{equation}
		is also differentiable for any $\mathrm{z}\in Q^o$. Thus, the supreme  $C_2^N$ is a well defined random variable. 
		
		Since $p_{\mathrm{z}}^S$ is a differentiable density function defined on $\mathcal{Q}$, there exists a $\mathrm{z}^*\in \mathcal{Q}$ such that $p_{\mathrm{z}}^S(\mathrm{z^*})=1$. By mean value theorem, we have
		\begin{equation}
		\begin{aligned}
		|&1-p_{\mathrm{z}}^S(\mathrm{z})|\\ \leq
		& |\triangledown p_\mathrm{z}^S(\mathrm{\xi})\cdot(\mathrm{z^*}-\mathrm{z})|\\ \leq
		& ||\triangledown p_\mathrm{z}^S(\mathrm{\xi})||_1 \cdot ||\mathrm{z^*}-\mathrm{z}||_\infty\\ \leq
		& C_2^N,
		\end{aligned}
		\end{equation}
		where $\xi$ is some vector in $\mathcal{Q}$. Thus, we have
		\begin{equation}
		1-C_2^N\leq p_\mathrm{x}^N(\mathrm{x}) \leq 1+C_2^N.
		\end{equation}
		Now define $C_1^N = \inf\limits_{\mathrm{z}\in\mathcal{Q}}p_\mathrm{z}^S(\mathrm{z})$. For $N>M$, the bias can then be bounded by
		\begin{equation}
		\begin{aligned}
		& \big|\mathbb{E}[\widehat{H}_\mathrm{UM-tKL}(X)]-H(X)\big|\\ 
		\leq& \mathbb{E}_{UM}\big|\mathbb{E}_{X}[\widehat{H}_\mathrm{UM-tKL}(X)]-H(X)\big|
		\\ \leq
		& \mathbb{E}\big[\frac{C_2^N}{(C_1^ N)^{1+1/d}}\big]\big(\frac{k}{ N}\big)^\frac{1}{d}\\ \leq
		& C_{UM-tKL}^N \big(\frac{k}{ N}\big)^\frac{1}{d},
		\end{aligned}
		\end{equation}
		where $C^N_{UM-tKL}=\frac{1}{(1-\bar{C})^{1+1/d}}\mathbb{E}[C_2^N]$. Note that $C_2^N \underset{N \rightarrow \infty}{\stackrel{\mathbb{P}}{\longrightarrow}} 0$ and $C_2^N \leq \bar{C}, \,a.s.$ for any $N>M$, we have $\lim\limits_{N\rightarrow \infty}\mathbb{E}[C_2^N]=0$ and therefore $\lim\limits_{N\rightarrow \infty}C^N_{UM-tKL}=0$.
		The MSE can be bounded by
		\begin{equation}\label{eq50}
		    \begin{aligned}
		    & \mathbb{E}[(\widehat{H}_\mathrm{UM-tKL}(X)-H(X))^2]\\
		    \leq & 2\mathbb{E}[(\widehat{H}_\mathrm{UM-tKL}(X)-\mathbb{E}_{X}[\widehat{H}_\mathrm{UM-tKL}(X)])^2]+2\mathbb{E}[(\mathbb{E}_{X}[\widehat{H}_\mathrm{UM-tKL}(X)]-H(X))^2]\\
		    =&2\mathbb{E}_{\mathrm{UM}}\mathbb{E}_{X}[(\widehat{H}_\mathrm{UM-tKL}(X)-\mathbb{E}_{X}[\widehat{H}_\mathrm{UM-tKL}(X)])^2]+2\mathbb{E}_{\mathrm{UM}}[(\mathbb{E}_{X}[\widehat{H}_\mathrm{UM-tKL}(X)]-H(X))^2]
		    \end{aligned}
		\end{equation}
		Note that when $N>M$, $C_1^N$ and $C_2^N$ satisfy Assumption~\ref{assumption1}. Then by Theorem~1, we can bound the first term of Eq.~\eqref{eq50} by
		\begin{equation}
		    \begin{aligned}
		    2\mathbb{E}_{\mathrm{UM}}\mathbb{E}_{X}[(\widehat{H}_\mathrm{UM-tKL}(X)-\mathbb{E}_{X}[\widehat{H}_\mathrm{UM-tKL}(X)])^2]
		    \leq  C_1 \frac{1}{N},
		    \end{aligned}
		\end{equation}
		for some $C_1>0$.
		The second term of Eq.~\eqref{eq50} can be bounded by
		\begin{equation}
		    \begin{aligned}
		    &2\mathbb{E}_{\mathrm{UM}}[(\mathbb{E}_{X}[\widehat{H}_\mathrm{UM-tKL}(X)]-H(X))^2]\\
		    \leq & 2\mathbb{E}\big[\frac{(C_2^N)^2}{(C_1^ N)^{2(1+1/d)}}\big]\big(\frac{k}{ N}\big)^\frac{2}{d}\\
		    \leq & D^N_{UM-tKL}\big(\frac{k}{ N}\big)^\frac{2}{d}
		    \end{aligned}
		\end{equation}
		where $D^N_{UM-tKL}=\frac{2}{(1-\bar{C})^{2(1+1/d)}}\mathbb{E}[(C_2^N)^2]$. Again, we have,$\lim\limits_{N\rightarrow \infty}D^N_{UM-tKL}=0$ for any $N>M$. Thus, the MSE is bounded by 
		\begin{equation}
		    \mathbb{E}[(\widehat{H}_\mathrm{UM-tKL}(X)-H(X))^2]\leq C_1\frac{1}{N}+D^N_{UM-tKL}\big(\frac{k}{ N}\big)^\frac{2}{d}.
		\end{equation}
	\end{proof}

	\section{Proof of Corollary~3}\label{sec:cor3}
	\begin{proof}
	For $N>M$, the bias can be bounded by
		\begin{equation}
		\begin{aligned}
		&\big|\mathbb{E}[\widehat{H}_\mathrm{UM-tKSG}(X)]-H(X)\big|\\ \leq
		& C\mathbb{E}\big[\frac{\bar{p}_\mathrm{z}^S\big((\bar{p}_\mathrm{z}^S)^d+1\big)}{C_1^{k+1}}\big] \frac{(\log N)^{k+2}}{ N^{\frac{1}{d}}} \\ \leq
		& C_{UM-tKSG}\frac{(\log N)^{k+2}}{ N^{\frac{1}{d}}},
		\end{aligned}
		\end{equation}
		where $C$ is a positive constant, $\bar{p}_\mathrm{z}^S=\sup\limits_{\mathrm{z}\in \mathcal{Q}}{p}_\mathrm{z}^S(\mathrm{z})$ and $C_{UM-tKSG}=C\frac{(1+\bar{C})\big((1+\bar{C})^d+1\big)}{(1-\bar{C})^{k+1}}$. Similarly as the proof of Corollary 2 and by Theorem 2, we can bound the MSE by
		\begin{equation}
		    \mathbb{E}[(\widehat{H}_\mathrm{UM-tKSG}(X)-H(X))^2]\leq C_2\frac{(\log N)^{k+2}}{N}+D^N_{UM-tKSG}\frac{(\log N)^{2(k+2)}}{ N^{\frac{2}{d}}},
		\end{equation}
		where $C_2$ is a positive constant and $D^N_{UM-tKSG}=\Big(C\frac{(1+\bar{C})\big((1+\bar{C})^d+1\big)}{(1-\bar{C})^{k+1}}\Big)^2$. 
	\end{proof}

	\section{Further details of the numerical examples}
	\subsection{Implementation details of the estimators}
	
	{\bf The setup of MAF:} We use a MAF built by 10 autoregressive layers \cite{germain2015made} for Hybrid Rosenbrock distribution and one built by 5 autoregressive layers for Even Rosenbrock distribution and the application of experimental design. Each layer has two hidden layers of 50 units and tanh nonlinearities. In each experiment, half of the samples  are used to train the MAF model and the other half are used to estimate the entropy.
	
	{\bf The implementation of CADEE and non-Mises estimator:} The two estimators are implemented using the code provided by \cite{ariel2020estimating} and \cite{kandasamy2015nonparametric} with the default parameters.
	
	\subsection{The two multivariate Rosenbrock distributions} \label{sec:hr}
	\textbf{Hybrid Rosenbrock Distribution.}
	The density of the hybrid Rosenbrock distribution is given by
	\begin{equation}
	\pi(\mathbf{x}) \propto \exp \left\{-a(x_{1}-\mu)^{2}-\sum_{j=1}^{n_{2}} \sum_{i=2}^{n_{1}} b_{j, i}(x_{j, i}-x_{j, i-1}^{2})^{2}\right\},\label{e:hr}
	\end{equation}
	where the dimensionality of $\bf{x}$ is $d=(n_1-1)n_2+1$. The variable $x_{j,1}=x_1$ for $j=1,...,n_2$.
	The normalization constant of Eq.~\eqref{e:hr}  is 
	\begin{equation}
	\frac{\sqrt{a} \prod_{i=2, j=1}^{n_{1}, n_{2}} \sqrt{b_{j, i}}}{\pi^{d / 2}}.
	\end{equation}
	
	In this experiment, we set $\mu=1.0$, $a=1.0$, $b_{j,i}=0.1$ for all $i$ and $j$, $n_1=4$ and $n_2$ ranging from 1 to 7. This setting forms a class of distributions with dimensions ranging from 4 to 22. 
	
	\textbf{Even Rosenbrock Distribution.} 
	The density of the even Rosenbrock distribution is given by
	\begin{equation}
	\pi(\mathbf{x}) \propto \exp \left\{-\sum_{i=1}^{d / 2}\left[\left(x_{2 i-1}-\mu_{2 i-1}\right)^{2}-c_i\left(x_{2 i}-x_{2 i-1}^{2}\right)^{2}\right] \right\}, \label{e:er}
	\end{equation}
	where the dimensionality $d$ must be an even number. The normalization constant for Eq.~\eqref{e:er} is 
	\begin{equation}
	\frac{\prod_{i=1}^{d/2}\sqrt{c_i}}{\pi^{d/2}}.
	\end{equation}
	In this experiment, we set $\mu_{2i-1}=0$, $c_i = 12.5$ for $i=1,...,d/2$ with $d$ ranging from 2 to 22. This setting forms a class of distributions with dimensions ranging from 2 to 22.
	
	\textbf{Hybrid Rosenbrock Distribution with Discontinuous Density.}
	The density of the hybrid Rosenbrock distribution with discontinuous density is given by
	\begin{equation}
	    \pi(\mathbf{x}) = \mathrm{unifpdf}(x_1,\mu, \sqrt{\frac{1}{8a}})\times \prod_{j=1}^{n_{2}} \prod_{i=2}^{n_{1}} \mathrm{unifpdf}(x_{j,i}, x_{j,i-1}^2,\sqrt{\frac{1}{8b}})\label{e:hru}
	\end{equation}
	where $\mathrm{unifpdf}(x,\alpha,\beta)$ is the pdf of the continuous uniform distribution on the interval $[\alpha-\beta,\alpha+\beta]$, evaluated at the values in $x$, and where the dimensionality of $\bf{x}$ is $d=(n_1-1)n_2+1$. The variable $x_{j,1}=x_1$ for $j=1,...,n_2$.
	
	In this experiment, we set $\mu=1.0$, $a=1.0$, $b_{j,i}=0.1$ for all $i$ and $j$, $n_1=4$ and $n_2$ ranging from 1 to 7. This setting forms a class of distributions with dimensions ranging from 4 to 22.
	
	\textbf{Even Rosenbrock Distribution with Discontinuous Density.} 
	The density of the even Rosenbrock distribution with discontinuous density is given by
	\begin{equation}
	    \pi(\mathbf{x}) =\prod_{i=1}^{d / 2}\left[\mathrm{unifpdf}(x_{2i-1},\mu_{2i-1},0.5)\times \mathrm{unifpdf}(x_{2i}, x_{2i-1}^2, c_i)\right], \label{e:eru}
	\end{equation}
	where the dimensionality $d$ must be an even number.
	
	In this experiment, we set $\mu_{2i-1}=0$, $c_i = 0.025$ for $i=1,...,d/2$ with $d$ ranging from 2 to 22. This setting forms a class of distributions with dimensions ranging from 2 to 22.
	
	\subsection{Entropy estimator only using NF}
	In this section we describe a simplified version of the proposed method, which estimate the entropy only using NF (without the truncated entropy estimators). 
	To start with, we recall Eq.~(12) in the main paper,
	\begin{equation}
	H(X) = H(Z)+\int p_\mathrm{z}(\mathrm{z})\log\bigg|\det\frac{\partial f^{-1}(\mathrm{z})}{\partial \mathrm{z}}\bigg|d\mathrm{z}.
	\end{equation}
	The main idea of this simplified method is to assume that  the transformed random variable $Z$ exactly follows 
	a uniform distribution and as a result $H(Z)=0$. 
	Therefore the entropy of $X$ is estimated as, 
	\begin{equation}
	\hat{H}_{NF}(X) = \frac{1}{n}\sum_{i=1}^{n}\log\bigg|\det\frac{\partial f^{-1}(\mathrm{z}^{(i)})}{\partial \mathrm{z}}\bigg|,
	\end{equation}
	where $\mathrm{z}^{(i)}=f(\mathrm{x}^{(i)})$. 
	A limitation of this method is quite obvious -- the transformed random variable $Z$ is usually not uniformly distributed and simply taking its entropy 
	to be zero will undoubtedly introduce bias, which is demonstrated by the numerical examples in the main paper. 
	It should also be noted that, while not in the context of entropy estimation, a NF based approach has been used for maximum entropy modeling \cite{DBLP:conf/iclr/Loaiza-GanemGC17}.
	
	\subsection{The Beta scheme for parametrizing the observation times}\label{sec:betascheme}
	
	In the optimal experimental design (OED) example, we use a lower dimensional
	parameterization scheme to reduce the dimensionality of the optimization problem~ 
	\cite{Ryan2014towards}. 
	In particular we use the Beta scheme \cite{Ryan2014towards} to allocate the placements of the observation times. Specifically, 
	let $Q(\cdot,\alpha,\beta)$ be the quantile function of the beta distribution with shape parameters $\alpha$ and $\beta$,
	and
	the $d$ observation times $\lambda=(t_1,...,t_d)$ in the time interval $[0, T]$ are allocated as, 
	\begin{equation}
	t_i = T\cdot Q( \frac{i}{d+1},\alpha, \beta), \quad i=1, ... ,d.
	\end{equation}
	As such the $d$-dimensional variable $\lambda$ is parametrized  by $\alpha>0$ and $\beta>0$.  
	
	\subsection{Nested Monte Carlo}\label{sec:nestedMC}
	Here we describe the Nested Monte Carlo (NMC) approach that is used to estimate the entropy in the experimental design example.
	Recall that the entropy of interest is $H(Y)$ (here for simplicity we omit the design parameter $\lambda$): 
	\begin{equation}
	H(Y) = \int \log p(y) p(y) dy,
	\end{equation} 
	which can be estimated via Monte Carlo (MC): 
	\begin{equation}
	H(Y)\approx -\frac{1}{M}\sum_{i=1}^M \log p(y^{(i)}), \label{e:outer}
	\end{equation}
	where $y^{(i)}$ are drawn from $p(y)$. 
	A difficulty here is that we do not have an explicit expression of $p(y)$.
	Note however that in this example the likelihood $p(y|\theta)$ and the prior $p(\theta)$ are available and we can therefore write
	\begin{equation}
	p(y) = \int p(y|\theta)p(\theta) d\theta.
	\end{equation}
	It follows that $p(y)$ can also be estimated via MC: 
	\begin{equation}
	p(y^{(i)})\approx\frac{1}{N}\sum_{j=1}^N p(y^{(i)}|\theta^{(j)}), \label{e:inner}
	\end{equation}
	where $\theta^{(j)}$ are drawn from $p(\theta)$.
	Combining Eq.~\eqref{e:inner} and Eq.~\eqref{e:outer}, we obtain an estimator of $H(Y)$,
	which is referred to as the NMC method~\cite{ryan2003estimating}.
	In particular, Eq.~\eqref{e:inner} is usually referred to as the inner MC  and Eq.~\eqref{e:outer} is referred to as the outer one. 
	Since the theoretical results in \cite{ryan2003estimating,rainforth2018nesting} show that the mean squared
	error of NMC estimator decays at a rate of $O(\frac{1}{M}+\frac{1}{N})$, we can obtain an accurate evaluation of $H(Y)$ with a sufficiently large number of samples, and in the numerical example we use $M=N=1\times10^5$.
	We emphasize that such a large number of samples is not computationally feasible to use in the experimental design procedure,
	and thus in the example we have to resort to other entropy estimation methods. 
\end{document}